\documentclass[11pt]{article}
\usepackage[perpage]{footmisc}

\usepackage[utf8]{inputenc}
\usepackage{amsthm,amsmath,amssymb}
\numberwithin{equation}{section}
\usepackage{amsfonts,mathrsfs}
\usepackage{braket}
\usepackage{caption,subcaption}
\usepackage{threeparttable}
\usepackage{booktabs}
\usepackage{multirow}
\usepackage{longtable}
\usepackage[dvipsnames]{xcolor}
\usepackage[linesnumbered,ruled,vlined]{algorithm2e}
\usepackage[letterpaper, margin=1in]{geometry}
\usepackage{xspace,makecell,tikz}
\usepackage{siunitx}
\usepackage{float}
\usepackage[sort,numbers]{natbib}
\usepackage{aliascnt}
\usepackage[pagebackref]{hyperref}
\usepackage[capitalise]{cleveref}
\usepackage{bbm}
\usepackage{enumerate}
\usepackage{tikz}
\usepackage[shortlabels]{enumitem}
\usepackage{tocloft}
\usepackage{comment}
\usetikzlibrary{arrows.meta, positioning,calc, shapes.geometric}
\hypersetup{
linktocpage=true,
colorlinks=true,
linkcolor=red,
filecolor=magenta,
urlcolor=cyan,
citecolor=blue,
}
\tikzset{  
      thickline/.style={line width=2pt},
}  
\usetikzlibrary{decorations.pathreplacing}
\usepackage[toc,page]{appendix}

\newcommand{\bool}[1]{\ifmmode\text{\textsc{#1}}\else\textsc{#1}\fi}

\newcommand{\poly}{\mathsf{poly}}

\newcommand{\OPT}{\mathsf{OPT}}

\newcommand{\eat}[1]{}

\newtheorem{theorem}{Theorem}[section]
\newaliascnt{lemma}{theorem}
\newtheorem{lemma}[lemma]{Lemma}
\aliascntresetthe{lemma}

\newaliascnt{claim}{theorem}

\aliascntresetthe{claim}
\newaliascnt{proposition}{theorem}

\aliascntresetthe{proposition}
\newaliascnt{corollary}{theorem}
\newtheorem{corollary}[corollary]{Corollary}
\aliascntresetthe{corollary}

\newaliascnt{remark}{theorem}

\aliascntresetthe{remark}
\newaliascnt{example}{theorem}

\aliascntresetthe{example}

\newaliascnt{conjecture}{theorem}
\newtheorem{conjecture}[conjecture]{Conjecture}
\aliascntresetthe{conjecture}
\newaliascnt{fact}{theorem}

\aliascntresetthe{fact}
\newaliascnt{property}{theorem}

\aliascntresetthe{property}
\crefname{lemma}{Lemma}{Lemmas}
\Crefname{lemma}{Lemma}{Lemmas}
\crefname{claim}{Claim}{Claims}
\Crefname{claim}{Claim}{Claims}
\crefname{proposition}{Proposition}{Propositions}
\Crefname{proposition}{Proposition}{Propositions}
\crefname{corollary}{Corollary}{Corollaries}
\Crefname{corollary}{Corollary}{Corollaries}
\crefname{remark}{Remark}{Remarks}
\Crefname{remark}{Remark}{Remarks}
\crefname{example}{Example}{Examples}
\Crefname{example}{Example}{Examples}
\crefname{conjecture}{Conjecture}{Conjectures}
\Crefname{conjecture}{Conjecture}{Conjectures}
\crefname{fact}{Fact}{Facts}
\Crefname{fact}{Fact}{Facts}
\crefname{property}{Property}{Properties}
\Crefname{property}{Property}{Properties}

\theoremstyle{definition}
\newtheorem{definition}{Definition}[section]



\SetKwComment{Comment}{\textcolor{Magenta}{$\triangleright$\ }}{}
\SetKw{KwMaintain}{\textbf{Maintain:}}
\SetKw{KwFind}{\textbf{find}}
\SetKw{KwDefine}{\textbf{define}}
\tikzset{
      czlink/.style={line width=0.45pt},
      czdot/.style={fill, circle, inner sep=2pt, radius=2pt},
      cczlink/.style={line width=0.45pt, draw=black},
      cczdot/.style={fill, circle, inner sep=2pt, radius=2pt} 
    }

\newcommand{\class}[1]{\ensuremath{\mathsf{#1}}\xspace}
\newcommand{\gap}{\operatorname{gap}}
\newcommand{\F}{\mathbb{F}}
\newcommand{\NP}{\class{NP}}
\newcommand{\BQP}{\class{BQP}}
\newcommand{\PP}{\class{PP}}
\newcommand{\SharpP}{\#\class{P}}
\newcommand{\SAT}{\textnormal{$\textsc{3-SAT}$}\xspace}
\newcommand{\problem}[1]{\textnormal{$\textsc{CircuitWidth}{#1}$}}
\newcommand{\approxproblem}[1]{\textnormal{$\textsc{CircuitWidth}$-\,#1\,-\textnormal{Approx}}}

\newcommand{\supp}{\mathsf{supp}}
\newcommand{\lit}[1]{\textcolor{blue}{\mathbf{#1}}}
\newcommand{\cl}[1]{\textcolor{red}{\mathtt{#1}}}

\newcommand{\authorfootnotesize}{\fontsize{8.3pt}{9.3pt}\selectfont}

\title{On the Complexity of the Circuit Width Problem}

\author{
      Zhengfeng Ji\thanks{{\authorfootnotesize Department of Computer Science and Technology, Tsinghua University. Email: \texttt{jizhengfeng@tsinghua.edu.cn}}}
      \and
      Yinchen Liu\thanks{{\authorfootnotesize Institute for Interdisciplinary Information Sciences, Tsinghua University. Email: \texttt{liuyinch23@mails.tsinghua.edu.cn}}}
      \and
      Zhe'ou Zhou\thanks{{\authorfootnotesize Institute for Interdisciplinary Information Sciences, Tsinghua University. Email: \texttt{zzo23@mails.tsinghua.edu.cn}}}
}
\date{\today}

\begin{document}
\maketitle

\begin{abstract}
  Polynomial representations of quantum circuits provide a clean algebraic
  framework that connects quantum computation with classical counting
  complexity~\cite{dawson,mont,Aaronson2011Permanent}.
  In this representation, amplitudes of quantum circuits can be expressed as
  normalized gaps of low-degree polynomials over $\F_2$ where the normalization
  is determined by a parameter $w(f)$ called the \emph{circuit width} of
  polynomial $f$ first introduced by Montanaro~\cite{mont}.
  In other words, the width of a polynomial directly governs the precision with
  which a quantum computer can approximate the gap, and thus quantifies the
  minimal quantum resources measured in terms of the number of qubits needed to
  carry out such approximations.
  Understanding the complexity of computing or approximating circuit width is
  therefore essential to assessing whether this counting-based approach can
  yield a meaningful approximate-counting-based combinatorial characterization
  of quantum computation and \BQP.
  We establish five complexity and algorithmic results on determining the
  circuit width, resolving the conjecture posed in~\cite{mont}:
  \begin{itemize}
    \item \textbf{Exact complexity.}
          We prove that the decision problem \problem{(f,k)}, which asks whether
          $w(f) \leq k$ on input $(f,k)$, is \NP-complete.
          The reduction encodes a \textsc{3-SAT} instance $\Phi$ to a polynomial
          $f_\Phi$ whose width satisfies $w(f_{\Phi})\leq k$ if and only if $\Phi$ is
          satisfiable.

    \item \textbf{Hardness of approximation.}
          We prove that \problem{(f,k)} remains \NP-hard to approximate within
          any multiplicative factor $(\tfrac{49}{48}-\epsilon)$ for any constant
          $0<\epsilon<1/48$, ruling out quantum-efficient additive
          $O(\log n)$-approximations for the width function unless
          $\NP \subseteq \mathsf{BQP}$.

    \item \textbf{Hardness for degree-2 polynomials.}
          We show that the exact and approximate hardness results also extend to
          degree-2 polynomials, demonstrating that the hardness result holds for
          more general gate sets and for circuit width defined on graphs rather
          than polynomials.

    \item \textbf{Nondeterministic search algorithm.}
          We give a nondeterministic polynomial-time algorithm for the search
          version of \problem{(f,k)} that uses only $O(\log_2\binom{n}{k}) = O(k \log(en/k))$ bits
          of witness.

    \item \textbf{Fixed-parameter tractability.}
          We show that \problem{(f,k)} is fixed-parameter tractable with runtime
          $k^{O(k)}\cdot n$ when parameterized with the bound $k$ on the circuit
          width.
  \end{itemize}
  Taken together, our results indicate that there is a fundamental barrier to
  the approach of Montanaro for obtaining a combinatorial characterization of
  \BQP from the approximate counting perspective.
  Our results place circuit width firmly within classical complexity theory and
  indicate that it behaves as a structural parameter---analogous to graph width
  measures---rather than as a direct measure of intrinsic quantum computational
  power.
\end{abstract}

\newpage
\tableofcontents
\newpage

\newcommand{\GH}{\mathcal{G}_H}
\newcommand{\localendpoint}[1]{#1}

\section{Introduction}\label{sec:intro}
Quantum computation is widely believed to outperform classical computation on
important tasks such as simulating quantum systems~\cite{Feynman1982} and
factoring large integers~\cite{Shor1997}.
Despite more than four decades of research, understanding the power and
limitations of quantum computers remains a central challenge in theoretical
computer science.
A primary approach to addressing this challenge involves studying combinatorial
and algorithmic problems related to quantum circuits, which have become the
central formalism for quantum computation ever since Yao's proof of the
equivalence theorem between the quantum circuit model and the quantum Turing
machine model in 1993~\cite{Yao1993}.

One fruitful area of research in this direction is to relate quantum circuits to
classical counting problems.
It has long been known that quantum circuits can be simulated using counting
oracles, yielding the inclusion
$\BQP \subseteq \PP \subseteq \class{P}^{\SharpP}$~\cite{Adleman1997}.
Later refinements showed that quantum algorithms can be captured even more
tightly within counting-based complexity classes such as
$\textsf{AWPP}$~\cite{FortnowRogers1999}.
Similar counting-based ideas are utilized in Aaronson's proof of the hardness of
the permanent via linear optics~\cite{Aaronson2011Permanent}, Kuperberg's
results on approximating the Jones polynomial~\cite{Kuperberg2009}, and
Fujii-Morimae's work on Ising partition functions~\cite{FujiiMorimae2017}.

This counting perspective becomes more explicit through an algebraic
representation of quantum circuits, motivated by a discrete analogue of the
Feynman path-integral formulation applied to circuits.
Dawson et al.~\cite{dawson} showed that amplitudes of quantum circuits composed
of Toffoli and Hadamard gates can be expressed as normalized counts of solutions
to systems of polynomial equations over $\F_{2}$.
Montanaro~\cite{mont} refined this framework by adopting the gate set
$\mathcal{Z} = \{H,Z,\mathrm{CZ},\mathrm{CCZ}\}$, showing that any quantum
circuit $\mathcal{C}$ of $w$ qubits can be associated with a single degree-3
polynomial $f_{\mathcal{C}}$ over $\F_{2}$ without a constant term.
In this representation, the amplitude $\braket{0^w | \mathcal{C} | 0^w }$ is
proportional to the gap denoted by $\gap(f_{\mathcal{C}})$, which is the
difference between the number of solutions and non-solutions of the polynomial
$f_{\mathcal{C}}$.
This connection between an entry of a unitary circuit $\mathcal{C}$, which is
\BQP-complete, and the gap of $f_{\mathcal{C}}$, which reflects massive
cancellation effects analogous to quantum interference, is therefore believed to
underlie both the power of quantum computation and the classical hardness of
simulating quantum systems.
Closely related ideas also underpin hardness results for quantum sampling
problems on instantaneous quantum polynomial-time (IQP) circuits or
boson-sampling circuits~\cite{Bremner2011,Aaronson2011}.
In~\cite{Dalzell2020}, fine-grained complexity assumptions regarding the
counting problem of the gap of degree-3 polynomials are employed to estimate the
number of qubits needed for quantum sampling supremacy for IQP and boson
sampling.

Since quantum computers can estimate the value
$\braket{0^w | \mathcal{C} | 0^w }$, we can also approximate the gap of a polynomial $f$ to a
certain precision.
Given a circuit $\mathcal{C}$ over the gate set $\mathcal{Z}$ and its associated
degree-3 polynomial $f_{\mathcal{C}}$, a first natural question is: what precision can
a quantum computer achieve in the approximate counting of the gap?
To address this question, Montanaro~\cite{mont} defined the width $w(f)$ of a
polynomial $f$ as the minimum number of qubits required by any quantum circuit
whose associated polynomial equals $f$.
The smaller the width, the higher the precision that a quantum computer can
achieve.
Montanaro showed that, given a circuit of width $w(f)$, a quantum computer can
approximate $\gap(f)$ to absolute error $2^{(n + w(f))/2} / 3$ where $n$ is the
number of variables of $f$.
This observation shows that the complexity of width optimization is directly
tied to the complexity of approximate counting problems that characterize
quantum computation.
The next natural question is then how hard it is to compute or approximate the
width of a polynomial.
If width could be efficiently minimized, even approximately, then one could
combine width optimization with standard amplitude-estimation procedures to
obtain a new and conceptually simple combinatorial characterization of $\BQP$
based entirely on approximate counting with known precision~\cite{mont}.

In this paper, we undertake a systematic study of the quantum circuit width
problem, formally denoted as \problem{(f,k)}, which asks whether the width
$w(f)$ is less than or equal to $k$.
We show that \problem{(f,k)} is \NP-complete, and that even approximating the
circuit width within a small constant factor remains \NP-hard.
Furthermore, we extend our hardness results to degree-2 polynomials,
demonstrating that both the circuit width problem and its approximate version
remain \NP-hard.
This extension has broad implications because many quantum gate sets yield
degree-2 polynomials: for instance, circuits using $\{H, \mathrm{CS}\}$ or Clifford$+T$ gates produce polynomials with at most quadratic
terms modulo $4$ or $8$ depending on the gate set~\cite{WCYJ25}.
Thus, our approximation hardness result establishes that the same computational
intractability is universal across all these diverse gate sets.
Moreover, since degree-2 polynomials without linear terms correspond to graphs,
circuit width in this setting can be interpreted as a graph parameter related to
treewidth or pathwidth, but carrying a quantum information-theoretic
interpretation concerning the arrangement of Hadamard gates and diagonal
two-qubit gates.
Consequently, our results rule out efficient width optimization across all these
practically relevant quantum computing scenarios on classical or quantum
computers unless unlikely complexity-theoretic collapses occur.
We further complement these results with algorithmic upper bounds, including a
nondeterministic search algorithm with succinct witnesses and an efficient fixed-parameter
tractable algorithm parameterized by the circuit width.

The implications of our results are threefold.
First, we answer the open question posed in Montanaro's paper concerning the
complexity of the circuit width problem~\cite{mont}.
Our \NP-hardness results show that there is a fundamental barrier to obtaining a
combinatorial characterization of $\BQP$ via this approximate counting approach.
Second, circuit width provides a natural framework for minimizing qubit usage in
quantum computation, with potential relevance to practical implementations and
compilation of quantum algorithms.
Our hardness results indicate that, in the worst case, even quantum computers
cannot solve this problem efficiently or even approximately.
Any further progress in this direction must therefore rely on heuristic methods
and exploit the structure of circuits for specific algorithms.
Third, our hardness results, combined with the fixed-parameter tractability
algorithm, suggest that circuit width behaves more like a structural
parameter---such as treewidth---than a parameter that captures intrinsic quantum
computational power.
Every given polynomial $f$ can be realized by an IQP circuit of $n$ qubits,
implementing all quadratic terms in $f$ using $\mathrm{CZ}$ gates; and if there
is a circuit of width $w$ for $f$ using $h$ internal Hadamard gates for the
quadratic terms, we have $n = w + h$ and the total number of Hadamard gates is
$2w + h = n + w$, where the $2w$ term arises from the Hadamard gates at the
beginning and end of $\mathcal{C}$ as required in the definition of
$f_{\mathcal{C}}$.
Thus, minimizing the circuit width is equivalent to minimizing the use of the
interference-creating gate $H$ when creating the amplitude
$\braket{0^{w} | \mathcal{C} | 0^{w}}$.
To put it in perspective, treewidth governs the complexity of classical
simulation methods based on tensor networks~\cite{MarkovShi}, while linear
rank-width characterizes the complexity of simulation based on decision
diagrams~\cite{CWD+25}.
The circuit width considered in this work naturally captures the quantum
resources, as measured by the number of qubits required for creating certain
amplitudes.
Our hardness results, however, suggest that circuit width minimization problem
primarily reflects combinatorial structure because of its \NP-hardness rather
than the quantum computational power typically characterized by \BQP.

\subsection{Our Results}\label{sec:results}
In this section, we summarize our main results on the circuit width problem
\problem{(f,k)}.
The results are twofold: hardness results and algorithmic results.
\subsubsection{Hardness Results}
For the hardness of circuit width, we rule out the possibility of efficiently
computing the width of a degree-3 polynomial $f$ with no constant term by
proving that the problem of deciding $w(f)\le k$ (denoted by \problem{(f,k)}) is
\NP-complete.

\begin{theorem}[See also \cref{cor:np-completeness}]
    \textnormal{\problem{}} is \NP-complete.
\end{theorem}

In fact, we are able to prove an inapproximability result: \problem{} is even
\NP-hard to approximate within a multiplicative factor of about $1.02$ (even
restricted to instances with $k = \Theta(n)$), and hence it is impossible to
approximate the width efficiently (in $ \mathsf{BQP}$ power) within additive
$O(\log n)$ error unless $\mathsf{NP} \subseteq \mathsf{BQP} $.

\begin{theorem}[Informal, see also \cref{thm:approxnphard}]
  For any constant $0<\epsilon<1/48$, it is \NP-hard to distinguish degree-3
  polynomials $f$ with no constant term and \(w(f)\leq k\) from those with
  \(w(f)\geq (\tfrac{49}{48}-\epsilon)k\).
\end{theorem}

The previous hardness results crucially rely on cubic terms (CCZ gates) to
enforce complex geometric constraints.
However, many quantum computational models naturally correspond to degree-2
polynomials.
For instance, circuits composed of Hadamard and diagonal gates such as
$\{H, \mathrm{CS}\}$ (controlled-S gate), or even Clifford$+T$ gates
$\{H, \mathrm{CZ}, T\}$ yield polynomials with at most quadratic terms modulo
$4$, or $8$ depending on the gate set.
In these cases, it is also possible that many different circuits may correspond
to the same polynomial as the $\mathrm{CZ}$ and Hadamard gate $H$ contribute to
quadratic terms of the same form (see e.g.~\cite{WCYJ25}).
A natural question thus arises: does the approximation hardness persist when we
restrict to \emph{only} degree-2 polynomials?
We answer this affirmatively by showing that approximation hardness persists
with the same factor:
\begin{theorem}[Informal, see \cref{thm:degree2-hard}]
  The approximation hardness of \textnormal{\problem{}} holds for degree-2
  polynomials: for any constant $0<\epsilon<1/48$, it is \NP-hard to distinguish
  quadratic $f$ with \(w(f)\leq k\) from those with
  \(w(f)\geq (\tfrac{49}{48}-\epsilon)k\).
\end{theorem}

Finally, from the viewpoint of conditional lower bounds, our reduction from
\SAT{} also yields an exponential lower bound under the standard Exponential
Time Hypothesis (ETH)~\cite{ETH0}:

\begin{theorem}[See also \cref{thm:eth}]\label{thm:eth-intro}
  Assuming the Exponential Time Hypothesis \textnormal{(ETH)}, no algorithm can
  solve \problem{(f,k)} in time $2^{o(n)}$.
\end{theorem}

The conditional exponential lower bound will match the runtime of one of our
algorithms (up to constant factors in the exponent) when $k=\Theta(n)$, see
\cref{cor:deterministic-intro} below.

\subsubsection{Algorithmic results}
On the positive side, we provide a nondeterministic polynomial-time algorithm
for the search version of \problem{(f,k)} with improved witness complexity:

\begin{theorem}[Informal, see \cref{thm:nondeterministic}]
  There exists a nondeterministic polynomial-time algorithm that, given a
  witness of size $2 \log_2 \binom{n}{k} =  O(k \log(en/k))$ bits, either
  constructs a quantum circuit on at most $k$ qubits with associated polynomial
  $f$, or correctly concludes that no such circuit exists.
\end{theorem}

As a corollary, we can use the nondeterministic algorithm to derive a
deterministic algorithm with runtime
$O\left(\binom{n}{k}^2 \cdot \textnormal{poly}(n)\right)$ by enumerating all
possible witnesses.
In the case where $k = \Theta(n)$, the running time is $2^{O(n)}$, while the
trivial algorithm that enumerates the order of all $\Theta(n)$ Hadamard gates
reads $O(n\log n)$ bits of witness or runs in time $n^{\Theta(n)}$
deterministically.
Moreover, this matches the ETH lower bound \cref{thm:eth-intro} up to constant
factors in the exponent:
\begin{corollary}[see also \cref{cor:deterministic}]\label{cor:deterministic-intro}
  There exists a deterministic algorithm that solves the search version of
  \textnormal{\problem{(f,k)}} in time
  $\binom{n}{k}^2 \cdot \textnormal{poly}(n)$, which is tight up to constant
  factors in the exponent under the Exponential Time Hypothesis
  \textnormal{(ETH)} when $k = \Theta(n)$.
\end{corollary}

Taking $k$ as a constant, this yields an XP algorithm with runtime
$n^{2k+O(1)}$.
Moreover, from the perspective of parameterized complexity~\cite{Downey2013,
  Cygan2015}, the nondeterministic algorithm can be easily transformed into a
placement of \problem{} in the parameterized complexity class $\mathsf{W}[1]$.
A natural question then arises: is the problem $\mathsf{W}[1]$-complete, or is it
fixed-parameter tractable ($\mathsf{FPT}$) with respect to parameter $k$?
We answer the latter affirmatively:

\begin{theorem}[Informal, see also \cref{thm:main_algorithm}]
  \problem{(f,k)} is fixed-parameter tractable with respect to parameter $k$.
  There exists an algorithm with runtime $k^{O(k)} \cdot n$ that solves the
  search version of \problem{(f,k)}.
\end{theorem}

The concrete runtime is $k^{6k + o(k)} \cdot n$, which is faster than the XP
algorithm when $k \leq n^{1/4-o(1)}$, making it more efficient and practical for
polynomially small $k$.

\subsection{Organization}

The remainder of this paper is organized as follows. In \cref{subsec:tech-overview}, we provide a detailed technical overview of our proofs. In \cref{sec:prelim}, we present the necessary background and preliminaries, including the circuit-to-polynomial correspondence (\cref{subsec:extraction}) and algorithm preliminaries (\cref{subsec:alg_prelim}).

In \cref{sec:hardness}, we prove the \NP-hardness of \problem{} by presenting a polynomial-time reduction from \textsc{3-SAT} (\cref{sec:reduction}). We then establish the structural lemmas and correctness proofs for this reduction (\cref{subsec:lemmas}), culminating in the main \NP-completeness theorem (\cref{subsec:main-theorems}).

In \cref{sec:approx-hardness}, we extend our hardness results in two directions. First, we prove the inapproximability of \problem{} within a constant factor by reducing from \textsc{Max-3-SAT}, establishing \cref{thm:approxnphard}. Second, we show that this approximation hardness persists even when restricted to degree-2 polynomials via a twin duplication strategy (\cref{def:twin-poly,thm:degree2-hard}), thereby establishing hardness for a broad class of gate sets admitting associated polynomials. 

In \cref{sec:alg}, we present our algorithmic results. We first develop a nondeterministic polynomial-time algorithm with succinct witnesses (\cref{thm:nondeterministic}), yielding a deterministic algorithm tight in the exponent (\cref{cor:deterministic}). We then prove that \problem{} is fixed-parameter tractable by presenting an $\mathsf{FPT}$ algorithm with runtime $k^{O(k)} \cdot n$ based on dynamic programming over tree decompositions (\cref{sec:fpt_algorithm}, \cref{thm:main_algorithm}).

\subsection{Technical Overview}\label{subsec:tech-overview}
We briefly recall the circuit--polynomial correspondence~\cite{mont} (see \cref{subsec:extraction} for details). For circuits over $\{H,Z,\mathrm{CZ},\mathrm{CCZ}\}$, we need to add $H$ gates at the beginning and end of each qubit wire. Each wire is then split by its $H$ gates into segments (time intervals). We label the segments using \emph{symbols}, and we denote $\mathcal{S}$ as the set of all symbols. The polynomial has the form
\[
f_{\mathcal{C}}(x)=\sum_{u\in \mathcal{S}} \alpha_u u \;+\!\!\!\!\sum_{\{u,v\}\subseteq \mathcal{S}} \beta_{uv} uv \;+\!\!\!\!\sum_{\{u,v,w\}\subseteq \mathcal{S}} \gamma_{uvw} uvw,
\]
where coefficients are in $\F_2$ and the constant term is zero. A $Z$ gate on symbol $u$ sets $\alpha_u=1$, a $\mathrm{CZ}$ gate between overlapping symbols $u,v$ and an internal $H$-gate link between adjacent symbols $u,v$ both set $\beta_{uv}=1$, and a $\mathrm{CCZ}$ gate on three mutually overlapping symbols sets $\gamma_{uvw}=1$. An example of a quantum circuit and its polynomial is shown in Figure~\ref{fig:example-circuit}.

Our main focus is to investigate the complexity of the following problem, \problem{(f,k)}:
Given a degree-3 polynomial $f$ with no constant term, the task is to decide whether one can realize $f$ using a quantum circuit with at most $k$ qubits, i.e., whether $w(f)\le k$.

We also study the search version of \problem{(f,k)}: given $f$ and $k$, either find a quantum circuit on at most $k$ qubits whose polynomial equals $f$, or conclude that no such circuit exists.

  \begin{figure}[H]
    \centering
\begin{tikzpicture}[xscale=1.5, yscale=1.5, baseline=(current bounding box.center)]
  \foreach \y in {2,1,0} {
    \draw (0,\y) -- (8,\y);
  }

  \foreach \y in {2,1,0} {
    \node[draw,fill=blue!20,inner sep=5pt] at (0,\y) {$H$};
    \node[draw,fill=blue!20,inner sep=5pt] at (8,\y) {$H$};
  }

  \foreach \pos/\y in {1/2, 6/2, 3/1, 5/0} {
    \node[draw,fill=blue!20,inner sep=5pt] at (\pos,\y) {$H$};
  }

  \foreach \pos/\y in {6.5/1} {
    \node[draw,fill=blue!20,inner sep=5pt] at (\pos,\y) {$Z$};
  }

  \draw[cczlink] (4,2) -- (4,0);
  \foreach \y in {2,1,0} {
    \fill[cczdot] (4,\y) circle (2pt);
  }

  \draw[czlink] (2,2) -- (2,1);
  \foreach \y in {2,1} {
    \fill[czdot] (2,\y) circle (2pt);
  }
  \draw[czlink] (5,2) -- (5,1);
  \foreach \y in {2,1} {
    \fill[czdot] (5,\y) circle (2pt);
  }
  \node[above] at (0.5,  2) {$x_1$};
  \node[above] at (3.5,  2) {$x_2$};
  \node[above] at (7.0,  2) {$x_3$};

  \node[above] at (1.5,  1) {$y_1$};
  \node[above] at (5.5,  1) {$y_2$};

  \node[above] at (2.5,  0) {$z_1$};
  \node[above] at (6.5,  0) {$z_2$};
\end{tikzpicture}
\caption{A quantum circuit with $3$ qubits and $4$ internal Hadamard gates. The polynomial extracted from this circuit is
\(f_\mathcal{C} = y_2 + x_1x_2 + x_2x_3 + y_1y_2 + z_1z_2 + x_2y_1 + x_2y_2 + x_2y_2z_1 \).}
  \label{fig:example-circuit}
  \end{figure}
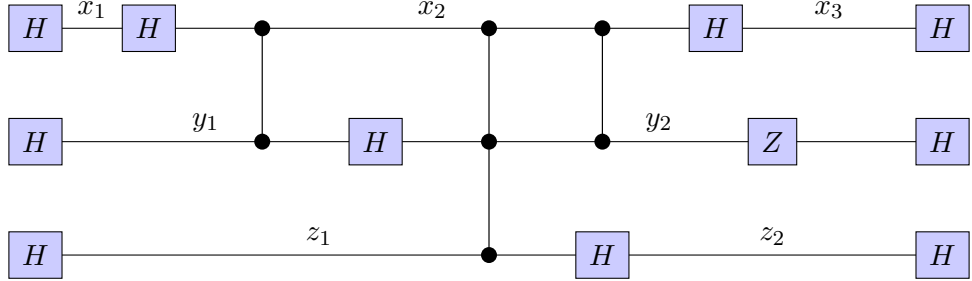

\subsubsection{$\mathsf{NP}$-completeness}\label{subsubsec:intro-nphard}
We reduce \textsc{3-SAT} to \problem{} by mapping a formula $\Phi$ with $m$
clauses to a degree-3 polynomial $f_\Phi$ and a width budget $k=1+6m$ such that
$\Phi$ is satisfiable if and only if $w(f_\Phi)\le k$. The reduction encodes
Boolean information in the \emph{geometry} of a minimum-width realization: truth
values are represented by which segment overlaps which, and the width bound
forces a rigid layout that can only occur when all clauses are satisfied.

\paragraph{Circuit-to-polynomial intuition.}
By the circuit--polynomial correspondence, each segment between consecutive $H$ gates becomes
a symbol, and monomials arise precisely from overlaps induced by $Z$, $H$, CZ, and
CCZ gates. Hence, by prescribing a polynomial we can constrain which segments
must overlap in any realizing circuit. Specifically, a cubic term $u v w$ forces the three segments $u,v,w$ to overlap at some time, while a quadratic term $u v$ may either arise from a CZ gate between overlapping segments $u,v$ or from an $H$ gate link between adjacent segments $u,v$ on the same wire. A direct computation shows that a quantum circuit $\mathcal{C}$ on $k$ wires with $h$ internal $H$ gates has exactly $k + h$ segments, and hence the number of symbols in the polynomial is $|\mathcal{S}| = k + h$. Therefore, intuitively, to minimize the number of wires $k$ for a given polynomial $f_{\Phi}$, one should try to maximize the number of internal $H$ gates $h$ in the realizing circuit. Note that maximizing $h$ is equivalent to minimizing the number of CZ gates, since their sum is fixed by the number of quadratic terms in $f_{\Phi}$.

Our construction uses this intuition by constructing $f_{\Phi}$ using $1 + 18m$ symbols and $12m$ quadratic terms (recall $m$ is the number of clauses in $\Phi$). Therefore, by the above argument, any realizing circuit must have at most $h\leq 12m$ internal $H$ gates, so that if the circuit realizes $f_{\Phi}$, the number of wires $k \geq |\mathcal{S}| - h \geq 1+ 6m$, giving a lower bound on the width. Furthermore, a circuit realizes $f_{\Phi}$ with $k=1+6m$ wires if and only if the above lower bound achieves equality, i.e., if and only if the circuit uses no CZ gates, constraining which symbols are on the same wires. Additional cubic terms in $f_{\Phi}$ then enforce that the only way to achieve wire-optimality (and hence no CZ gates) is to have a geometric layout corresponding to a satisfying assignment of $\Phi$. The construction is detailed below.

\paragraph{Notation: variables vs.\ literals.}
We write $x,y,z$ for Boolean variables and $\lit{x},\lit{y},\lit{z}$ for the
corresponding literals, each of which may be positive or negated. For an
occurrence whose underlying variable is $x$, the local symbols are always
$x_{\cl{c}},\hat{x}_{\cl{c}},\bar{x}_{\cl{c}}$; the literal $\lit{x}$ is
represented by $x_{\cl{c}}$ if it is positive and by $\bar{x}_{\cl{c}}$ if it is
negated. The same convention applies to $y$ and $z$.

\paragraph{Gadgets and roles.}
Each qubit wire is cut by at most two $H$ gates into three segments. The gadgets
are organized as follows:
\begin{enumerate}[(i)]
  \item A \emph{base wire} $b$ contains no $H$ gates and appears only in selected
        CCZ terms. It acts as a global synchronizer that enforces the intended
        relative placement of segments across different wires.
  \item For each clause $\cl{c}$, a \emph{truth-marker wire} is split into
        $t_{\cl{c}},\hat{t}_{\cl{c}},\tilde{t}_{\cl{c}}$ (from left to right). This can be achieved by adding quadratic terms $t_{\cl{c}}\hat{t}_{\cl{c}} + \hat{t}_{\cl{c}}\tilde{t}_{\cl{c}} $ to force this sequential structure in wire-optimal (no CZ gates permitted) circuits and using additional CCZ terms to force its order to be the same across clauses; hence, we can assume the order is left to right. In the following, we will interpret a variable or clause segment as \bool{true} exactly when it overlaps $t_{\cl{c}}$; \bool{false} when it does not overlap $t_{\cl{c}}$.
  \item For each literal occurrence, a \emph{variable wire} carries
        $x_{\cl{c}},\hat{x}_{\cl{c}},\bar{x}_{\cl{c}}$ (sequential structure is the same as before, but order is not required). The CCZ terms
        $b\,\hat{t}_{\cl{c}}\,x_{\cl{c}}$, $b\,\hat{t}_{\cl{c}}\,\hat{x}_{\cl{c}}$, and
        $b\,\hat{t}_{\cl{c}}\,\bar{x}_{\cl{c}}$ force $\hat{t}_{\cl{c}}$ to overlap all three
        segments, which in turn ensures that exactly one of $x_{\cl{c}}$ or
        $\bar{x}_{\cl{c}}$ overlaps $t_{\cl{c}}$ in any size-optimal realization. Thus, we will interpret $x_{\cl{c}}$ overlapping $t_{\cl{c}}$ as the corresponding variable being set to \bool{true}, and $\bar{x}_{\cl{c}}$ overlapping $t_{\cl{c}}$ as \bool{false}. Additional CCZ terms will ensure that all occurrences of the same variable across different clauses are consistent.
  \item Each clause $\cl{c}=(\lit{x}\lor\lit{y}\lor\lit{z})$ is evaluated by two
        nested OR gadgets on two additional wires with segments
        $\overleftarrow{p_{\cl{c}}},p_{\cl{c}},\overrightarrow{p_{\cl{c}}}$ and $\overleftarrow{o_{\cl{c}}},o_{\cl{c}},\overrightarrow{o_{\cl{c}}}$.
        Here $p_{\cl{c}}$ is the output segment for the first disjunction
        $\lit{x}\lor\lit{y}$, while $o_{\cl{c}}$ is the final output segment for
        the whole clause, namely $p_{\cl{c}}\lor\lit{z}$.
        The constructed CCZ terms have the following soundness/completeness behavior: if $p_{\cl{c}}$ (resp.\ $o_{\cl{c}}$) overlaps $t_{\cl{c}}$, then at least one input segment is \bool{true}; conversely, whenever at least one input segment is \bool{true}, the gadget can be placed so that the output overlaps $t_{\cl{c}}$. A final CCZ term $b\,t_{\cl{c}}\,o_{\cl{c}}$ pins $o_{\cl{c}}$ to overlap $t_{\cl{c}}$, thereby enforcing clause satisfaction. An illustration is in Figure~\ref{fig:clause-gadget-overview}.
\end{enumerate}
\begin{figure}[!htbp]
  \centering
  \begin{tikzpicture}[xscale=1.15, yscale=1.0, baseline=(current bounding box.center)]
    \tikzset{hgate/.style={draw,fill=blue!20,inner sep=3pt,font=\scriptsize}}

    \draw (0,0) -- (8,0); 
    \draw (0,1) -- (8,1); 
    \draw (0,2) -- (8,2); 
    \draw (0,3) -- (8,3); 
    \draw (0.8,4) -- (4.8,4); 
    \draw (0,5) -- (8.0,5); 
    \draw (0,6) -- (8,6); 

    \foreach \t in {0,8} { \node[hgate] at (\t,0) {$H$}; }
    \foreach \t in {0,2.8,6.4,8} { \node[hgate] at (\t,1) {$H$}; }
    \foreach \t in {0.8,4.8} { \node[hgate] at (\t,4) {$H$}; }
    \foreach \t in {0,3.6,5.2,8.0} { \node[hgate] at (\t,5) {$H$}; }
    \foreach \t in {0,1.6,6.4,8} { \node[hgate] at (\t,6) {$H$}; }

    \node at (4,-0.35) {$b$};
    \node at (0.8,1.25) {$t_{\cl{c}}$};
    \node at (4.2,1.25) {$\hat{t}_{\cl{c}}$};
    \node at (7.2,1.25) {$\tilde{t}_{\cl{c}}$};
    \node at (7.8,2.25) {$x_{\cl{c}}\ \text{qubit}$};
    \node at (7.8,3.25) {$y_{\cl{c}}\ \text{qubit}$};
    \node at (2.9,4.25) {$p_{\cl{c}}$};
    \node at (0.8,5.25) {$\bar{z}_{\cl{c}}$};
    \node at (4.2,5.25) {$\hat{z}_{\cl{c}}$};
    \node at (7.2,5.25) {$z_{\cl{c}}$};
    \node at (0.8,6.25) {$\overleftarrow{o_{\cl{c}}}$};
    \node at (4.2,6.25) {$o_{\cl{c}}$};
    \node at (7.2,6.25) {$\overrightarrow{o_{\cl{c}}}$};

    \draw[cczlink, blue] (1.2,0) -- (1.2,6);
    \foreach \y in {0,4,6} { \fill[cczdot, blue] (1.2,\y) circle; }
    \draw[cczlink, blue] (2.0,0) -- (2.0,6);
    \foreach \y in {0,4,6} { \fill[cczdot, blue] (2.0,\y) circle; }
    \draw[cczlink, blue] (6.0,0) -- (6.0,6);
    \foreach \y in {0,5,6} { \fill[cczdot, blue] (6.0,\y) circle; }
    \draw[cczlink, blue] (6.8,0) -- (6.8,6);
    \foreach \y in {0,5,6} { \fill[cczdot, blue] (6.8,\y) circle; }

    \draw[cczlink, red] (3.2,0) -- (3.2,5);
    \foreach \y in {0,1,5} { \fill[cczdot, red] (3.2,\y) circle; }
    \draw[cczlink, red] (4.4,0) -- (4.4,5);
    \foreach \y in {0,1,5} { \fill[cczdot, red] (4.4,\y) circle; }
    \draw[cczlink, red] (5.6,0) -- (5.6,5);
    \foreach \y in {0,1,5} { \fill[cczdot, red] (5.6,\y) circle; }

    \draw[cczlink] (2.4,0) -- (2.4,6);
    \foreach \y in {0,1,6} { \fill[cczdot] (2.4,\y) circle; }
  \end{tikzpicture}
  \caption{Clause gadget for the illustrative clause $\cl{c}=(\lit{x} \lor \lit{y}\lor z)$, where $\lit{x}$ and $\lit{y}$ may be either positive or negated literals but the third literal is the positive literal $z$. The $b$ and $t_{\cl{c}}$ wires are explicit; the $\lit{x}$ and $\lit{y}$ variable wires are suppressed; the $\lit{x}\lor \lit{y}$ wire is compressed to its middle segment $p_{\cl{c}}$. The encoding convention is: a symbol is \bool{true} if and only if its segment overlaps $t_{\cl{c}}$. In the drawing, $p_{\cl{c}}$ overlaps $t_{\cl{c}}$ (so $p_{\cl{c}}$ is \bool{true}, meaning at least one of $\lit{x}, \lit{y}$ is \bool{true}), while the $z$ wire is drawn with $\bar{z}_{\cl{c}}$ overlapping $t_{\cl{c}}$ (so $z$ is \bool{false}). The CCZ columns include the gadget terms on the $z$ wire and the $(p_{\cl{c}}\lor z)$ gadget, plus the pinning term $b\,t_{\cl{c}}\,o_{\cl{c}}$.}
  \label{fig:clause-gadget-overview}
\end{figure}
\paragraph{Wire optimality.}
As mentioned earlier, a circuit realizes $f_{\Phi}$ with $k$ wires if and only if it uses no CZ gates. Therefore, quadratic terms force each triple of symbols to lie on the required wires (including \emph{truth-marker wires}, \emph{variable wires} and \emph{clause wires}) with two
Hadamard gates. The only freedom is the order of variable segments on their wires, which
corresponds to truth assignments measured by their overlap with truth-marker segments. Then, the cubic terms and clause gadgets in $f_{\Phi}$ enforce that each clause is satisfied under this assignment. Completeness and soundness follow from the construction and the above arguments, thereby proving $\Phi$ is satisfiable if and only if $w(f_\Phi)\le k$.

\paragraph{Exponential Time Hypothesis lower bound.}
The polynomial-time reduction from \SAT{} also yields a conditional lower bound under the Exponential Time Hypothesis (ETH). Since our reduction transforms a \SAT{} instance with $m$ clauses into a \problem{} instance with $n=\Theta(m)$ symbols and parameter $k=\Theta(m)$, any $2^{o(n)}$-time algorithm for \problem{} would imply a $2^{o(m)}$-time algorithm for \SAT{}. When we restrict to linear-size \SAT{} (i.e., the number of clauses is bounded by the number of variables), the hypothetical algorithm would contradict a sparsification variant of ETH~\cite{ETH}. This establishes that \problem{} requires $2^{\Omega(n)}$ time (specifically when restricted to $k=\Theta(n)$).

\subsubsection{Approximation Hardness}\label{subsubsec:approx-intro}
Beyond exact computation, we show that \problem{} is \NP-hard to approximate within any factor better than $\tfrac{49}{48}-\epsilon$ for every fixed $0<\epsilon<1/48$. We first observe that the previous reduction in \cref{subsubsec:intro-nphard} is ``almost'' completely local: each clause in $\Phi$ contributes a constant number of symbols to $f_{\Phi}$, denoted by $\mathcal{S}_{\cl{c}}$ (with some quadratic and cubic terms inside them):
\[
    \underbrace{\{t_{\cl{c}}, \hat{t}_{\cl{c}}, \tilde{t}_{\cl{c}}\}}_{\text{Truth marker}} \cup 
    \underbrace{\{x_{\cl{c}}, \hat{x}_{\cl{c}}, \bar{x}_{\cl{c}}\}}_{\text{Var } x} \cup 
    \underbrace{\{y_{\cl{c}}, \hat{y}_{\cl{c}}, \bar{y}_{\cl{c}}\}}_{\text{Var } y} \cup 
    \underbrace{\{z_{\cl{c}}, \hat{z}_{\cl{c}}, \bar{z}_{\cl{c}}\}}_{\text{Var } z} \cup 
    \underbrace{\{\overleftarrow{p_{\cl{c}}}, p_{\cl{c}}, \overrightarrow{p_{\cl{c}}}\}}_{\text{OR } p_{\cl{c}}} \cup 
    \underbrace{\{\overleftarrow{o_{\cl{c}}}, o_{\cl{c}}, \overrightarrow{o_{\cl{c}}}\}}_{\text{OR } o_{\cl{c}}},
    \]
    plus some additional \emph{cubic} terms to force consistency of positions of clauses and variable assignments across clauses. Specifically, no quadratic term exists between $b$ and symbols from $\mathcal{S}_{\cl{c}}$ for different clauses $\cl{c}\neq \cl{c'}$. Thus, any realizing circuit must place $b$ and $\mathcal{S}_{\cl{c}}$ on separate qubit wires.
    
    Moreover, when a clause $\cl{c}$ is placed on exactly 6 wires (called a \emph{good clause}), the variable segments \emph{encode} a truth assignment: whether it overlaps the truth marker segment $t_{\cl{c}}$ (encoding \bool{true}), or does not overlap (encoding \bool{false}), the same as in the exact case. The \emph{OR gadgets} formed by cubic terms on $\mathcal{S}_{\cl{c}}$ then \emph{enforce} that at least one of the three literals in $\cl{c}$ is \bool{true} under this assignment. The \emph{variable gadgets} across clauses ensure that all occurrences of the same variable across different clauses are consistent.

    Therefore, when we restrict to good clauses, the construction behaves like in the exact case and gives a variable assignment (partial at first, but we can arbitrarily assign remaining variables). If a clause $\cl{c}$ is unsatisfied, then the corresponding clause gadget in $f_{\Phi}$ forces $\mathcal{S}_{\cl{c}}$ to use at least one additional wire in any realizing circuit (and in fact, one additional wire suffices). As a result, we can create a gap between the ``Yes'' and ``No'' instances by leveraging the hardness of approximating \textsc{Max-3-SAT}, whose task is to determine the maximum number of clauses that can be simultaneously satisfied by any assignment:
    \begin{theorem}[Reduction from \textsc{Max-3-SAT}, see also \cref{thm:maxsat-reduction}]\label{thm:maxsat-reduction-intro}
  Let $\Phi$ be a \SAT{} instance with $m$ clauses and let $t$ be an integer.
  Construct $f_\Phi$ in polynomial time as in \cref{subsubsec:intro-nphard}.
  Then $w(f_\Phi)\le 1 + 7m - t$ if and only if $\Phi$ has an assignment satisfying at least $t$ clauses.
\end{theorem}
It is known from Håstad's $7/8$ theorem~\cite{Hastard} that \textsc{Max-3-SAT} is \NP-hard to distinguish satisfiable instances from those in which at most a $7/8+\delta$ fraction of clauses can be satisfied, for any constant $\delta>0$, see \cref{thm:hastad}. Now, the reduction works as follows: if $\Phi$ is satisfiable, the constructed polynomial $f_\Phi$ has width $w(f_\Phi)\le 1+6m$; if every assignment satisfies at most $(7/8+\delta)m$ clauses, then the reduction guarantees $w(f_\Phi)\ge 1+6m+(1/8-\delta)m$. Choosing $\delta=\epsilon$, the additive $1$ is negligible for large $m$, and the gap is enough to distinguish $w(f)\le k$ from $w(f)\ge(\tfrac{49}{48}-\epsilon)k$ for every fixed $0<\epsilon<1/48$, matching the formal statement in \cref{thm:approxnphard}.

\subsubsection{Hardness for Degree-2 Polynomials}\label{subsubsec:degree2}
We further extend our inapproximability result to degree-2 polynomials, showing that \problem{} remains \NP-hard to approximate within any factor better than $\tfrac{49}{48}-\epsilon$ even when the input polynomial has degree 2, for every fixed $0<\epsilon<1/48$, see \cref{thm:degree2-hard}. 

Due to technical losses in proving lower bounds on the number of wires needed, we cannot obtain a ``tight'' reduction from \textsc{Max-3-SAT} as in \cref{thm:maxsat-reduction-intro}. Instead, we directly use the inapproximability result from Håstad's $7/8$ theorem~\cite{Hastard}. Given a \SAT{} instance $\Phi$ with $m$ clauses, we first construct the degree-3 polynomial $f_\Phi$ as in \cref{subsubsec:intro-nphard}, and then transform it into a degree-2 polynomial $f'_\Phi$ using a \emph{twin duplication} strategy (\cref{def:twin-poly}).
\paragraph{The twin construction.}
The core technique replaces each symbol $s$ in the original polynomial $f_\Phi$ with many copies $s^{(1)}, s^{(2)}, \ldots, s^{(M)}$, connected by a complete graph of CZ gates (a ``\emph{clique constraint}'') implemented via quadratic terms:
\[
\sum_{1 \le i < j \le M} s^{(i)} s^{(j)}.
\]
The base symbol $b$ receives even more copies ($Nm$ copies), where $N$ and $M$ are two large parameters to be set with $N \ll M$. Each quadratic term $u v$ is replaced by a ``\emph{matching}'' connecting only copies with the same index. Each cubic term $b\,u\,v$ in $f_\Phi$ is then replaced by a ``\emph{tripartite complete graph}'' connecting all pairs of copies across the three symbols. The resulting ``\emph{twin polynomial}'' $f'_\Phi$ clearly has degree 2. It is easy to see that, if $\Phi$ is satisfiable, then by repeating the original circuit (with $1+6m$ wires) $M$ times (and using additional $Nm-M$ wires for copies of $b$), we can realize $f'_\Phi$ using a circuit of width at most $(6M + N)m$.

For soundness, with Håstad gap parameter $\delta>0$, we need to show that if at most $t = (7 / 8+\delta)m$ clauses of $\Phi$ can be satisfied, then any circuit realizing $f'_\Phi$ must have width at least $(7M+N)m-(7/8+\delta)Mm-750m$ (recall $N \ll M$ and $M$ is chosen large enough compared with the constant loss). 
The key geometric insight is the following clique survival property:

\begin{lemma}[see \cref{lem:twin-survival}]\label{lem:clique-survival-intro}
Let $\{s_1, \ldots, s_r\}$ be a set of symbol copies with a clique constraint $\sum_{1 \le i < j \le r} s_i s_j$ appearing in the twin polynomial $f'_\Phi$. Then in any circuit realizing $f'_\Phi$, \textbf{at most two copies from this set can share a wire} with another copy from the same set.
\end{lemma}

The proof of the lemma is clear: no three symbols can lie on the same wire, and no two groups of two symbols can share two wires respectively, establishing that the only possibilities are $r$ symbols on distinct wires, or possibly one pair sharing a wire.

This constraint ensures that the vast majority of copies ``survive'' as isolated symbols. Specifically, using \cref{lem:clique-survival-intro} and a careful analysis for cubic terms in original $f_{\Phi}$, we can show that: in any circuit realizing $f'_\Phi$, for any symbol $s\neq b$ in $f_\Phi$, at least $M - O(1)$ of its copies do not share wires with \textbf{any other copy} of $s$ and \textbf{any other copy} of symbol $u$ where $su$ appear in a \textbf{cubic} term in $f_{\Phi}$. We will call these copies ``\emph{survive globally}''.  Similarly, at least $Nm - O(m)$ copies of the base symbol $b$ survive globally.

\paragraph{Enforcing the original structure.}
Another key insight is that these globally surviving copies effectively recreate the geometric constraints of the original cubic terms. Consider a cubic term $b\,u\,v$ in $f_{\Phi}$. If three copies $b^{(i)}, u^{(j)}, v^{(k)}$ all survive globally, then by definition they lie on distinct wires from any other copy of $b$, $u$, or $v$. The tripartite complete graph connecting all pairs of these copies (arising from the twin construction) then enforces that their supports must pairwise intersect, the same as required by the original CCZ gate in $f_{\Phi}$ using 1-dimensional Helly's theorem (\cref{prelim:helly}). 

As a result, we can use the same proof strategy as in \cref{subsubsec:approx-intro} by restricting attention to globally surviving groups of symbols in $\mathcal{S}_{\cl{c}}^{(j)}$ (the groups $\mathcal{S}_{\cl{c}}^{(j)}, j\in \{1, 2, \dots, M\}$ so that all symbol copies in it are globally surviving). A union bound argument shows that for each clause $\cl{c}$, at least $M - O(1)$ such globally surviving groups exist. 

We can call such a group a \emph{good clause copy} if it uses exactly 6 wires. Then, by the same arguments as before, each good clause copy encodes a variable assignment (consistency ensured by global survival) and enforces clause satisfaction. For each clause $\cl{c}$, there are at least $M - O(1)$ \emph{globally surviving clause copies}---each copy uses 6 wires if the clause is satisfied under the encoded assignment, and at least 7 wires otherwise. The lower bound argument then proceeds identically to the degree-3 case, up to the constant additive loss from the survival union bound: $w(f'_\Phi) \ge (7M+N)m-(7/8+\delta)Mm-750m$ when at most a $(7/8+\delta)$ fraction of clauses is satisfiable. Taking $\delta=\epsilon$ yields the same $\tfrac{49}{48}$ hardness factor when $N \ll M = \Omega(1/\epsilon)$.

\subsubsection{Nondeterministic Algorithm} 

To complement the hardness results, we present a nondeterministic polynomial-time algorithm that improves the witness length. The algorithm builds on a structural characterization of valid circuits\footnote{Some additional clean-up procedures for $f$ are needed for the formal statement; see the discussion before \cref{lem:structure} for details.}:

\begin{lemma}[Informal, see \cref{lem:structure}]\label{lem:structure-intro}
A valid $k$-qubit circuit with polynomial $f$ exists if and only if the symbols can be partitioned into $k$ ordered lists $L_i$ (representing wires, with local source/sink sentinels) such that the corresponding quadratic monomial for each adjacent real-symbol pair in the lists appears in $f$, and a directed constraint graph $\GH$ (whose vertices represent Hadamard gates and edges enforce temporal ordering from the polynomial structure) is acyclic.
\end{lemma}

Our algorithm uses a witness specifying the \emph{head} $S_{\textnormal{head}}$ and \emph{tail} $S_{\textnormal{tail}}$ symbol sets---the leftmost and rightmost segments on each qubit wire---requiring $2\log_2\binom{n}{k}=O(k\log(en/k))$ bits. Given this witness, the algorithm deterministically constructs the circuit via a \textbf{greedy extension} procedure. It maintains a set $\mathsf{CurrentHeads}$ of current head symbols, initially equal to $S_{\textnormal{head}}$, representing the rightmost symbol on each unfinished partially-constructed wire. In each step, it identifies a head symbol $v \in \mathsf{CurrentHeads}$ that has \emph{exactly} one neighbor among the untouched symbols in the connection graph $G_f$. Here, $G_f$ is defined as follows: vertices are symbols in $\mathcal{S}$; for each quadratic term $uv$ in $f$, add an edge $\{u,v\}$; for each cubic term $uvw$ in $f$, add three edges $\{u,v\}, \{u,w\}, \{v,w\}$ (equivalently, $G_f$ is the skeleton of the hypergraph where quadratic terms are edges and cubic terms are $3$-uniform hyperedges). This \emph{unique neighbor} property is crucial: in the underlying valid circuit, the only interaction of $v$ with untouched symbols is with its immediate successor $u$ on the same wire. The algorithm then appends $u$ to $v$'s wire and updates $\mathsf{CurrentHeads}$ by replacing $v$ with $u$. If a head symbol is in $S_{\textnormal{tail}}$, that wire is marked complete and removed from $\mathsf{CurrentHeads}$.  An illustration is in \cref{fig:greedy}.

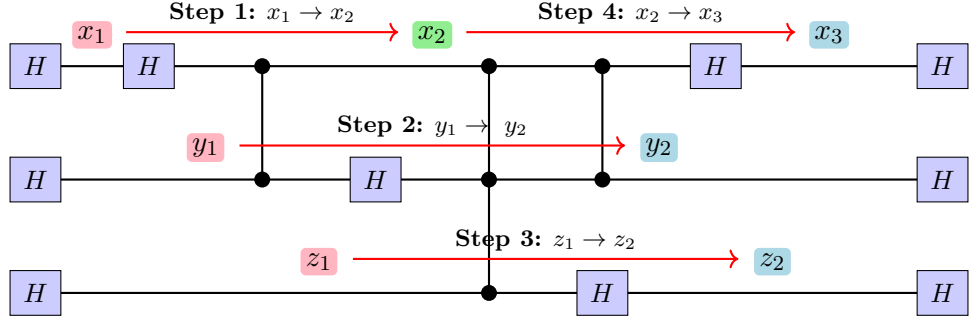
\begin{figure}[t]
\centering
\begin{tikzpicture}[xscale=1.5, yscale=1.5, baseline=(current bounding box.center)]
  \definecolor{headcolor}{RGB}{255,182,193}
  \definecolor{assignedcolor}{RGB}{144,238,144}
  \definecolor{currentcolor}{RGB}{255,255,0}
  \definecolor{tailcolor}{RGB}{173,216,230}
  
  \foreach \y in {2,1,0} {
    \draw[thick] (0,\y) -- (8,\y);
  }

  \foreach \y in {2,1,0} {
    \node[draw,fill=blue!20,inner sep=5pt,font=\small] at (0,\y) {$H$};
    \node[draw,fill=blue!20,inner sep=5pt,font=\small] at (8,\y) {$H$};
  }

  \foreach \pos/\y in {1/2, 6/2, 3/1, 5/0} {
    \node[draw,fill=blue!20,inner sep=5pt,font=\small] at (\pos,\y) {$H$};
  }

  \draw[thick] (4,2) -- (4,0);
  \foreach \y in {2,1,0} {
    \fill[black] (4,\y) circle (2pt);
  }

  \draw[thick] (2,2) -- (2,1);
  \foreach \y in {2,1} {
    \fill[black] (2,\y) circle (2pt);
  }

    \draw[thick] (5,2) -- (5,1);
  \foreach \y in {2,1} {
    \fill[black] (5,\y) circle (2pt);
  }

  \node[above,fill=headcolor,rounded corners=2pt,inner sep=2pt] at (0.5,2.15) {$x_1$};
  \node[above,fill=headcolor,rounded corners=2pt,inner sep=2pt] at (1.5,1.15) {$y_1$};
  \node[above,fill=headcolor,rounded corners=2pt,inner sep=2pt] at (2.5,0.15) {$z_1$};
  
  \node[above,fill=assignedcolor,rounded corners=2pt,inner sep=2pt] at (3.5,2.15) {$x_2$};
  
  \node[above,fill=tailcolor,rounded corners=2pt,inner sep=2pt] at (5.5,1.15) {$y_2$};
  
  \node[above,fill=tailcolor,rounded corners=2pt,inner sep=2pt] at (7.0,2.15) {$x_3$};
  \node[above,fill=tailcolor,rounded corners=2pt,inner sep=2pt] at (6.5,0.15) {$z_2$};

  \draw[->,thick,red] (0.8,2.3) -- (3.2,2.3);
\draw[->,thick,red] (1.8,1.3) -- (5.2,1.3);
  \draw[->,thick,red] (2.8,0.3) -- (6.2,0.3);
\draw[->,thick,red] (3.8,2.3) -- (6.7,2.3);
\node[below,font=\footnotesize] at (5.25,2.65) {\textbf{Step 4:} $x_2 \rightarrow x_3$};

  \node[below,font=\footnotesize] at (2,2.65) {\textbf{Step 1:} $x_1 \rightarrow x_2$};
  \node[below,font=\footnotesize] at (3.5,1.65) {\textbf{Step 2:} $y_1\rightarrow \ y_2$};
  \node[below,font=\footnotesize] at (4.5,0.65) {\textbf{Step 3:} $z_1 \rightarrow z_2$};

\end{tikzpicture}
\caption{\small Illustration of the nondeterministic algorithm's progress for $f = x_1x_2 + x_2x_3 + y_1y_2 + z_1z_2 + x_2y_1 + x_2y_2 + x_2y_2z_1 $. Given the witness $S_{\text{head}} = \{x_1, y_1, z_1\}$ (pink) and $S_{\text{tail}} = \{x_3, y_2, z_2\}$ (blue), the algorithm deterministically finds the circuit. At each step, it finds a current head symbol (green) with degree $1$ in the graph $G_f$ and extends the wire. The first step detects (for example) $x_1$ and replaces $x_1$ by $x_2$ in the current head, then $y_1$ by $y_2$. After replacement, the algorithm detects that the wire should be complete. Then, it detects $z_1$ and replaces $z_1$ by $z_2$, marks the wire complete, and finally replaces $x_2$ by $x_3$.
}
\label{fig:greedy}
\end{figure}

The unique neighbor property and correctness of \textbf{greedy extension} are abstracted as follows:

\begin{lemma}[Informal, see \cref{lem:degree1}]\label{lem:unique neighbor}
If the witness ($S_{\textnormal{head}}$ and $S_{\textnormal{tail}}$) are exactly the symbols at the start and end of each wire in some valid $k$-qubit circuit realizing $f$, then at every step of the greedy algorithm, there exists at least one current head symbol $v \in \mathsf{CurrentHeads}$ with degree exactly $1$ to untouched symbols in $G_f$.
\end{lemma}

The proof uses a time-ordering argument: among all symbols in $\mathsf{CurrentHeads}$ (each occupying a time interval in the circuit), consider the one whose interval \emph{ends earliest}. This symbol cannot have edges to untouched symbols on different wires (as those symbols start later and would not overlap), and thus has degree 1---its unique neighbor being the next symbol on the same wire.

Once all symbols are partitioned into $k$ ordered lists, a subsequent Check subroutine verifies that Hadamard gates can be consistently ordered across wires by building a directed acyclic graph of temporal constraints and computing a topological sort. 

\subsubsection{Efficient Fixed-Parameter Tractable Algorithm}\label{sec:fpt}

In this section, we present an $\mathsf{FPT}$ algorithm with runtime $k^{O(k)} \cdot n$, see \cref{thm:main_algorithm}.

The key insight is that if a valid $k$-qubit circuit exists, then the connection graph $G_f$ (whose vertices are symbols in $\mathcal{S}$ and edges correspond to pairs appearing in monomials of $f$) has small treewidth---specifically, pathwidth at most $k$. This follows from the correctness of the \textbf{greedy extension} procedure in \cref{lem:unique neighbor}: simulating the procedure, while retaining completed tails in the auxiliary head sets, yields a path decomposition of $G_f$ whose bags are the initial head set and the consecutive auxiliary head-set unions $S_{t-1}\cup S_t$, each of size at most $k+1$ (see \cref{lem:pathdecomposition}).

With the pathwidth (and hence treewidth) upper bound at hand, establishing that \problem{} is \textsf{FPT} becomes a direct consequence of Courcelle's Theorem~\cite{Courcelle1990}. The existence of a valid partition into $k$ wires can be expressed in Monadic Second-Order logic ($\textnormal{MSO}_2$) over the natural labelled incidence structure of the polynomial: we check for $k$ disjoint ordered paths whose direct adjacencies correspond to quadratic monomials of $f$, such that the constraint graph $\GH$ induced by all quadratic and cubic terms (as defined in \cref{lem:structure}) remains acyclic. Courcelle's Theorem then implies this property can be checked in linear time with a constant depending on $k$ (see \cref{sec:fpt_algorithm} for details). This shows that \problem{} is \textsf{FPT}.

However, since the form of $f(k)$ in Courcelle's Theorem can be multi-exponential, we provide a more concrete dynamic programming algorithm with runtime $k^{O(k)} \cdot n$ for practical purposes.

The algorithm employs dynamic programming on a tree decomposition of $G_f$. To obtain such a decomposition, we use a recent 2-approximation algorithm by Korhonen~\cite{Korhonen2021} to obtain a tree decomposition of $G_f$ with width $w \leq 2k+1 = O(k)$ in time $2^{O(k)} n$. If the algorithm determines that the treewidth exceeds $k$, we can immediately conclude no solution exists. We then convert this to a nice tree decomposition with four node types (leaf, introduce, forget, join) and $O(kn)$ nodes.

The dynamic programming follows the standard tree decomposition paradigm. For each node $x$ in the tree (processed bottom-up from leaves to root), we maintain a table $\mathcal{D}_x$ of \emph{realizable states}. Each state in $\mathcal{D}_x$ encodes a partial wire assignment and topological ordering for the symbols on the boundary bag $B_x$ of node $x$. Let $V_x$ denote all symbols appearing in the subtree rooted at $x$. Then, a state is \emph{realizable} and included in $\mathcal{D}_x$ if there exists a \textbf{compatible extension} to a ``valid circuit configuration'' using all symbols in $V_x$. Intuitively, the \emph{valid circuit configuration} for $V_x$ consists of:
\begin{itemize}
    \item \textbf{Wire assignments:} A partition of symbols in $V_x$ into $k$ ordered lists $L_1, \ldots, L_k$;
    \item \textbf{Edge type:} An edge type assignment $\rho$ that specifies, for each pair of consecutive symbols on the same list, their connectivity type (adjacent via a Hadamard gate on the wire, or disconnected);
\end{itemize}
such that they satisfy the acyclicity constraint of graph $\GH[V_x(\rho)]$ (the directed constraint graph in \cref{lem:structure-intro} which only reflects temporal ordering of Hadamard gates incident to symbols in $V_x$ under edge type $\rho$).

This definition captures exactly the structure needed to realize a circuit for the symbols in $V_x$: the lists determining wire assignments, the edge types encoding local connectivity, and the acyclicity of $\GH[V_x(\rho)]$ ensuring that Hadamard gates can be consistently ordered in time. At the root node $r$ where $V_r = \mathcal{S}$ (all symbols), a valid circuit configuration directly corresponds to a complete quantum circuit on $k$ qubits with polynomial $f$.

Now, a state in $\mathcal{D}_x$ has two components $(\mathcal{W}_x, \mathcal{O}_x)$ which encode partial information about boundary symbols in $B_x$:
\begin{enumerate}
        \item $\mathcal{W}_x = (W_1, \ldots, W_k, \tau)$: A partition of boundary symbols in $B_x$ into $k$ ordered lists $W_i$, each with endpoints $\localendpoint{s},\localendpoint{e}$, along with an edge type $\tau$ on consecutive entries of each $W_i$:
        \begin{itemize}
            \item $\tau(u,v) = 1$: real symbols $u,v$ are adjacent on the wire (connected by a Hadamard gate);
            \item $\tau(u,v) = 2$: the two entries are connected through forgotten real symbols in $V_x \backslash B_x$ (each consecutive real-symbol pair on the path is connected by a Hadamard gate);
            \item $\tau(u,v) = 0$: the two entries are not connected inside the current subtree (no Hadamard gate or forgotten symbols between them).
        \end{itemize}
        
        \item $\mathcal{O}_x$: A topological ordering on the boundary virtual vertices $\{S,E\}\cup\{L(u),R(u)\mid u\in B_x\}$, after the identifications induced by direct adjacencies, with $S$ first and $E$ last.
\end{enumerate}

As mentioned, a state $(\mathcal{W}_x, \mathcal{O}_x)$ is hence \emph{realizable}---and included in $\mathcal{D}_x$---if it can be \textbf{compatibly extended} to a \emph{valid circuit configuration} for $V_x$: there exist lists $L_1, \ldots, L_k$ partitioning $V_x$ and edge types $\rho\in\{0,1\}$ on consecutive entries in each $L_i$ such that (i) $W_i \subseteq L_i$ with order preserved, (ii) $\tau$ extends to $\rho$ (type 0 or 1 edges remain unchanged; type 2 edges become paths of type 1 edges via forgotten real symbols), (iii) the partial constraint graph $\GH[V_x(\rho)]$ is acyclic and $\mathcal{O}_x$ is the restriction of one of its topological orderings to the boundary virtual vertices (including $S,E$).

For each of the four node types in the nice tree decomposition, we provide explicit transition rules to compute $\mathcal{D}_x$ from the tables of its children. The crux of the correctness proof is showing that these local transitions preserve realizability: a state is in $\mathcal{D}_x$ if and only if it can be compatibly extended to a valid circuit configuration for $V_x$.

Finally, we extract the solution from the root. At the root node $r$ (which has an empty bag $B_r = \emptyset$), we check whether $\mathcal{D}_r$ is non-empty. If so, any state in $\mathcal{D}_r$ corresponds to a complete valid circuit on $k$ qubits; we trace back through the DP table to extract the circuit explicitly. If $\mathcal{D}_r = \emptyset$, no valid circuit exists.

The algorithm runs in time $k^{O(k)} \cdot n$, dominated by the DP step: the number of possible states per node is bounded by $k^{O(k)}$ (from counting ordered boundary wire structures, edge types, and boundary topological orderings), each node transition contributes only a $2^{O(k)}\poly(k)$ factor, and there are $O(kn)$ nodes in the nice tree decomposition.

The obtained \textsf{FPT} algorithm in fact has a time bound of the form $k^{6k + o(k)} \cdot n$, see \cref{thm:main_algorithm}. Therefore, its running time is faster than the XP algorithm $O(\binom{n}{k}^2 \textnormal{poly}(n)) \le O((en/k)^{2k} \textnormal{poly}(n))$ when $k \leq n^{1/4 - o(1)}$, making it more efficient and practical for polynomially small $k$.

\section{Preliminaries}
\label{sec:prelim}

We briefly recall the polynomial representation of quantum circuits as
introduced by Dawson \emph{et~al.}~\cite{dawson} and further extended by Montanaro~\cite{mont}. After the terminal normalization gates are fixed, a quantum circuit~$\mathcal{C}$ on
$w$ qubits with $h$ internal Hadamard gates and diagonal gates drawn from
\{\textsc{Z}, \textsc{CZ}, \textsc{CCZ}\} defines a degree-3 polynomial
$f_\mathcal{C} \in \mathbb{F}_2[\mathcal{S}]$, $|\mathcal{S}| = w + h$ such that
\begin{equation*}
  \bigl\langle 0^w \bigl| \mathcal{C}\bigr| 0^w \bigr\rangle
  = \frac{\gap(f_\mathcal{C})}{2^{w+h/2}}, \quad \text{where} \quad
  \gap(f_\mathcal{C}) = |\{x\in \mathbb{F}_2^{\mathcal{S}}: f_\mathcal{C}(x) = 0 \}| - |\{x\in \mathbb{F}_2^{\mathcal{S}}: f_\mathcal{C}(x) = 1 \}|.
\end{equation*}
All algebra takes place over $\F_2$. In the following section, we briefly recall the transformation from quantum circuits to associated polynomials following~\cite{mont}, and then we define the decision problem \problem{(f,k)} and introduce the main results of this manuscript.

\subsection{From circuits to associated polynomials}\label{subsec:extraction}
The main steps in the transformation from quantum circuits to associated polynomials are as follows:
\paragraph{Normalization step.}
Given any quantum circuit $\mathcal{C}$ whose gate set is
\(\{H,\textsc{Z},\textsc{CZ},\textsc{CCZ}\}\),
\emph{prepend and append} an $H$ gate to \emph{every} qubit wire.  

\paragraph{Symbol assignment.}
On every wire place a fresh Boolean
symbol at \emph{each open segment} between consecutive $H$ gates.
Thus the $i$-th qubit with
$h_i := \ell_i + 1$ Hadamard gates contributes $\ell_i$ symbols
\[ x_{i,1},x_{i,2},\dots,x_{i,\ell_i}\in\mathcal S .\]

\paragraph{Gate-to-monomial map.}
Let \(f_\mathcal{C}\in\F_2[\mathcal S]\) be the associated polynomial obtained by assigning one monomial to each non-Hadamard gate \emph{and} to each internal Hadamard gate.  More precisely:
\[
  f_\mathcal{C} \;=\;
  \sum_{\substack{\text{each } \textsc{Z},\,\textsc{CZ},\,\textsc{CCZ} \\ \text{gate in }\mathcal{C}}}
    \bigl(\text{product of symbols on which the gate acts}\bigr)
  \;+\;
  \sum_{\substack{\text{each internal }H\text{ gate}\\\text{between symbols }x_i\text{ and }x_{i+1}}}
    \bigl(x_i\,x_{i+1}\bigr).
\]
In other words:
\[
\begin{array}{cl}
  \textsc{Z}\text{ acting on segment }x 
    & \;\Longrightarrow\; x,\\[6pt]
  \textsc{CZ}\text{ acting on segments }(x,y) 
    & \;\Longrightarrow\; x\,y,\\[6pt]
  \textsc{CCZ}\text{ acting on segments }(x,y,z) 
    & \;\Longrightarrow\; x\,y\,z,\\[6pt]
  \text{interior }H\text{ gate between symbols }x_i\text{ and }x_{i+1}
    & \;\Longrightarrow\; x_i\,x_{i+1}.
\end{array}
\]
By construction, each non-terminal \(H\) gate contributes a quadratic term equal to the product of the two adjacent symbols it separates, and each diagonal gate ($Z$, CZ, CCZ) contributes exactly one monomial of degree \(1\), \(2\), or \(3\). Moreover, no constant terms are introduced in the associated polynomial. See Figure~\ref{fig:example-circuit} in \cref{subsec:tech-overview} for an example of a quantum circuit and its corresponding associated polynomial.

\subsection{Problem Definitions and Notation}
In the previous section, we saw that any quantum circuit can be represented by a degree-3 polynomial (with zero constant term) over the set of symbols \(\mathcal{S}\). In fact, Montanaro~\cite{mont} showed that the converse is also true: any degree-3 $\mathbb{F}_2$ polynomial with zero constant term is the associated polynomial of a quantum circuit. Therefore, for any such polynomial $f$, we can define the width of $f$ as in the following definition.
\begin{definition}[Width of a polynomial]
  \label{def:width}
  The \emph{width} of a degree-3 polynomial \(f\in \mathbb{F}_2[\mathcal{S}]\) with no constant term is the minimum number of qubits required for a quantum circuit to realize \(f\) as its associated polynomial. We denote the width of \(f\) by \(w(f)\).
\end{definition}
Now, we can define the decision version of computing $w(f)$, denoted by \problem{(f,k)}.
\begin{definition}[Decision version of \problem{}]
  \label{def:problem}
  Given a degree-3 polynomial $f \in\mathbb{F}_2[\mathcal{S}]$ without a constant term and an integer $k$, \problem{(f,k)} asks whether $w(f)\leq k$.
\end{definition}

In the algorithmic part, we will also consider the search version of \problem{}: 
\begin{definition}[Search version of \problem{}]\label{def:search-problem}
    Given a degree-3 polynomial $f \in \mathbb{F}_2[\mathcal{S}]$ without a constant term and an integer $k$, the search version of \problem{(f,k)} asks to either find a quantum circuit on at most $k$ qubits whose associated polynomial is equal to the given $f$, or conclude that no such circuit exists.
\end{definition}

It has been proven in~\cite{mont} that any degree-3 polynomial $f$ on $n$ symbols with no constant term can be realized by an IQP circuit (``Instantaneous Quantum Polynomial-time'')~\cite{Shepherd2009} on at most $n$ qubits, so $w(f) \leq n$. Therefore, both the decision and the search version of \problem{(f,k)} are trivial when $k \ge n$.

Moreover, if there exists a circuit using $k' < n$ qubits, then we can increase the number of qubits by splitting a qubit that carries more than one symbol: remove the second Hadamard gate on that qubit, create a new qubit initialized with the original leftmost symbol, add a CZ gate between the two segments to preserve the quadratic product, and keep the rest of the circuit unchanged. Thus, it suffices to consider instances with $k \leq n$ and to search for a circuit on exactly $k$ qubits.

Finally, we introduce some technical notation for quantum circuits that will be used in the analysis of the reduction in \cref{sec:hardness}. 
\begin{definition}[Support]\label{def:supp}
  Let \(\mathcal{C}\) be a quantum circuit on \(n\) qubits consisting of \(h\) Hadamard gates and \(\{Z,\textsc{CZ},\textsc{CCZ}\}\).  In the standard procedure for extracting the degree-3 associated polynomial from \(\mathcal{C}\), each qubit wire is partitioned by its Hadamard gates into a sequence of segments, and on each segment we record a symbol \(x\in \mathcal{S}\).  We then equip each wire with a normalized “time” axis, identifying the input of the wire with \(0\) and the output with \(1\).  Denote the entire wire‐time interval by \(I = (0,1)\).

  For any symbol \(x\) appearing on a given qubit wire, its \emph{support} is the open interval
  \[
    \supp(x) \;=\; (\,l_x, r_x\,)\;\subseteq\; I,
  \]
  where
  \begin{itemize}[nosep]
    \item \(l_x\) is the time coordinate of the left endpoint of the segment on which \(x\) appears, and
    \item \(r_x\) is the time coordinate of the right endpoint of that segment.
  \end{itemize}
  Thus \(\supp(x)\) precisely captures the contiguous portion of the qubit's time axis during which the symbol \(x\) is active.
\end{definition}

\begin{definition}[Overlap]\label{def:overlap}
  Let \(x\) and \(y\) be two symbols appearing on (not necessarily the same) qubit wires in the circuit \(\mathcal{C}\), with supports
  \[
    \supp(x) = (l_x, r_x),
    \quad
    \supp(y) = (l_y, r_y)
    \;\subseteq\; I = (0,1).
  \]
  We say that \(x\) and \(y\) \emph{overlap} if and only if their support intervals have non‐empty intersection:
  \[
    \supp(x)\,\cap\,\supp(y)\;\neq\;\emptyset,
    \quad\text{equivalently}\quad
    l_x < r_y \quad\text{and}\quad l_y < r_x.
  \]
\end{definition}

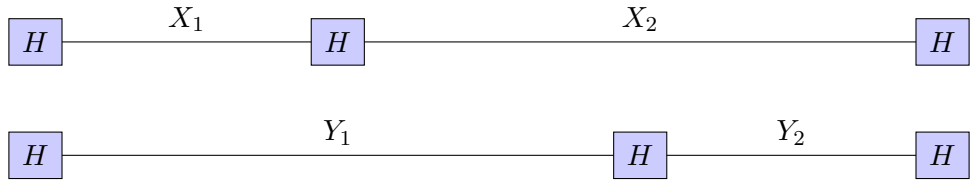
\begin{figure}[h]
  \centering
  \begin{tikzpicture}[xscale=1.5, yscale=1.5]
    
    \draw (0,1) -- (8,1);
    \draw (0,0) -- (8,0);

    \foreach \t in {0, 2.6664, 8} {
      \node[draw,fill=blue!20,inner sep=5pt] at (\t,1) {$H$};
    }
    
    \node at (1.3336, 1.2) {\(X_1\)};
    \node at (5.3336, 1.2) {\(X_2\)};

    \foreach \t in {0, 5.3336, 8} {
      \node[draw,fill=blue!20,inner sep=5pt] at (\t,0) {$H$};
    }
    
    \node at (2.6666, 0.2) {\(Y_1\)};
    \node at (6.6666, 0.2) {\(Y_2\)};
\end{tikzpicture}
  \caption{
    Segments labeled $X_1,X_2$ on the first qubit and $Y_1,Y_2$ on the second qubit with supports
    $\supp(X_1)=(0,\tfrac13)$, $\supp(X_2)=(\tfrac13,1)$,
    $\supp(Y_1)=(0,\tfrac23)$, $\supp(Y_2)=(\tfrac23,1)$.
    Overlapping pairs are $(X_1,Y_1)$, $(X_2,Y_1)$, and $(X_2,Y_2)$, but $\supp(X_1)\cap\supp(Y_2)=\varnothing$.
  }
\end{figure}

To simplify arguments concerning intersections of supports of multiple symbols (which are open intervals on the time axis), we recall a classical fact as follows.

\begin{theorem}[Helly's theorem for intervals~\cite{Helly1923}]\label{prelim:helly}
Let $I_1,I_2,I_3\subseteq \mathbb{R}$ be open intervals. Then
\[ I_1\cap I_2\cap I_3 \neq \varnothing \quad\Longleftrightarrow\quad I_i\cap I_j \neq \varnothing\ \text{for all }1\le i<j\le 3. \]
More generally, any family of intervals on the real line with the property that every two members intersect has non-empty total intersection (the Helly number for intervals is 2).
\end{theorem}

\subsection{Complexity Preliminaries}
\subsubsection{Classical Complexity}
In this paper, we will prove $\mathsf{NP}$-completeness of \problem{} by a reduction from the decision version of \SAT{}, which we list as follows.

\begin{definition}[3-SAT]
  An instance of \SAT{} consists of a Boolean formula
  $\Phi = c_1 \land \dotsb \land c_m$ in conjunctive normal form, where
  each clause contains exactly three literals.  The task is to decide
  whether $\Phi$ is satisfiable.
\end{definition}

To establish the approximation hardness of \problem{}, we will use the following celebrated result of Håstad~\cite[Theorem 6.5]{Hastard}.
\begin{theorem}[Håstad's 7/8 Theorem~\cite{Hastard}]
  \label{thm:hastad}
  For any \(\epsilon > 0\), it is $\mathsf{NP}$-hard to distinguish between the following two cases, given a \textnormal{3-SAT} formula \(\Phi\) as input:
  \begin{itemize}[nosep]
    \item \(\Phi\) is satisfiable, and
    \item At most \((7/8 + \epsilon)\) fraction of the clauses in \(\Phi\) can be satisfied.
  \end{itemize}
\end{theorem}

We now state the Exponential Time Hypothesis (ETH)~\cite{ETH0} and a variant via the sparsification lemma~\cite{ETH}, which will be used in our complexity analysis.

\begin{conjecture}[Exponential Time Hypothesis (ETH)~\cite{ETH0}]
  \label{conj:eth}
  There exists a constant \(\epsilon  > 0\) such that no algorithm can decide \SAT{} with \(n\) variables and \(m\) clauses in time \(O(2^{\epsilon n})\).
\end{conjecture}

The ETH implies that solving \SAT{} in subexponential time is unlikely. A useful ingredient of ETH is the sparsification lemma, which states that any \SAT{} instance can be expressed as a disjunction of sparse \SAT{} instances. As a result, \SAT{} instances with only a linear number of clauses in \(n\) remain hard in subexponential time.

\begin{theorem}[Sparsification Lemma~\cite{ETH}]
  \label{thm:sparsification}
  For every \(\epsilon > 0\), there exists a constant \(c_\epsilon > 0\) such that any \SAT{} formula \(\Phi\) with \(n\) variables can be expressed as a disjunction of at most \(2^{\epsilon n}\) \SAT{} formulas efficiently (in $2^{\epsilon n} \cdot \textnormal{poly}(n)$ time), each with at most \(c_\epsilon n\) clauses. Furthermore, there exist constants \(D > 0\) and \(\epsilon > 0\) such that there is no algorithm that solves \SAT{} with \(n\) variables and \(m \leq Dn\) clauses in time \(2^{\epsilon n}\), unless the Exponential Time Hypothesis (ETH) fails.
\end{theorem}

\subsubsection{Parameterized Complexity}

Beyond classical complexity, we also analyze \problem{} from the perspective of parameterized complexity~\cite{Downey2013, Cygan2015}. In parameterized complexity, we measure the computational complexity of a problem not only in terms of the input size $n$ but also in terms of a \emph{parameter} $k$ that captures structural properties of the instance. For our problem \problem{(f,k)}, the natural parameter is the width $k$.

\begin{definition}[Fixed-Parameter Tractable ($\mathsf{FPT}$)]
  \label{def:fpt}
  A parameterized problem with parameter $k$ is \emph{fixed-parameter tractable} ($\mathsf{FPT}$) if there exists an algorithm that solves it in time $f(k) \cdot \textnormal{poly}(n)$, where $n$ is the input size and $f$ is a computable function depending only on $k$.
\end{definition}

The class $\mathsf{FPT}$ is considered the tractable class in parameterized complexity. Furthermore, (potentially) harder problems are classified into a hierarchy of complexity classes $\mathsf{FPT} = \textsf{W}[0]  \subseteq \textsf{W}[1] \subseteq \textsf{W}[2] \subseteq \cdots$, analogous to the polynomial hierarchy in classical complexity ($\mathsf{P} \subseteq \mathsf{NP} \subseteq \cdots$). A problem in $\textsf{W}[1]$ that is not $\mathsf{FPT}$ is believed to require time $n^{f(k)}$ for some unbounded function $f$. Such algorithms, where the exponent depends on $k$, are called XP algorithms. In this paper, we will not go into the formal definitions of these classes. For formal definitions and the theory of parameterized complexity, we refer to~\cite{Downey2013, Cygan2015}.

\subsection{Algorithm Preliminaries}
\label{subsec:alg_prelim}

In this subsection, we introduce key concepts from graph theory and parameterized complexity that will be used in our algorithmic results.

\subsubsection{Tree Decomposition and Treewidth}

We introduce the concepts of Tree Decomposition and Path Decomposition~\cite{Robertson1983, Bodlaender1994}, which provide a useful theoretical framework for understanding the algorithm's structure.

\begin{definition}[Tree Decomposition and Path Decomposition]
\label{def:tree_decomposition}
A \emph{tree decomposition} of a graph $G = (V, E)$ is a pair $(T, \{B_t\}_{t \in V(T)})$, where $T$ is a tree and each $B_t$ is a subset of $V$, called a \emph{bag}, satisfying the following properties:
\begin{enumerate}
    \item \textbf{Vertex Cover}: Every vertex of $G$ belongs to at least one bag $B_t$. $\bigcup_{t \in V(T)} B_t = V$.
    \item \textbf{Edge Cover}: For every edge $(u, v) \in E$, there exists a bag $B_t$ containing both $u$ and $v$.
    \item \textbf{Running Intersection}: For any vertex $v \in V$, the set of nodes $\{t \in V(T) \mid v \in B_t\}$ forms a connected subtree of $T$. 
\end{enumerate}
The \emph{width} of a tree decomposition is defined as $\max_{t \in V(T)} |B_t| - 1$. The \emph{treewidth} of $G$, denoted $\textnormal{tw}(G)$, is the minimum width over all possible tree decompositions of $G$.

A \emph{path decomposition} is a special case of tree decomposition where the tree $T$ is a simple path (a sequence of nodes). The \emph{pathwidth} of $G$, denoted $\textnormal{pw}(G)$, is the minimum width over all path decompositions of $G$.
\end{definition}

To facilitate dynamic programming, it is convenient to work with a special type of tree decomposition called a \emph{nice tree decomposition}.

\begin{definition}[Nice Tree Decomposition~\cite{Kloks1994}]\label{def:nice_td}
A tree decomposition $(T, \{B_t\}_{t \in V(T)})$ is called \emph{nice} if $T$ is a rooted tree and every node $t \in V(T)$ is one of the following four types:
\begin{itemize}
    \item \textbf{Leaf node:} $t$ is a leaf of $T$ (it has no children) and $|B_t| = 0$.
    \item \textbf{Introduce node:} $t$ has exactly one child $t'$ and $B_t = B_{t'} \cup \{v\}$ for some vertex $v \notin B_{t'}$.
    \item \textbf{Forget node:} $t$ has exactly one child $t'$ and $B_t = B_{t'} \setminus \{v\}$ for some vertex $v \in B_{t'}$.
    \item \textbf{Join node:} $t$ has exactly two children $t_1, t_2$ and $B_t = B_{t_1} = B_{t_2}$.
\end{itemize}
\end{definition}

It is well-known that any tree decomposition can be transformed into a nice one efficiently.

\begin{lemma}[\cite{Kloks1994, Cygan2015}]\label{lem:nice_td}
Given a tree decomposition of a graph $G$ of width $w$ with $N$ nodes, one can compute a nice tree decomposition of $G$ of width $w$ and with $O(wN)$ nodes in time $O(w^2N)$. In particular, for an $n$-vertex graph, a nice tree decomposition of width $w$ with $O(wn)$ nodes can be found in time $O(w^2 n)$.
\end{lemma}

We also state the crucial property of the nice tree decomposition that justifies the correctness of the dynamic programming algorithm. This is a special case of~\cite[Lemma 7.3]{Cygan2015}. We include the proof in \cref{app:bag_subtree} for completeness.
\begin{lemma}\label{lem:bag_subtree}
Let $(T, \{ B_t \}_{t\in V(T)})$ be a rooted tree decomposition of a graph $G$ and let $x$ be a node of $T$ with children $y_1, y_2, \dots, y_m$. For each $i$, let
$
    V_{y_i}:=\bigcup_{z\in V(T_{y_i})} B_z,
$
where $T_{y_i}$ is the subtree of $T$ rooted at $y_i$. Then there is no edge of $G$ between $V_{y_i}\setminus B_x$ and $V(G)\setminus V_{y_i}$. Consequently, if $i\neq j$, then there is no edge of $G$ between $V_{y_i}\setminus B_x$ and $V_{y_j}\setminus B_x$.
\end{lemma}
\subsubsection{Topological Ordering and Graph Acyclicity}

Throughout the algorithmic sections, we frequently work with directed graphs and their topological orderings. We recall these concepts here.

\begin{definition}[Topological Ordering]
\label{def:topological_ordering}
Let $G = (V, E)$ be a directed graph. A \emph{topological ordering} of $G$ is a total ordering $<$ on $V$ such that for every directed edge $(u, v) \in E$, we have $u < v$. A directed graph admits a topological ordering if and only if it is acyclic (i.e., it is a directed acyclic graph, or DAG)~\cite{Kahn1962}.
\end{definition}

\begin{definition}[Consistent Orderings]
\label{def:consistent_ordering}
Let $G = (V, E)$ be a directed acyclic graph and $V_0 \subseteq V$ be a subset of vertices. An ordering $\mathcal{O}_0$ on $V_0$ is said to be \emph{consistent with $G$} if for any two vertices $u, v \in V_0$ connected by a directed path in $G$, we have $u < v$ in $\mathcal{O}_0$.
\end{definition}

The following fundamental property of topological orderings will be used repeatedly in our dynamic programming algorithm, whose proof is in \cref{app:acyclic_union}.

\begin{lemma}\label{lem:acyclic_union}
 Let $G_1, G_2$ be two directed acyclic graphs on vertices $V_1, V_2$ with intersection $V_0 = V_1 \cap V_2$.  If there is a total order $\mathcal O_0$ on $V_0$ that is consistent with both $G_1$ and $G_2$ in the sense of \cref{def:consistent_ordering}, then $G_1\cup G_2$ is acyclic and admits a topological ordering extending $\mathcal O_0$. Conversely, if $G_1\cup G_2$ is acyclic, then the restriction of any topological ordering of $G_1\cup G_2$ to $V_0$ is consistent with both $G_1$ and $G_2$.
\end{lemma}

\section{$\mathsf{NP}$-hardness of \problem{}}\label{sec:hardness}
\subsection{Reduction from \SAT{}}
\label{sec:reduction}
In this section, we present a polynomial-time reduction from \SAT{} to \problem{}. Given a \SAT{} instance $\Phi$ with $n$ variables and $m$ clauses, we efficiently
construct a polynomial $f_{\Phi}$ of degree at most~3 with no constant term and set a qubit
budget $k = 1 + 6m$ such that
\begin{equation*}
  \Phi \text{ satisfiable } \Longleftrightarrow  w(f_{\Phi}) \le k \Longleftrightarrow 
  f_{\Phi} \text{ is implementable on at most $k$ qubits.}
\end{equation*}

\subsubsection{Overall Strategy}\label{subsec:overall-strategy}

Our reduction maps a 3-SAT instance \((X,C)\)—with variables
\(X=\{x_1,\dots,x_{|X|}\}\) (we do not include $\bar{x}$ in $X$ for each $x\in X$) and clauses
\(C=\{\cl{c_1},\dots,\cl{c_{|C|}}\}\)—to a degree-3 polynomial
\(f\in\mathbb F_2[\mathcal{S}]\) such that \emph{every} quantum circuit realizing \(f\) with the minimum possible number of qubits must represent a satisfying assignment of the original formula.  

\paragraph{Notation: variables vs.\ literals.}
We write $x\in X$ for Boolean variables of the \SAT{} instance and $\lit{x}$ for a literal (either $x$ or $\bar x$). For an occurrence of a literal whose underlying variable is $x$, the polynomial/circuit uses the three occurrence-symbols $x_{\cl{c}},\hat{x}_{\cl{c}},\bar{x}_{\cl{c}}$; the bar is part of the symbol name. The truth value is encoded geometrically: if the literal is $x$, then the literal is \bool{true} exactly when $x_{\cl{c}}$ overlaps $t_{\cl{c}}$; if the literal is $\bar x$, then it is \bool{true} exactly when $\bar{x}_{\cl{c}}$ overlaps $t_{\cl{c}}$. See \cref{def:encode} for a formal definition.

The polynomial representation of an arbitrary quantum circuit is hard to control. Our idea is to carefully design the polynomial and the qubit budget $k=1+6|C|$ to show that any size-optimal quantum circuit realizing $f_{\Phi}$ must satisfy some additional constraints. Under these conditions, any quantum circuit for $f_{\Phi}$ corresponds to a witness for $\Phi$. Two main principles are listed below.
\begin{enumerate}[(i)]
    \item \textbf{Various qubit wire roles.}
    
    We require various qubit wires in the circuit to assume distinct functional roles:
      \begin{itemize}
      \item \emph{Base qubit wire.}
            A single qubit wire \(b\) contains \emph{no} \(H\) gates and is not in any CZ gate. $b$ is only in a few CCZ gates to ensure that certain symbols are in distinct qubit wires but on the same vertical link. We can always assume that $\supp(b)$ is the whole interval $I$; see \cref{lem:qubit-iso}.
      \item \emph{Truth-marker qubit wires.}
            Exactly two \(H\) gates split the qubit wire into three segments \(t_{\cl{c}}\), \(\hat{t}_{\cl{c}}\), and \(\tilde{t}_{\cl{c}}\) in order for each clause $\cl{c}$.
      \item \emph{Variable qubit wires.}
            For each literal occurrence, a variable wire contains \(x_{\cl{c}}\), \(\hat{x}_{\cl{c}}\), and \(\bar{x}_{\cl{c}}\) (in either order).
      \item \emph{Clause gadgets via nested ORs.}
            Each clause \(\cl{c}=(\lit{x}\lor \lit{y}\lor \lit{z})\) is evaluated with two extra
            qubit wires: the first has middle segment \(p_{\cl{c}}\), representing the
            intermediate Boolean value \(\lit{x}\lor \lit{y}\), and the second has middle
            segment \(o_{\cl{c}}\), representing the final Boolean value
            \(p_{\cl{c}}\lor \lit{z}\) of the whole clause. Both qubit wires bear two \(H\) gates and three
            segments, arranged as \(\{\overleftarrow{q},q,\overrightarrow{q}\}\) or \(\{\overrightarrow{q},q,\overleftarrow{q}\}\), \(q\in\{p_{\cl{c}},o_{\cl{c}}\}\).
            The decorations \(\overrightarrow{\bullet}\) and \(\overleftarrow{\bullet}\) are only notational placeholders to squeeze \(q\) into the correct geometric position.
            We construct clause gadgets (see \cref{subsec:clause-gadget}) with the following soundness/completeness behavior: if \(q\) overlaps with \(t_{\cl{c}}\) in a \(k\)-qubit circuit, then at least one input literal is \bool{true} (i.e., its segment overlaps $t_{\cl{c}}$); conversely, if at least one input literal is \bool{true}, then the gadget can be realized with \(q\) overlapping \(t_{\cl{c}}\).
            Finally, we add a cubic monomial
            \(b\cdot t_{\cl{c}}\cdot o_{\cl{c}}\) for every clause, forcing the segment
            \(o_{\cl{c}}\) to overlap with \(t_{\cl{c}}\) and therefore enforcing global
            satisfiability.
            \end{itemize}
      \item \textbf{Size-optimality guarantee.}
      In the construction, we guarantee that every qubit wire in any realizing quantum circuit for \(f\) has at most two \(H\) gates, and \(b\) is \(H\)-free. Then, the circuit must use at least \(k\) qubits; see also \cref{lem:good-cl}. Moreover, by further exploiting the structure of the polynomial, we are able to prove that size equality is attained \emph{if and only if} each qubit wire has exactly the configuration above. This yields a bijection between witnesses of the \SAT{} instance and size-optimal quantum circuits for the polynomial \(f_{\Phi}\).
\end{enumerate}

After discussing the overview strategy, we introduce the construction of polynomials formally in the following sections.  Throughout the construction, sums over pairs of clauses or literal occurrences are over unordered pairs.  When a clause contains repeated literals, the corresponding occurrence symbols should be understood as carrying suppressed occurrence indices, so that distinct occurrences give distinct algebraic symbols.

\subsubsection{Truth-Marker Qubit Wires}\label{subsec:marker-gadget}

In the global polynomial, the only quadratic terms that mention the
symbols \(t_{\cl{c}}\), \(\hat{t}_{\cl{c}}\), or \(\tilde{t}_{\cl{c}}\) are (TM stands for \emph{truth-marker}):
\[
   G_{\mathrm{TM}_{\cl{c}}} \;=\,\hat{t}_{\cl{c}}\ t_{\cl{c}}\;+\;\hat{t}_{\cl{c}}\,\tilde{t}_{\cl{c}} .
\]
Consequently, no other symbol can link any of
\(t_{\cl{c}},\hat{t}_{\cl{c}},\tilde{t}_{\cl{c}}\) through Hadamard gates.  If these three symbols were placed on two or more
qubit wires, at least one additional qubit wire would be consumed.  Hence, every \(k\)-qubit circuit must host them
on a \emph{single} qubit wire intersected by exactly two Hadamard gates.

To ensure each truth-marker qubit wire encodes the same Boolean value, we include the following terms in our polynomial $f_{\Phi}$:
\[
G_{\mathrm{TM}} = b\sum_{\{\cl{c_1},\cl{c_2}\}\in \binom{C}{2}} \left( t_{\cl{c_1}}t_{\cl{c_2}} + \hat{t}_{\cl{c_1}}\hat{t}_{\cl{c_2}} + \tilde{t}_{\cl{c_1}}\tilde{t}_{\cl{c_2}}\right).
\]

Since we can always reverse the circuit and variable orderings without changing the represented polynomial, without loss of generality we fix the left-to-right temporal order
\[
   t_{\cl{c}} \;\;\hat{t}_{\cl{c}} \;\;\tilde{t}_{\cl{c}} ,
\]
and we refer to these qubit wires as the \emph{truth-marker qubit wires}.  All later
gadgets assume this canonical ordering.

\subsubsection{Variable Qubit Wires}\label{subsec:var-gadget}

For every Boolean variable $x\in X$ and clause $\cl{c}$ that contains $\lit{x}$ (that is, $\cl{c}$ contains the literal $x$ or $\bar{x}$), we construct symbols denoted $\{x_{\cl{c}},\hat{x}_{\cl{c}},\bar{x}_{\cl{c}}\}$ and append to the global
polynomial \(f_{\Phi}\) the \emph{variable gadget}
\[
   G_{x_{\cl{c}}} \;=\;
   b\,\hat{t}_{\cl{c}}\,(x_{\cl{c}}+\hat{x}_{\cl{c}}+\bar{x}_{\cl{c}})\;+\;
   \hat{x}_{\cl{c}}\,(x_{\cl{c}}+\bar{x}_{\cl{c}}).
\]
To ensure each variable qubit wire for the same variable $x$ encodes the same Boolean value, let $\operatorname{Occ}(x)$ be the multiset of literal occurrences whose underlying Boolean variable is $x$, and include the following terms in our polynomial $f_{\Phi}$:
\[
G_x=b\sum_{\{(\ell_1,c_1),(\ell_2,c_2)\}\in\binom{\operatorname{Occ}(x)}{2}} \left( x_{\cl{c_1}}x_{\cl{c_2}} + \hat{x}_{\cl{c_1}}\hat{x}_{\cl{c_2}} + \bar{x}_{\cl{c_1}}\bar{x}_{\cl{c_2}}\right).
\]
The summation is over unordered pairs of distinct \emph{literal occurrences} of $x$ or $\bar{x}$ across all clauses. We do \emph{not} require $c_1\neq c_2$: if a clause contains multiple literals based on $x$, we also include the corresponding intra-clause synchronization terms. We only exclude pairing a literal occurrence with itself; when a clause contains two copies of $x$ (e.g., $x\lor x\lor z$), we treat them as distinct occurrences but keep the same notation $x_{\cl{c}}$ for simplicity.
\paragraph{Intended semantics for variable qubits.}
\begin{enumerate}[label=(\alph*)]
\item  In a \(k\)-qubit circuit, size-optimality forces \(x_{\cl{c}},\hat{x}_{\cl{c}},\bar{x}_{\cl{c}}\) to lie in the \emph{same} qubit wire. Now, the terms \(\hat{x}_{\cl{c}}\,x_{\cl{c}}\) and \(\hat{x}_{\cl{c}}\,\bar{x}_{\cl{c}}\) jointly force \(x_{\cl{c}}\)
      and \(\bar{x}_{\cl{c}}\) to lie on the two edges of the middle segment
      \(\hat{x}_{\cl{c}}\). To summarize, two Hadamard gates cut that qubit wire into the three segments \(\{x_{\cl{c}},\hat{x}_{\cl{c}},\bar{x}_{\cl{c}}\}\) (or the reversed ordering). 
\item The terms \(b\ \hat{t}_{\cl{c}}\ (\cdot)\) enforce that \(\hat{t}_{\cl{c}}\) must overlap \emph{each} of the symbols \(\{x_{\cl{c}},\hat{x}_{\cl{c}},\bar{x}_{\cl{c}}\}\), hence the $t_{\cl{c}}$ segment in the truth-marker qubit wire overlaps exactly one of $x_{\cl{c}}$ and $\bar{x}_{\cl{c}}$ vertically in a \(k\)-qubit circuit.
\end{enumerate}

\begin{figure}[h]
  \centering
  \begin{tikzpicture}[xscale=1.5,yscale=1.5]
    \foreach \y in {0,1,2} {\draw (0,\y) -- (8,\y);}
    \foreach \t in {0,8}{
        \node[draw,fill=blue!20,inner sep=5pt] at (\t,0) {$H$};
    }
    \node at (4,-0.25) {$b$};
    \foreach \t in {0,1.6,6.4,8}{
        \node[draw,fill=blue!20,inner sep=5pt] at (\t,1) {$H$};
    }
    \node at (1,1.25) {$t_{\cl{c}}$};
    \node at (4,1.25) {$\hat{t}_{\cl{c}}$};
    \node at (7,1.25) {$\tilde{t}_{\cl{c}}$};
    \foreach \t in {0,2.8,5.2,8}{
        \node[draw,fill=blue!20,inner sep=5pt] at (\t,2) {$H$};
    }
    \node at (1.2,2.25) {$x_{\cl{c}}$};
    \node at (4,2.25) {$\hat{x}_{\cl{c}}$};
    \node at (6.8,2.25) {$\bar{x}_{\cl{c}}$};
    \draw[cczlink] (2.24,0) -- (2.24,2);
    \foreach \y in {0,1,2}{\fill[cczdot] (2.24,\y) circle;}
    \draw[cczlink] (4.4,0) -- (4.4,2);
    \foreach \y in {0,1,2}{\fill[cczdot] (4.4,\y) circle;}
    \draw[cczlink] (5.76,0) -- (5.76,2);
    \foreach \y in {0,1,2}{\fill[cczdot] (5.76,\y) circle;}
\end{tikzpicture}
  \caption{Variable gadget.  Supports: $b=(0,1)$; $t_{\cl{c}}=(0,0.25)$, $\hat{t}_{\cl{c}}=(0.25,0.75)$, $\tilde{t}_{\cl{c}}=(0.75,1)$; $x_{\cl{c}}=(0,0.3)$, $\hat{x}_{\cl{c}}=(0.3,0.7)$, $\bar{x}_{\cl{c}}=(0.7,1)$.  CCZ columns enforce the cubic terms $b\,\hat{t}_{\cl{c}}\,x_{\cl{c}}$, $b\,\hat{t}_{\cl{c}}\,\hat{x}_{\cl{c}}$, and $b\,\hat{t}_{\cl{c}}\,\bar{x}_{\cl{c}}$.}
\end{figure}
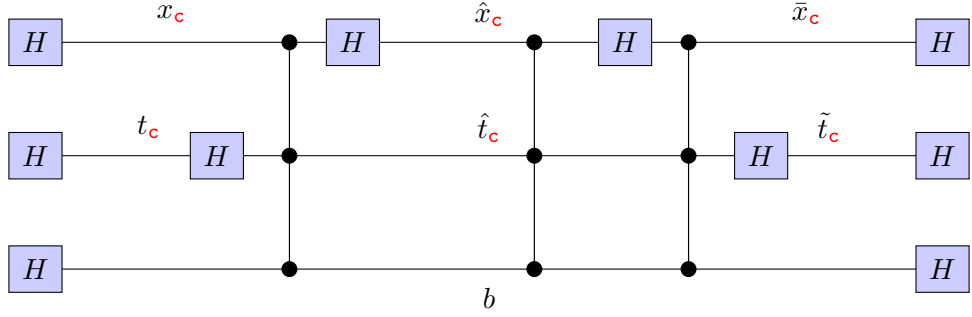

\subsubsection{Clause Gadget}\label{subsec:clause-gadget}

A three-literal clause
\[
    \cl{c} \;=\; (\lit{x}\;\lor\;\lit{y}\;\lor\;\lit{z})
\]
is represented as the \emph{nesting of two binary disjunctions}
\[
   p_{\cl{c}} \;=\; \lit{x}\lor \lit{y},
   \qquad
   o_{\cl{c}} \;=\; p_{\cl{c}}\lor \lit{z} .
\]
Thus \(p_{\cl{c}}\) is the intermediate truth value of the first two literals,
whereas \(o_{\cl{c}}\) is the output truth value of the whole clause.
The construction therefore introduces \emph{two} additional qubits and six new algebraic symbols that appear only in the present gadget:
\[
   \overleftarrow{p_{\cl{c}}},\,p_{\cl{c}},\,\overrightarrow{p_{\cl{c}}}
   \quad\text{and}\quad
   \overleftarrow{o_{\cl{c}}},\,o_{\cl{c}},\,\overrightarrow{o_{\cl{c}}}.
\]

\paragraph{Binary OR Gadgets.}
For any ordered pair \((u,v)\) of input occurrence-symbols (either $p_{\cl{c}}$ or one of the literal-occurrence symbols of the clause), we append the following polynomial to the global polynomial $f_{\Phi}$:
\[
   G_{u\lor v}
   \;=\;
   b\,\overleftarrow{q}\,u
   \;+\;
      b\,q\,u
   \;+\;
   b\,q\,v
   \;+\;
         b\,\overrightarrow{q}\,v
   \;+\;
   \overleftarrow{q}\,q
   \;+\;
   \overrightarrow{q}\,q ,
\]
where \(q\in\{p_{\cl{c}},o_{\cl{c}}\}\). One can check that, after restoring any suppressed occurrence indices, the monomials of $G_{u\lor v}$ are distinct from each other and from previously appended ones.

\paragraph{Intended semantics for Binary OR Gadgets.}
\begin{enumerate}[label=(\alph*)]
\item \emph{One-qubit placement.}  
      The same reason as in \cref{subsec:var-gadget} shows that the quadratic terms \(\overleftarrow{q}\,q\) and \(\overrightarrow{q}\,q\) force
      \(\overleftarrow{q},q,\overrightarrow{q}\) to lie on a single qubit wire intersected by two \(H\) gates in any \(k\)-qubit circuit.
\item \emph{Bounding window via \(b\).}  
      Cubic monomials $b\,\overleftarrow{q}\,u$ and $b\,q\,u$ force the junction point of $\overleftarrow{q}$ and $q$ to lie in $\supp(u)$; recall the definition of $\supp(\cdot)$ in \cref{def:supp}. Similarly, the junction point of $\overrightarrow{q}$ and $q$ must lie in $\supp(v)$. Consequently, $\supp(q)$ must be contained in the
  interval spanned by \(\supp(u)\cup\supp(v)\). This yields the soundness direction: if $\supp(q)$ overlaps $\supp(t_{\cl{c}})$, then at least one of $\supp(u)$ and $\supp(v)$ also overlaps $\supp(t_{\cl{c}})$. Conversely, if at least one of $\supp(u)$ and $\supp(v)$ overlaps $\supp(t_{\cl{c}})$, the junction points can be chosen so that $\supp(q)$ overlaps $\supp(t_{\cl{c}})$.\footnote{A subtlety is that $\supp(q)$ may intersect with $\supp(\tilde{t}_{\cl{c}})$. However, this is not a problem, since we should think of a variable to be \bool{true} if and only if its segment overlaps with $\supp(t_{\cl{c}})$, and \bool{false} otherwise. We do not care about $\supp(\tilde{t}_{\cl{c}})$.} In other words, the gadget can be placed with $q$ \bool{true} whenever at least one of $u,v$ is supposed to be \bool{true} in the assignment, and any placement with $q$ \bool{true} certifies that at least one input is \bool{true}.
\end{enumerate} 

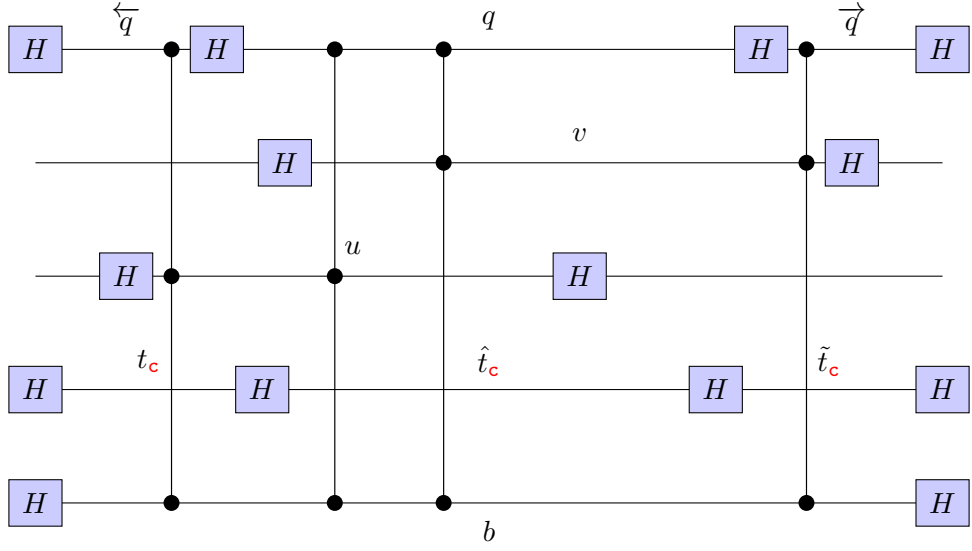
\begin{figure}[!htbp]
  \centering
  \begin{tikzpicture}[xscale=1.5, yscale=1.5]
    \foreach \y in {0,1,2,3,4} {
        \draw (0,\y) -- (8,\y);
    }

    \foreach \t in {0,8} {
        \node[draw,fill=blue!20,inner sep=5pt] at (\t,0) {$H$};
    }
    \node at (4,-0.25) {$b$};

    \foreach \t in {0,2,6,8} {
        \node[draw,fill=blue!20,inner sep=5pt] at (\t,1) {$H$};
    }
    \node at (1,1.25)   {$t_{\cl{c}}$};
    \node at (4,1.25)   {$\hat{t}_{\cl{c}}$};
    \node at (7,1.25)   {$\tilde{t}_{\cl{c}}$};

    \foreach \t in {0.8,4.8} {
        \node[draw,fill=blue!20,inner sep=5pt] at (\t,2) {$H$};
    }
    \node at (2.8,2.25) {$u$};

    \foreach \t in {2.2,7.2} {
        \node[draw,fill=blue!20,inner sep=5pt] at (\t,3) {$H$};
    }
    \node at (4.8,3.25) {$v$};

    \foreach \t in {0,1.6,6.4,8} {
        \node[draw,fill=blue!20,inner sep=5pt] at (\t,4) {$H$};
    }
    \node at (0.8,4.25)   {$\overleftarrow{q}$};
    \node at (4,4.25)    {$q$};
    \node at (7.2,4.25)  {$\overrightarrow{q}$};

    \draw[cczlink] (1.2,0) -- (1.2,4);
    \foreach \y in {0,2,4} {\fill[cczdot] (1.2,\y) circle;}
    \draw[cczlink] (2.64,0) -- (2.64,4);
    \foreach \y in {0,2,4} {\fill[cczdot] (2.64,\y) circle;}
    \draw[cczlink] (3.6,0) -- (3.6,4);
    \foreach \y in {0,3,4} {\fill[cczdot] (3.6,\y) circle;}
    \draw[cczlink] (6.8,0) -- (6.8,4);
    \foreach \y in {0,3,4} {\fill[cczdot] (6.8,\y) circle;}
\end{tikzpicture}
  \caption{OR gadget. Supports: $b=(0,1)$; $t_{\cl{c}}=(0,\tfrac14)$, $\hat{t}_{\cl{c}}=(\tfrac14,\tfrac34)$, $\tilde{t}_{\cl{c}}=(\tfrac34,1)$; $u=(0.1,0.6)$; $v=(0.3,0.9)$; $\protect\overleftarrow{q}=(0,0.2)$, $q=(0.2,0.8)$, $\protect\overrightarrow{q}=(0.8,1)$.  Columns mark CCZ gates on $(b,u,\protect\overleftarrow{q})$, $(b,u,q)$, $(b,v,q)$, and $(b,v,\protect\overrightarrow{q})$.}
\end{figure}

\paragraph{Full clause gadget.}
For the clause \(\cl{c}=(\lit{x}\lor \lit{y}\lor \lit{z})\), we concatenate two OR gadgets
and add one final synchronizing term:
\[
   G_{\cl{c}}
   \;=\;
   G_{\lit{x}_{\cl{c}}\lor \lit{y}_{\cl{c}}}
   \;+\;
   G_{p_{\cl{c}}\lor \lit{z}_{\cl{c}}}
   \;+\;
   b\,t_{\cl{c}}\,o_{\cl{c}} .
\]
The last cubic monomial pins the segment \(o_{\cl{c}}\) to overlap the truth
marker \(t_{\cl{c}}\), thereby \emph{forcing} the clause to evaluate to
\bool{true} in every \(k\)-qubit circuit.

\subsubsection{Global Polynomial}\label{subsubsec:global-poly}
Summarizing the above constructions from \cref{subsec:marker-gadget} to \cref{subsec:clause-gadget}, we have constructed the global polynomial as follows:
\[
\begin{aligned}
f_{\Phi} \;&=G_{\mathrm{TM}} + \sum_{x\in X}G_x + \sum_{\cl{c}=(\lit{x}\lor \lit{y}\lor \lit{z})\in C}\bigl(G_{\mathrm{TM}_{\cl{c}}} + G_{x_{\cl{c}}} + G_{y_{\cl{c}}} + G_{z_{\cl{c}}} + G_{\cl{c}} \bigr)
\\
&\;=\sum_{\{\cl{c_1},\cl{c_2}\}\in \binom{C}{2}} \Bigl(b\,t_{\cl{c_1}}t_{\cl{c_2}} + b\,\hat{t}_{\cl{c_1}}\hat{t}_{\cl{c_2}} + b\,\tilde{t}_{\cl{c_1}}\tilde{t}_{\cl{c_2}}\Bigr)
\\[8pt]
&+ \sum_{x\in X}\!\ \sum_{\{(\ell_1,c_1),(\ell_2,c_2)\}\in\binom{\operatorname{Occ}(x)}{2}}\Bigl(b\,x_{\cl{c_1}}x_{\cl{c_2}} + b\,\hat{x}_{\cl{c_1}}\hat{x}_{\cl{c_2}} + b\,\bar{x}_{\cl{c_1}}\bar{x}_{\cl{c_2}}\Bigr)
\\
&+ \sum_{\cl{c}\in C}\!\ \Bigl(t_{\cl{c}}\hat{t}_{\cl{c}} + \hat{t}_{\cl{c}}\,\tilde{t}_{\cl{c}}\Bigr)
\\
&+ \sum_{\cl{c}\in C}\!\ \Bigl(\sum_{\substack{\text{literal occurrences}\\ \lit{x}\text{ in }\cl{c}}}\!\bigl(
      b\,\hat{t}_{\cl{c}}\,{x_{\cl{c}}} + b\,\hat{t}_{\cl{c}}\,\hat{x}_{\cl{c}}+b\,\hat{t}_{\cl{c}}\,\bar{x}_{\cl{c}}
      + \hat{x}_{\cl{c}}\,x_{\cl{c}}
      + \hat{x}_{\cl{c}}\,\bar{x}_{\cl{c}}
   \bigr)\Bigr)
\\
&+ \sum_{\cl{c}=(\lit{x}\lor \lit{y}\lor \lit{z})\in C}\!\ \Bigl(
      b\,\overleftarrow{p_{\cl{c}}}\,\lit{x}_{\cl{c}}
      + b\,p_{\cl{c}}\,\lit{x}_{\cl{c}}
      + b\,p_{\cl{c}}\,\lit{y}_{\cl{c}}
      + b\,\overrightarrow{p_{\cl{c}}}\,\lit{y}_{\cl{c}}
      + \overleftarrow{p_{\cl{c}}}\,p_{\cl{c}}
      + \overrightarrow{p_{\cl{c}}}\,p_{\cl{c}}
\\
&\hphantom{+\sum_{\cl{c}=(\lit{x}\lor \lit{y}\lor \lit{z})\in C}\!}
      + b\,\overleftarrow{o_{\cl{c}}}\,p_{\cl{c}}
            + b\,o_{\cl{c}}\,p_{\cl{c}}
      + b\,\overrightarrow{o_{\cl{c}}}\,\lit{z}_{\cl{c}}
      + b\,o_{\cl{c}}\,\lit{z}_{\cl{c}}
      + \overleftarrow{o_{\cl{c}}}\,o_{\cl{c}}
      + \overrightarrow{o_{\cl{c}}}\,o_{\cl{c}}
      + b\,t_{\cl{c}}\,o_{\cl{c}}\Bigr).
\end{aligned}
\]
With the unordered-pair and occurrence conventions above, all these monomials appear in $f_{\Phi}$ and never cancel out, and the construction of $f_{\Phi}$ from $\Phi$ is clearly polynomial-time (and also in logspace).

\subsection{Structural Lemmas and Correctness Proofs}\label{subsec:lemmas}

To establish the validity of the reduction, we must prove that the
polynomial $f_\Phi$ constructed above satisfies two complementary properties:

\begin{description}
  \item[Soundness.]
        If the original \SAT{} instance $(X,C)$ is \emph{unsatisfiable},
        then \textbf{every} quantum circuit realizing $f_{\Phi}$ requires
        strictly more than $k=1+6|C|$ qubits.
  \item[Completeness.]
        If $(X,C)$ \emph{is satisfiable}, there exists a quantum circuit with
        exactly $k=1+6|C|$ qubits that realizes $f_{\Phi}$.
\end{description}

Recall that in the global construction, all $18m+1$ symbols can be naturally partitioned into a single $\{b\}$ and the associated $18$ symbols for each clause $\cl{c}=(\lit{x} \lor \lit{y} \lor \lit{z})$, denoted $\mathcal{S}_{\cl{c}}$:
    \[
    \underbrace{\{t_{\cl{c}}, \hat{t}_{\cl{c}}, \tilde{t}_{\cl{c}}\}}_{\text{Truth marker}} \cup 
    \underbrace{\{x_{\cl{c}}, \hat{x}_{\cl{c}}, \bar{x}_{\cl{c}}\}}_{\text{Var } x} \cup 
    \underbrace{\{y_{\cl{c}}, \hat{y}_{\cl{c}}, \bar{y}_{\cl{c}}\}}_{\text{Var } y} \cup 
    \underbrace{\{z_{\cl{c}}, \hat{z}_{\cl{c}}, \bar{z}_{\cl{c}}\}}_{\text{Var } z} \cup 
    \underbrace{\{\overleftarrow{p_{\cl{c}}}, p_{\cl{c}}, \overrightarrow{p_{\cl{c}}}\}}_{\text{OR } p_{\cl{c}}} \cup 
    \underbrace{\{\overleftarrow{o_{\cl{c}}}, o_{\cl{c}}, \overrightarrow{o_{\cl{c}}}\}}_{\text{OR } o_{\cl{c}}}.
    \]
The key observation is that $\mathcal{S}_{\cl{c}}$ naturally requires exactly $6$ qubit wires, and any $6$-wire realization uniquely encodes a truth assignment. We establish this through the following structural lemmas.

\noindent
\begin{lemma}[Base-qubit and distinct clause isolation]\label{lem:qubit-iso}
    Every circuit that realizes \(f_{\Phi}\) places the symbol \(b\) on its own qubit containing \emph{no internal} Hadamard gates, i.e., $\supp(b) = I$. Moreover, such a circuit places $\mathcal{S}_{\cl{c}}$ on separate qubit wires for different clauses \(\cl{c}\in C\), i.e., there is no qubit wire containing symbols from two different clauses.
\end{lemma}

\begin{proof}
    \(b\) appears only in cubic monomials of \(f_{\Phi}\).  If an internal Hadamard gate appears on the qubit wire containing the segment representing $b$, it generates a quadratic term of
    the form \(bs\) for some $s\in \mathcal{S}$. Since the two adjacent segments separated by that Hadamard gate have disjoint open supports, no diagonal gate can produce the same monomial \(bs\), so this term cannot be canceled. This contradicts the construction of \(f_{\Phi}\). For the same reason, because there are no quadratic terms between symbols from different clauses, consecutive symbols on the same qubit wire cannot belong to different clauses. Any wire containing symbols from two different clauses has such a transition, so the symbols from different clauses must be placed on separate qubit wires.
\end{proof}

\begin{lemma}[Good clause structure]\label{lem:good-cl}
    For any clause $\cl{c}=(\lit{x} \lor \lit{y} \lor \lit{z})$, where each literal $\lit{x}\in\{x,\bar{x}\}$ (and similarly for $y,z$), the symbol set $\mathcal{S}_{\cl{c}}$ requires at least $6$ qubit wires in any circuit realizing $f_{\Phi}$. If exactly $6$ wires are used, then the configuration on each wire is determined (up to reversal) as
        \[
        t_{\cl{c}}-\hat{t}_{\cl{c}}-\tilde{t}_{\cl{c}},\quad x_{\cl{c}}-\hat{x}_{\cl{c}}-\bar{x}_{\cl{c}},\quad y_{\cl{c}}-\hat{y}_{\cl{c}}-\bar{y}_{\cl{c}},\quad z_{\cl{c}}-\hat{z}_{\cl{c}}-\bar{z}_{\cl{c}},\quad \overleftarrow{p_{\cl{c}}}-p_{\cl{c}}-\overrightarrow{p_{\cl{c}}},\quad \overleftarrow{o_{\cl{c}}}-o_{\cl{c}}-\overrightarrow{o_{\cl{c}}}.
        \]
    If the conditions are met, we will call the clause $\cl{c}$ \emph{a good clause} for the circuit realizing $f_{\Phi}$.
\end{lemma}
\begin{proof}
Consider the symbols in $\mathcal{S}_{\cl{c}}$:
\[
\{t_{\cl{c}}, \hat{t}_{\cl{c}}, \tilde{t}_{\cl{c}}\}\cup\{x_{\cl{c}}, \hat{x}_{\cl{c}}, \bar{x}_{\cl{c}}\}\cup\{y_{\cl{c}}, \hat{y}_{\cl{c}}, \bar{y}_{\cl{c}}\}\cup\{z_{\cl{c}}, \hat{z}_{\cl{c}}, \bar{z}_{\cl{c}}\}\cup\{\overleftarrow{p_{\cl{c}}}, p_{\cl{c}}, \overrightarrow{p_{\cl{c}}}\}\cup\{\overleftarrow{o_{\cl{c}}}, o_{\cl{c}}, \overrightarrow{o_{\cl{c}}}\}.
\]
The baseline quadratic pairs among them are
\[
\hat{t}_{\cl{c}} t_{\cl{c}},\;\hat{t}_{\cl{c}}\tilde{t}_{\cl{c}};\quad
\hat{x}_{\cl{c}} x_{\cl{c}},\;\hat{x}_{\cl{c}}\bar{x}_{\cl{c}};\quad
\hat{y}_{\cl{c}} y_{\cl{c}},\;\hat{y}_{\cl{c}}\bar{y}_{\cl{c}};\quad
\hat{z}_{\cl{c}} z_{\cl{c}},\;\hat{z}_{\cl{c}}\bar{z}_{\cl{c}};\quad
\overleftarrow{p_{\cl{c}}}p_{\cl{c}},\;\overrightarrow{p_{\cl{c}}}p_{\cl{c}};\quad
\overleftarrow{o_{\cl{c}}}o_{\cl{c}},\;\overrightarrow{o_{\cl{c}}}o_{\cl{c}}.
\]
If a qubit wire were to contain $4$ of these symbols, the consecutive adjacency along that wire would yield a length-$4$ chain of quadratic terms, which would introduce a quadratic monomial not present in $f_{\Phi}$. Hence every wire contains at most $3$ symbols. Since $|\mathcal{S}_{\cl{c}}|=18$, at least $6$ wires are necessary. If exactly $6$ wires are used, then each wire contains exactly $3$ symbols, and the quadratic pairs uniquely force the six chains listed above (up to reversal on each wire).
\end{proof}
In the following, we will prove soundness and completeness simultaneously. The soundness result is an impossibility proof, and completeness is usually easier since it just constructs a satisfying circuit. We denote the qubit wires with symbols
\[
b,\quad t_{\cl{c}}-\hat{t}_{\cl{c}}-\tilde{t}_{\cl{c}} ,\quad x_{\cl{c}}-\hat{x}_{\cl{c}}-\bar{x}_{\cl{c}},\quad \overleftarrow{q}-q-\overrightarrow{q}
\]
as the $b$-, $\hat{t}_{\cl{c}}$-, $x_{\cl{c}}$-, and $q$-qubit wires, respectively.
\begin{definition}[Encode]\label{def:encode}
    For any \emph{good clause} \(\cl{c}=(\lit{x}\lor\lit{y}\lor\lit{z})\)
    and symbol \(s\in \{x_{\cl{c}}, \bar{x}_{\cl{c}}, y_{\cl{c}}, \bar{y}_{\cl{c}}, z_{\cl{c}},\allowbreak \bar{z}_{\cl{c}}, p_{\cl{c}}, o_{\cl{c}}\}\), 
    we say that the segment \(s\) \emph{encodes} \(\bool{true}\) 
    if it overlaps (\cref{def:overlap}) the truth marker segment \(t_{\cl{c}}\), 
    and encodes \(\bool{false}\) otherwise.
\end{definition}

\begin{lemma}[Encoding well-definedness]\label{lem:encode-wd}
    For any good clause $\cl{c}$ and any underlying variable $x$ represented in $\cl{c}$, the segment $t_{\cl{c}}$ overlaps with exactly one of $\{x_{\cl{c}}, \bar{x}_{\cl{c}}\}$, making the encoding in \cref{def:encode} well-defined.
\end{lemma}
\begin{proof}
    The cubics
    \[
        b\,\hat{t}_{\cl{c}}\,x_{\cl{c}},  \quad b\,\hat{t}_{\cl{c}}\,\hat{x}_{\cl{c}}, \quad b\,\hat{t}_{\cl{c}}\,\bar{x}_{\cl{c}}
    \]
    force the middle segment \(\hat{t}_{\cl{c}}\) of the truth-marker qubit wire to overlap
    \emph{each} of \(x_{\cl{c}},\hat{x}_{\cl{c}},\bar{x}_{\cl{c}}\).
    
    Without loss of generality, suppose 
    \[
    \supp(t_{\cl{c}}) = (0, l_{\hat{t}_{\cl{c}}}), \quad 
    \supp(\hat{t}_{\cl{c}}) = (l_{\hat{t}_{\cl{c}}}, r_{\hat{t}_{\cl{c}}}), \quad 
    \supp(\tilde{t}_{\cl{c}}) = (r_{\hat{t}_{\cl{c}}}, 1).
    \]
    Since the truth-marker qubit wire is ordered \(t_{\cl{c}}\,\hat{t}_{\cl{c}}\,\tilde{t}_{\cl{c}}\), the segment \(\hat{t}_{\cl{c}}\) touches \(t_{\cl{c}}\) on its left edge at \(l_{\hat{t}_{\cl{c}}}\) and \(\tilde{t}_{\cl{c}}\) on its right edge at \(r_{\hat{t}_{\cl{c}}}\).
    
    Because \(x_{\cl{c}}\) and \(\bar{x}_{\cl{c}}\) lie on \emph{opposite} sides of \(\hat{x}_{\cl{c}}\) on their qubit wire, we consider two cases:
    
    \noindent\textbf{Case 1: \(x_{\cl{c}}\) on the left, \(\bar{x}_{\cl{c}}\) on the right.}
    Write the variable-wire supports as
    \[
        \supp(x_{\cl{c}})=(0,a),\qquad
        \supp(\hat{x}_{\cl{c}})=(a,b),\qquad
        \supp(\bar{x}_{\cl{c}})=(b,1).
    \]
    The cubics involving \(x_{\cl{c}}\) and \(\bar{x}_{\cl{c}}\) force
    \(a>l_{\hat{t}_{\cl{c}}}\) and \(b<r_{\hat{t}_{\cl{c}}}\).  Since \(a<b\), this implies
    \(x_{\cl{c}}\) overlaps \(t_{\cl{c}}=(0,l_{\hat{t}_{\cl{c}}})\) and is disjoint from
    \(\tilde{t}_{\cl{c}}=(r_{\hat{t}_{\cl{c}}},1)\), while \(\bar{x}_{\cl{c}}\) overlaps
    \(\tilde{t}_{\cl{c}}\) and is disjoint from \(t_{\cl{c}}\).
    By \cref{def:encode}, $x_{\cl{c}}$ encodes \(\bool{true}\) and \(\bar{x}_{\cl{c}}\) encodes \(\bool{false}\).
    
    \noindent\textbf{Case 2: \(x_{\cl{c}}\) on the right, \(\bar{x}_{\cl{c}}\) on the left.}
    By symmetry, $x_{\cl{c}}$ encodes \(\bool{false}\) and \(\bar{x}_{\cl{c}}\) encodes \(\bool{true}\).
    
    Therefore, \(t_{\cl{c}}\) overlaps with exactly one of \(\{x_{\cl{c}}, \bar{x}_{\cl{c}}\}\), making the encoding well-defined.
\end{proof}

We now relate occurrences of the same variable across good clauses.

\begin{lemma}[Variable-value realizability]\label{lem:var-rlz}\hfill
    \begin{enumerate}[(a)]    
    \item (Soundness) For any fixed variable $x$, all occurrences whose underlying variable is $x$ in good clauses induce the same value of $x$: for each such occurrence in a good clause $\cl{c}$, the segment $x_{\cl{c}}$ overlaps $t_{\cl{c}}$ if and only if this common value is \bool{true}, and the segment $\bar{x}_{\cl{c}}$ overlaps $t_{\cl{c}}$ if and only if this common value is \bool{false}.
    \item (Completeness) Given any truth assignment $w$ of variables in $X$, there is a choice of the positions of the two Hadamard gates and the ordering of symbols on all $\hat{t}_{\cl{c}}$-qubit and $\hat{x}_{\cl{c}}$-qubit wires such that $x_{\cl{c}}$ overlaps $t_{\cl{c}}$ and $\bar{x}_{\cl{c}}$ overlaps $\tilde{t}_{\cl{c}}$ if $w(x)=\bool{true}$, or vice versa if $w(x)=\bool{false}$, while realizing exactly the same quadratic terms and cubic monomials as in $f_\Phi$.
    \end{enumerate}

\end{lemma}
\begin{proof}
    \smallskip\noindent
    \textbf{(a) Soundness.}
    The clause-synchronizing cubics
    \[
    b\,t_{\cl{c_1}}t_{\cl{c_2}},\quad b\,\hat{t}_{\cl{c_1}}\hat{t}_{\cl{c_2}},\quad b\,\tilde{t}_{\cl{c_1}}\tilde{t}_{\cl{c_2}}
    \]
    force the truth-marker wire of \(\cl{c_1}\) and that of \(\cl{c_2}\) to have the \emph{same order} (up to global reversal): if one of the two three-segment orders were reversed, the corresponding left-end and right-end segments could not both overlap.
    Likewise, the variable-synchronizing cubics
    \[
    b\,x_{\cl{c_1}}x_{\cl{c_2}},\quad b\,\hat{x}_{\cl{c_1}}\hat{x}_{\cl{c_2}},\quad b\,\bar{x}_{\cl{c_1}}\bar{x}_{\cl{c_2}}
    \]
    force the two occurrence wires whose underlying variable is \(x\) to have the same order (again up to global reversal), by the same interval-order argument.
    In each clause, the encoded value of the underlying variable is determined solely by whether \(x_{\cl{c}}\) or \(\bar{x}_{\cl{c}}\) overlaps \(t_{\cl{c}}\) (as in \cref{lem:encode-wd}).
    Since the relative orders on the \(t\)-qubit wire and the occurrence wire for \(x\) are synchronized across \(\cl{c_1},\cl{c_2}\), the overlap pattern of \(x_{\cl{c}}\) versus \(\bar{x}_{\cl{c}}\) with \(t_{\cl{c}}\) is identical in both occurrences. Hence all good occurrences with underlying variable \(x\) induce the same Boolean value of \(x\).
    
    \smallskip\noindent
    \textbf{(b) Completeness.}
    Start from any good-clause configuration in \cref{lem:good-cl} and fix a truth assignment \(w\).
    We now give an explicit placement of segment endpoints.
    
    For every clause \(\cl{c}\), place the truth-marker wire in the same order
    \(t_{\cl{c}}-\hat{t}_{\cl{c}}-\tilde{t}_{\cl{c}}\) with
    \[
        l_{\hat{t}_{\cl{c}}}=0.3,\qquad r_{\hat{t}_{\cl{c}}}=0.7,
    \]
    so that \(\supp(t_{\cl{c}})=(0,0.3)\) and \(\supp(\tilde{t}_{\cl{c}})=(0.7,1)\).
    
    For each variable \(x\) and each occurrence of \(x\) or \(\bar{x}\) in a clause \(\cl{c}\), choose the order on the \(\hat{x}_{\cl{c}}\)-qubit wire
    consistently across all such occurrences as follows.
    If \(w(x)=\bool{true}\), set
    \[
        l_{x_{\cl{c}}}=0.0,\quad r_{x_{\cl{c}}}=0.4,\quad
        l_{\bar{x}_{\cl{c}}}=0.6,\quad r_{\bar{x}_{\cl{c}}}=1.0,
    \]
    so that \(\supp(x_{\cl{c}})=(0,0.4)\) overlaps both \(\supp(t_{\cl{c}})\) and \(\supp(\hat{t}_{\cl{c}})\),
    while \(\supp(\bar{x}_{\cl{c}})=(0.6,1)\) overlaps both \(\supp(\tilde{t}_{\cl{c}})\) and \(\supp(\hat{t}_{\cl{c}})\).
    If \(w(x)=\bool{false}\), reverse the order by swapping the two choices above:
    \[
        l_{x_{\cl{c}}}=0.6,\quad r_{x_{\cl{c}}}=1.0,\quad
        l_{\bar{x}_{\cl{c}}}=0.0,\quad r_{\bar{x}_{\cl{c}}}=0.4.
    \]
    In either case, place the two Hadamard gates on the \(\hat{x}_{\cl{c}}\)-qubit wire at the boundaries
    between \(x_{\cl{c}},\hat{x}_{\cl{c}},\bar{x}_{\cl{c}}\), so the quadratic pairs on that wire are unchanged.
    
    This explicit placement ensures that the segment designated \(\bool{true}\) overlaps \(t_{\cl{c}}\),
    while the other overlaps \(\tilde{t}_{\cl{c}}\), and we can append CCZ gates between
    \(b\), \(\hat{t}_{\cl{c}}\), and \(x_{\cl{c}},\hat{x}_{\cl{c}},\bar{x}_{\cl{c}}\)
    to realize exactly the same marker--literal cubic monomials as in \(f_\Phi\).
    Because the truth-marker order is identical for all clauses and the \(x\)-qubit wire order is chosen
    consistently by \(w\), we can also append CCZ gates to realize every clause-synchronizing
    and literal-synchronizing cubic. Hence the circuit realizes all required monomials while
    encoding \(w\).
\end{proof}

\begin{lemma}[Binary OR correctness]\label{lem:or-rlz}
    \hfill
    \begin{enumerate}[label=(\alph*)]
        \item (Soundness) For a circuit realizing \(f_{\Phi}\) and any binary OR variable \(q = u \lor v\) in a \emph{good clause} $\cl{c}$ ($q$ can be $p_{\cl{c}}$ or $o_{\cl{c}}$), if segment \(q\) overlaps \(t_{\cl{c}}\) vertically, then at least one of the segments representing \(u,v\) overlaps \(t_{\cl{c}}\).
        \item (Completeness) Fix configurations of the input wires carrying \(u\) and \(v\). There is a choice of the positions of the two Hadamard gates and the ordering of symbols on the \(\overleftarrow{q}-q-\overrightarrow{q}\) qubit wire that realizes exactly the same quadratic terms and the same cubic monomials that concern $q$ as in $f_\Phi$. Moreover, if one of \(u,v\) overlaps \(t_{\cl{c}}\), the choice can be made so that \(q\) also overlaps \(t_{\cl{c}}\).
    \end{enumerate}
\end{lemma}

\begin{proof}
 (a) \textbf{Soundness.} 
    Suppose, for contradiction, that \(q\) overlaps \(t_{\cl{c}}\) while neither \(u\) nor \(v\) overlaps \(t_{\cl{c}}\).  
  Cubic monomials $b\,\overleftarrow{q}\,u$ and $b\,q\,u$ force the junction point of $\overleftarrow{q}$ and $q$ to lie in $\supp(u)$.
  Similarly, the junction point of $\overrightarrow{q}$ and $q$ must lie in $\supp(v)$.
  Consequently, $\supp(q)$ must be contained in the interval spanned by \(\supp(u)\cup\supp(v)\), i.e.,
  \[
      l_q > \min(l_u,l_v), \qquad r_q < \max(r_u,r_v).
  \]
  Since \(q\) overlaps \(t_{\cl{c}}\), we have \(l_q<r_{t_{\cl{c}}}\) and \(r_q>l_{t_{\cl{c}}}\).
  If \(t_{\cl{c}}\) is the leftmost truth-marker segment, then \(l_{t_{\cl{c}}}=0\).  The assumption that neither input overlaps \(t_{\cl{c}}\) implies
  \(l_u,l_v\ge r_{t_{\cl{c}}}\), hence \(l_q>\min(l_u,l_v)\ge r_{t_{\cl{c}}}\), contradicting \(l_q<r_{t_{\cl{c}}}\).
  If \(t_{\cl{c}}\) is the rightmost truth-marker segment, then \(r_{t_{\cl{c}}}=1\).  The assumption that neither input overlaps \(t_{\cl{c}}\) implies
  \(r_u,r_v\le l_{t_{\cl{c}}}\), hence \(r_q<\max(r_u,r_v)\le l_{t_{\cl{c}}}\), contradicting \(r_q>l_{t_{\cl{c}}}\).
  
 (b) \textbf{Completeness.}
    Pick points
    \[
        a\in \supp(u),
        \qquad
        b'\in \supp(v),
    \]
    with \(a\neq b'\).  If one of the inputs overlaps \(t_{\cl{c}}\), choose the corresponding point inside that intersection; for example, if \(u\) overlaps \(t_{\cl{c}}\), choose \(a\in \supp(u)\cap\supp(t_{\cl{c}})\), and if only \(v\) does, choose \(b'\in \supp(v)\cap\supp(t_{\cl{c}})\).  This is possible since supports are open intervals, perturbing the other point if necessary to keep \(a\neq b'\).

    If \(a<b'\), place the symbols on the \(q\)-qubit wire in the order
    \(\overleftarrow{q}-q-\overrightarrow{q}\) and set
    \[
        l_q=a,\qquad r_q=b'.
    \]
    If \(b'<a\), place them in the reversed order
    \(\overrightarrow{q}-q-\overleftarrow{q}\) and set
    \[
        l_q=b',\qquad r_q=a.
    \]
    The junction adjacent to \(u\) is chosen inside \(\supp(u)\), so
    \(u\) overlaps both \(q\) and \(\overleftarrow{q}\).
    Likewise, the junction adjacent to \(v\) lies inside \(\supp(v)\), so
    \(v\) overlaps both \(q\) and \(\overrightarrow{q}\).
    Thus we can add CCZ gates to realize exactly the four cubic monomials
    \[
        b\,\overleftarrow{q}\,u,\quad b\,q\,u,\quad b\,q\,v,\quad b\,\overrightarrow{q}\,v
    \]
    from \(G_{u\lor v}\).
    If one input overlaps \(t_{\cl{c}}\), the corresponding chosen junction is an interior point of \(t_{\cl{c}}\); therefore the open interval \(\supp(q)=(l_q,r_q)\) overlaps \(t_{\cl{c}}\).  Finally, placing the three symbols consecutively on one wire realizes precisely the quadratic
    terms \(\overleftarrow{q}\,q\) and \(\overrightarrow{q}\,q\), and no others.
\end{proof}

\begin{lemma}[Clause correctness]\label{lem:cl-rlz}
For any clause \(\cl{c}=(\lit{x}\lor\lit{y}\lor\lit{z})\):
    \begin{enumerate}[label=(\alph*)]
       \item (Soundness) For any circuit realizing \(f_{\Phi}\) and any \emph{good clause} \(\cl{c} = (\lit{x} \lor \lit{y} \lor \lit{z})\), if segment \(o_{\cl{c}}\) overlaps \(t_{\cl{c}}\), then at least one of the segments representing \(\lit{x},\lit{y},\lit{z}\) overlaps \(t_{\cl{c}}\).
        \item (Completeness) Fix a configuration of the \(\lit{x}_{\cl{c}},\lit{y}_{\cl{c}},\lit{z}_{\cl{c}}\)-qubit wires. There is a choice of the positions of the two Hadamard gates and the ordering of symbols on the \(p_{\cl{c}}\)-qubit wire and the \(o_{\cl{c}}\)-qubit wire such that if one of the segments representing \(\lit{x},\lit{y},\lit{z}\) overlaps \(t_{\cl{c}}\),
        the final \(o_{\cl{c}}\) also overlaps \(t_{\cl{c}}\),
        while realizing exactly the same quadratic terms and cubic monomials as in $f_\Phi$.
    \end{enumerate}
\end{lemma}

\begin{proof}
    The clause gadget is the cascade of two binary OR gadgets:
    \(p_{\cl{c}} = \lit{x}\lor\lit{y}\) and \(o_{\cl{c}} = p_{\cl{c}}\lor\lit{z}\).
    
    \textbf{(Soundness).}
    Apply \Cref{lem:or-rlz}(a) to the second OR-gadget with inputs \(p_{\cl{c}}\) and \(\lit{z}_{\cl{c}}\).
    If \(o_{\cl{c}}\) overlaps \(t_{\cl{c}}\), then either \(p_{\cl{c}}\) or \(\lit{z}_{\cl{c}}\) overlaps \(t_{\cl{c}}\).
    If \(\lit{z}_{\cl{c}}\) overlaps \(t_{\cl{c}}\), the desired conclusion already holds; otherwise \(p_{\cl{c}}\) overlaps \(t_{\cl{c}}\).
    Applying \Cref{lem:or-rlz}(a) to the first OR-gadget then yields that either \(\lit{x}_{\cl{c}}\) or \(\lit{y}_{\cl{c}}\)
    overlaps \(t_{\cl{c}}\). Hence at least one of \(\lit{x}_{\cl{c}},\lit{y}_{\cl{c}},\lit{z}_{\cl{c}}\) overlaps \(t_{\cl{c}}\).
    
    \textbf{(Completeness).}
    Fix a configuration of the \(\lit{x}_{\cl{c}},\lit{y}_{\cl{c}},\lit{z}_{\cl{c}}\)-qubit wires and suppose one of them overlaps \(t_{\cl{c}}\).
    If \(\lit{z}_{\cl{c}}\) overlaps \(t_{\cl{c}}\), first use \Cref{lem:or-rlz}(b) to realize the gadget
    \(p_{\cl{c}}=\lit{x}_{\cl{c}}\lor\lit{y}_{\cl{c}}\) (without requiring \(p_{\cl{c}}\) to overlap \(t_{\cl{c}}\)).  Then use \Cref{lem:or-rlz}(b) to realize
    \(o_{\cl{c}} = p_{\cl{c}}\lor\lit{z}_{\cl{c}}\) so that \(o_{\cl{c}}\) overlaps \(t_{\cl{c}}\), while keeping exactly the quadratic and cubic
    monomials that concern \(o_{\cl{c}}\).
    Otherwise, one of \(\lit{x}_{\cl{c}},\lit{y}_{\cl{c}}\) overlaps \(t_{\cl{c}}\); apply \Cref{lem:or-rlz}(b) to the gadget
    \(p_{\cl{c}}=\lit{x}_{\cl{c}}\lor\lit{y}_{\cl{c}}\) to make \(p_{\cl{c}}\) overlap \(t_{\cl{c}}\), and then apply it again to
    \(o_{\cl{c}}=p_{\cl{c}}\lor\lit{z}_{\cl{c}}\) to make \(o_{\cl{c}}\) overlap \(t_{\cl{c}}\).
    In both cases we realize exactly the quadratic and cubic monomials associated with the two OR gadgets, as required.
\end{proof}

\begin{lemma}[Clause gadget with one extra wire]\label{lem:extra-wire}
Fix a clause \(\cl{c}=(\lit{x}\lor\lit{y}\lor\lit{z})\) and configurations of its three literal wires.
There exists a realization of the clause-local monomials involving
\(t_{\cl{c}},\hat{t}_{\cl{c}},\tilde{t}_{\cl{c}},\overleftarrow{p_{\cl{c}}},p_{\cl{c}},\overrightarrow{p_{\cl{c}}},\overleftarrow{o_{\cl{c}}},o_{\cl{c}},\overrightarrow{o_{\cl{c}}}\) using at most \(7\) wires,
such that the segment \(o_{\cl{c}}\) overlaps \(t_{\cl{c}}\), regardless of whether any literal overlaps the original truth-marker segment.
\end{lemma}

\begin{proof}
Start from the canonical 6-wire clause gadget (one wire each for \(t_{\cl{c}}\), the three literals, \(p_{\cl{c}}\), and \(o_{\cl{c}}\)).
We use the extra wire to \emph{split off} the truth-marker segment \(t_{\cl{c}}\): place \(t_{\cl{c}}\) alone on a fresh wire with no
internal Hadamard gates, so \(\supp(t_{\cl{c}})=I\).
Then \(t_{\cl{c}}\) overlaps \(o_{\cl{c}}\) regardless of the literals, and the cubic monomial \(b\,t_{\cl{c}}\,o_{\cl{c}}\) is realizable.

On the remaining truth-marker wire, keep \(\hat{t}_{\cl{c}}\) and \(\tilde{t}_{\cl{c}}\) as two consecutive segments (one \(H\) gate).
Choose \(\supp(\hat{t}_{\cl{c}})\) wide enough to intersect all three symbols on each of the three literal wires, so all marker--literal CCZ terms
\(b\,\hat{t}_{\cl{c}}\,s\), with \(s\in\{x_{\cl{c}},\hat{x}_{\cl{c}},\bar{x}_{\cl{c}},y_{\cl{c}},\hat{y}_{\cl{c}},\bar{y}_{\cl{c}},z_{\cl{c}},\hat{z}_{\cl{c}},\bar{z}_{\cl{c}}\}\), are realized.
The quadratic term \(\hat{t}_{\cl{c}}\tilde{t}_{\cl{c}}\) is realized by adjacency on this wire, and \(\hat{t}_{\cl{c}}t_{\cl{c}}\) is realized by a CZ gate
between the \(\hat{t}_{\cl{c}}\)-wire and the \(t_{\cl{c}}\)-wire.

Finally, realize the \(p_{\cl{c}}\)-wire and \(o_{\cl{c}}\)-wire using the unconditional part of \Cref{lem:or-rlz}(b), which realizes exactly the quadratic and cubic
monomials in \(G_{\lit{x}_{\cl{c}}\lor\lit{y}_{\cl{c}}}\) and \(G_{p_{\cl{c}}\lor\lit{z}_{\cl{c}}}\).
Altogether this uses one extra wire and realizes all required monomials while ensuring \(o_{\cl{c}}\) overlaps \(t_{\cl{c}}\), since \(\supp(t_{\cl{c}})=I\).
\end{proof}

\subsection{Main Theorems}\label{subsec:main-theorems}
In this subsection, we state and prove the two main theorems of completeness and soundness, using the structural lemmas established above. Furthermore, a conditional lower bound under ETH is also provided.
\begin{theorem}[Completeness]\label{thm:completeness}
    If the \SAT{} instance \((X,C)\) is satisfiable, then there exists a circuit with exactly
    \(k=1+6|C|\) qubits that realizes \(f_{\Phi}\).
\end{theorem}
\begin{proof}
    Let \(w\) be a satisfying assignment. We construct the circuit using exactly $k=1+6|C|$ qubits as follows:
    The base qubit $b$ is isolated (\Cref{lem:qubit-iso}). For each clause \(\cl{c}\), we use the unique 6-wire configuration described in \Cref{lem:good-cl}.
    By \Cref{lem:var-rlz}(b), we can place the variable segments so that if a literal is true, its segment overlaps $t_{\cl{c}}$.
    Since \(w\) is a satisfying assignment, for every clause \(\cl{c}\), at least one literal is true. By \Cref{lem:cl-rlz}(b), there exists a valid configuration of the OR-gadget wires such that $f_\Phi$ is realized, on exactly $k$ qubits.
\end{proof}

\begin{theorem}[Soundness]\label{thm:soundness}
    If the \SAT{} instance \((X,C)\) is \emph{unsatisfiable}, 
    then every quantum circuit realizing \(f_{\Phi}\) must use \emph{strictly more} than \(k\) qubits.
\end{theorem}
\begin{proof}
    Suppose a circuit on \(\le k\) qubits realizes \(f_{\Phi}\). By \Cref{lem:qubit-iso,lem:good-cl}, the circuit must use exactly $k$ qubits (allocating the isolated base wire and 6 wires per clause) and the wire configuration is the canonical one. By \Cref{lem:var-rlz}(a), this configuration encodes a truth assignment \(w\).
    The presence of the monomial $b\,t_{\cl{c}}\,o_{\cl{c}}$ in $f_\Phi$ requires the segment $o_{\cl{c}}$ to overlap $t_{\cl{c}}$ for every clause.
    By \Cref{lem:cl-rlz}(a), the segment $o_{\cl{c}}$ can overlap $t_{\cl{c}}$ only if the encoded assignment $w$ satisfies the clause $\cl{c}$.
    Therefore, the existence of such a circuit implies $w$ satisfies all clauses, which contradicts the assumption that $\Phi$ is unsatisfiable.
\end{proof}

\begin{corollary}[$\mathsf{NP}$-completeness]\label{cor:np-completeness}
    \problem{} is $\mathsf{NP}$-complete.
\end{corollary}
\begin{proof}
    Membership in $\mathsf{NP}$ is straightforward: one can guess a polynomial-size circuit description (the ordered wire segments, hence the $H$ gates, together with a list of $Z$, CZ, and CCZ gates), and then verify by recomputing the associated polynomial. Hardness follows from the reduction (Theorems \ref{thm:completeness} and \ref{thm:soundness}).
\end{proof}

Finally, we provide a conditional lower bound under the Exponential Time Hypothesis (ETH)~\cite{ETH0,ETH} using the same reduction.

\begin{theorem}
\label{thm:eth}
Assuming the Exponential Time Hypothesis \textnormal{(ETH)}, no algorithm solves \textnormal{\problem{}} on instances with $n$ symbols in time $2^{o(n)}$, even when restricted to instances with $k = \Theta(n)$ qubits.
\end{theorem}
\begin{proof}
The polynomial reduction from \SAT{} to \problem{} transforms a \SAT{} instance with $|X|$ variables and $|C|$ clauses into a \problem{} instance with $n = |\mathcal{S}| = 18|C|+1$ symbols, giving $n = \Theta(|C|)$ and $k = 1+6|C| = \Theta(n)$.

By the Sparsification Lemma (\cref{thm:sparsification}), ETH implies that there exists a constant $D>0$ such that no algorithm solves sparse \SAT{} instances with $|C|\le D|X|$ in time $2^{o(|X|)}$. For every such instance, the reduction produces a \problem{} instance of size $n \le (18D+1)|X|+1$, so $|X| \ge \frac{n-1}{18D+1}$.

Suppose for contradiction that \problem{} can be solved in time $2^{\varepsilon n}$ for every constant $\varepsilon > 0$. Given a sparse \SAT{} instance with $|C| \le D|X|$, apply the reduction and run this algorithm with $\varepsilon = \epsilon/(18D+1)$. The total running time is at most $2^{\varepsilon n} \cdot \mathrm{poly}(|X|) \le 2^{\epsilon |X|} \cdot \mathrm{poly}(|X|)$. Since $\epsilon > 0$ was arbitrary, this yields a $2^{o(|X|)}$-time algorithm for sparse \SAT{}, contradicting ETH.
\end{proof}

\section{Approximation Hardness and Restricted Degree}\label{sec:approx-hardness}

\subsection{Approximation Hardness}
In this section, we prove that even the $\alpha$-approximation version of the problem \problem{} is \NP-hard. By observing that the previous reduction from \SAT{} to \problem{} in \Cref{sec:reduction} is actually a reduction for \textsc{Max-3-SAT}, we can invoke the \NP-hardness of approximating \textsc{Max-3-SAT} due to Håstad (see \Cref{thm:hastad}). We first review some basic definitions and then prove the main theorem.

\begin{definition}[$\alpha$-approximation version of \problem{}]
  \label{def:approxproblem}
  Given a real parameter $\alpha > 1$, the promise problem \approxproblem{$\alpha$} is to distinguish between instances in the following $(\mathsf{Yes}, \mathsf{No})$ languages:
  \begin{itemize}
    \item[$\mathsf{Yes}$:] $(f\in \mathbb{F}_2[\mathcal{S}],k\in \mathbb{N})$, where $f$ is a degree-3 polynomial with no constant term and $w(f)\leq k$.
    \item[$\mathsf{No}$:]  $(f\in \mathbb{F}_2[\mathcal{S}],k\in \mathbb{N})$, where $f$ is a degree-3 polynomial with no constant term and $w(f)\geq \alpha k$.
  \end{itemize}
\end{definition}

\begin{definition}[\textsc{Max-3-SAT}]
  \label{def:maxsat}
  An instance of \textsc{Max-3-SAT} is a \SAT{} formula $\Phi$ with $m$ clauses. The task is to determine the maximum number of clauses that can be simultaneously satisfied by any assignment.
\end{definition}

\begin{theorem}[Reduction from \textsc{Max-3-SAT}]\label{thm:maxsat-reduction}
  Let $\Phi$ be a \SAT{} instance with $m$ clauses and let $t$ be an integer.
  Construct $f_\Phi$ in polynomial time as in \Cref{sec:reduction} and set
  \[
    K \;=\; 1 + 6m + (m - t).
  \]
  Then $w(f_\Phi)\le K$ if and only if $\Phi$ has an assignment satisfying at least $t$ clauses.
\end{theorem}

\begin{proof}
  Let $\OPT(\Phi)$ be the maximum number of clauses satisfiable by any assignment.
  By \Cref{lem:cl-rlz}(a), any clause gadget realized on exactly $6$ wires must encode a
  satisfied clause. Conversely, by \Cref{lem:extra-wire}, any clause can be realized on $7$
  wires with the segment $o_{\cl{c}}$ overlapping $t_{\cl{c}}$ regardless of the literals.

  Now, if $\Phi$ has an assignment satisfying at least $t$ clauses, then we can set the truth-marker wires honestly and variable wires accordingly, and realize each satisfied clause gadget on $6$ wires and each unsatisfied clause gadget on $7$ wires. For clauses realized with the extra wire, choose the split-off $t_{\cl{c}}$ with full support $I$, and choose $\hat t_{\cl{c}}$ and $\tilde t_{\cl{c}}$ with the same coarse positions as in the good clauses; hence all global truth-marker synchronization terms involving \(\cl{c}\) remain realizable. The variable-synchronization terms are realized by the common assignment as in \cref{lem:var-rlz}(b). This yields a circuit realizing $f_\Phi$ on a number of wires at most
  \[
    1 + 6m + (m - \OPT(\Phi)) \leq 1 + 6m + (m - t) = K.
  \]
  Conversely, suppose there is a circuit realizing $f_\Phi$ on at most $K$ wires. 
  Let \(g\) be the number of clauses whose symbols are realized on exactly \(6\) wires.
  By \Cref{lem:qubit-iso,lem:good-cl}, each such clause is a \emph{good clause}.
  Every clause needs at least \(6\) wires, so the remaining \(m-g\) clauses use at least \(7\) wires each.
  Hence the total number of wires is at least
  \[
    1 + 6g + 7(m-g) \;=\; 1 + 7m - g.
  \]
  Since the circuit uses at most \(K = 1 + 7m - t\) wires, the pigeonhole principle gives \(g \ge t\).

  Now extract a Boolean assignment from the good clauses.
  For any variable \(x\) that appears in at least one good clause, define \(w(x)\) to be the truth value
  encoded by \(x_{\cl{c}}\) in that good clause (via \Cref{def:encode}).
  This is well-defined by \Cref{lem:var-rlz}(a), which synchronizes all occurrences of \(x\) across good clauses.
  For variables that do not appear in any good clause, assign an arbitrary truth value.

  For each good clause \(\cl{c}\), the cubic term \(b\,t_{\cl{c}}\,o_{\cl{c}}\) forces \(o_{\cl{c}}\) to overlap \(t_{\cl{c}}\).
  By \Cref{lem:cl-rlz}(a), this implies that at least one of the literals in \(\cl{c}\) overlaps \(t_{\cl{c}}\),
  hence is \bool{true} under the encoding. Therefore all \(g\) good clauses are satisfied by \(w\),
  and in particular at least \(t\) clauses are satisfied. This completes the reduction.
\end{proof}

We now apply Håstad's theorem (\Cref{thm:hastad}) to establish the inapproximability of our problem.
\begin{theorem}[\approxproblem{$\alpha$} Hardness]\label{thm:approxnphard}
    For any constant $0<\epsilon<1/48$, \textnormal{\approxproblem{($\tfrac{49}{48} - \epsilon$)}} is $\mathsf{NP}$-hard.
\end{theorem}

\begin{proof}
Fix a constant $0<\epsilon<1/48$, and apply \Cref{thm:hastad} with the gap parameter $\delta=\epsilon$. We use the same reduction $\Phi \mapsto f_{\Phi}$ from \Cref{sec:reduction}, together with \Cref{thm:maxsat-reduction}. If $\Phi$ is satisfiable, then $w(f_\Phi) \le 1+6m$ by setting $t=m$. If every assignment satisfies at most $(\tfrac{7}{8}+\delta)m$ clauses, then the counting argument in \Cref{thm:maxsat-reduction} gives
\[
  w(f_\Phi) \;\ge\; 1 + 6m + \left(\frac{1}{8} - \delta\right)m.
\]
Set $k=1+6m$. Since $\delta=\epsilon$, we have
\[
\begin{aligned}
  \left(1 + 6m + \left(\frac{1}{8} - \delta\right)m\right)
  - \left(\frac{49}{48}-\epsilon\right)(1+6m) = \epsilon(5m+1)-\frac{1}{48}.
\end{aligned}
\]
This quantity is positive for all sufficiently large $m$. By \Cref{thm:hastad}, distinguishing the satisfiable case from the case in which every assignment satisfies at most $(\tfrac{7}{8}+\delta)m$ clauses is $\mathsf{NP}$-hard. Hence \approxproblem{($\tfrac{49}{48}-\epsilon$)} is $\mathsf{NP}$-hard.
\end{proof}

\subsection{Restriction to Degree-2 Polynomials}
We now prove that approximation hardness persists even when the instance is restricted to degree-2 polynomials, i.e., cubic terms are forbidden. The key idea is to replace each symbol with multiple ``copies'' and connect copies via ``cliques'' of CZ gates. Since at most two copies of any symbol can share a single wire, sufficiently many surviving copies enforce the same geometric constraints as the original construction.

\subsubsection{Twin Construction}
Given a \textsc{Max-3-SAT} instance $\Phi$ with $m$ clauses and $n$ variables, we construct a degree-2 polynomial $f'_\Phi$ using the following twin duplication strategy.

\begin{definition}[Twin Polynomial]\label{def:twin-poly}
    Let $f_\Phi$ be the polynomial from \Cref{sec:reduction}, and fix sufficiently large integers $M\leq Nm$. We construct the \emph{twin polynomial} $f'_\Phi$ over an expanded symbol set $\mathcal{S}'$ as follows. Equivalently, we specify the quadratic support of $f'_\Phi$: if the same quadratic monomial is produced more than once by the rules below, it appears only once in $f'_\Phi$.
    \begin{enumerate}
        \item Replace the base symbol $b$ with $N\cdot m$ copies: $\{b^{(1)}, b^{(2)}, \dots, b^{(Nm)}\}$.
        \item Replace every other symbol $s \in \mathcal{S} \setminus \{b\}$ with $M$ copies: $\{s^{(1)}, s^{(2)}, \dots, s^{(M)}\}$.
        \item For each set of copies, add a ``\emph{clique constraint}'':
        \[
            \text{For } b:\quad \sum_{1 \le i < j \le Nm} b^{(i)} b^{(j)}, 
            \qquad \text{For } s \ne b:\quad \sum_{1 \le i < j \le M} s^{(i)} s^{(j)}.
        \]
        \item Replace each \emph{quadratic} term $u v$ in $f_\Phi$ ($u,v\neq b$) with a ``\emph{matching}'' only:
        \[
            u v \;\mapsto\; \sum_{1\leq i \leq M} u^{(i)} v^{(i)}.
        \]
        \item Replace each \emph{cubic} term $b\, u\, v$ in $f_\Phi$ with a ``\emph{tripartite complete graph}'':
        \[
            b\,u\,v \;\mapsto\; \sum_{\substack{1\leq i \leq Nm\\1\leq j \leq M}} b^{(i)}u^{(j)} +  \sum_{\substack{1\leq i \leq Nm\\1\leq k \leq M}} b^{(i)}v^{(k)}+ \sum_{1\leq j,k\leq M} u^{(j)} v^{(k)}.
        \]
    \end{enumerate}
    The resulting polynomial $f'_\Phi$ has degree 2.
\end{definition}

The crucial geometric insight is that ``clique constraints'' prevent more than two copies from sharing a wire, hence ensure the vast majority of copies ``survive'' as isolated symbols among symbols in the clique.

\begin{lemma}[Clique Survival]\label{lem:twin-survival}
    Let $\{s_1, \dots, s_r\} \subseteq  \mathcal{S}'$ be a set of expanded symbols so that a ``clique constraint'' $\sum_{1\leq i<j \leq r} s_i s_j$ appears in $f'_\Phi$.
    Then in any circuit realizing $f'_\Phi$, at most $2$ of $\{s_1, \dots, s_r\}$ can share a qubit wire with another symbol in $\{s_1, \dots, s_r\}$.
\end{lemma}

\begin{proof}
    First suppose three symbols $s_i, s_j, s_k$ all lie on the same wire.
    Assume the order of them along the wire is $s_i - s_j - s_k$ (possibly connected with Hadamard gates).
    Then, the term $s_i s_k$ cannot be realized: $s_i$ and $s_k$ are non-adjacent on the wire, and CZ gates can only connect symbols on \emph{distinct} wires. 

    Next, suppose four symbols $s_{i_1}, s_{i_2}$ lie on the same wire, while $s_{j_1}, s_{j_2}$ lie on another wire (with order from left to right). Then clearly the terms $s_{i_1} s_{i_2}$ and $s_{j_1} s_{j_2}$ force adjacency realization of both pairs. Without loss of generality, assume in the circuit realizing $f'_\Phi$, the Hadamard gate between $s_{i_1}$ and $s_{i_2}$ is applied before the Hadamard gate between $s_{j_1}$ and $s_{j_2}$. Then clearly $\supp(s_{i_1})$ and $\supp(s_{j_2})$ are disjoint, so the term $s_{i_1} s_{j_2}$ cannot be realized.

    Hence, we conclude that at most two symbols from the clique constraint can share a wire with another symbol from the same set.
\end{proof}

\begin{definition}[Surviving a Single Cubic Term]\label{def:cubic-survival}
    Consider a cubic term $b\,u\,v$ in the original polynomial $f_\Phi$, replaced by the ``tripartite complete graph'' in \Cref{def:twin-poly}.
    In any circuit realizing $f'_\Phi$, we say a copy of $b$, $u$, or $v$ \emph{survives} the cubic term $b\,u\,v$ if no other symbols in the set $\{b^{(i)}\}_{i=1}^{Nm}\cup \{u^{(j)}\}_{j=1}^{M} \cup \{v^{(k)}\}_{k=1}^{M}$ lie on the same wire as the copy.
\end{definition}

\begin{corollary}[Single Cubic Term Survival]\label{lem:single-cubic-survival}
    Consider any cubic term $b\,u\,v$ in the original polynomial $f_\Phi$.
    In any circuit realizing $f'_\Phi$, there are at least $Nm - 2$ copies of $b$ that survive this cubic term, and at least $M - 2$ copies of each of $u$ and $v$ that survive this cubic term.
\end{corollary}

\begin{proof}
    It is clear from the construction of the twin polynomial that each cubic term $b\,u\,v$ in $f_\Phi$ gives rise to a ``tripartite complete graph'' in $f'_\Phi$, which will form a complete ``clique constraint'' on all copies of $b,u$ and $v$ together with the clique constraints on the copies of $b,u,v$ themselves. Hence the result follows from \Cref{lem:twin-survival}.
\end{proof}
\begin{lemma}[Global Cubic Term Survival]\label{lem:global-cubic-survival}
    In any circuit realizing $f'_\Phi$, define a copy $s^{(j)}$ (for $s \in \mathcal{S}$) to be \emph{globally surviving} if it survives \emph{all} cubic terms in $f_\Phi$ that involve $s$ in the sense of \Cref{def:cubic-survival}. Then for any symbol $s \in \mathcal{S} \setminus \{b\}$, there are at least $M - 20$ globally surviving copies of $s$.
    Similarly, there are at least $Nm - 60m$ globally surviving copies of $b$.
\end{lemma}

\begin{proof}
    We analyze all cubic terms in $f_\Phi$ from the construction in \Cref{sec:reduction}. Recall from \cref{subsec:marker-gadget} to \cref{subsec:clause-gadget} that the cubic terms fall into the following categories:
    
    \paragraph{Type 1: Truth-Marker Synchronization ($G_{\mathrm{TM}}$).}
    For each pair of distinct clauses $\cl{c_1} \ne \cl{c_2}$:
    \[
        b\,t_{\cl{c_1}}\,t_{\cl{c_2}}, \quad b\,\hat{t}_{\cl{c_1}}\,\hat{t}_{\cl{c_2}}, \quad b\,\tilde{t}_{\cl{c_1}}\,\tilde{t}_{\cl{c_2}}.
    \]
    
    \paragraph{Type 2: Variable Gadget ($G_{x_{\cl{c}}}$).}
    For each clause $\cl{c}$ containing variable $x$:
    \[
        b\,\hat{t}_{\cl{c}}\,x_{\cl{c}}, \quad b\,\hat{t}_{\cl{c}}\,\hat{x}_{\cl{c}}, \quad b\,\hat{t}_{\cl{c}}\,\bar{x}_{\cl{c}}.
    \]
    
    \paragraph{Type 3: Variable Synchronization ($G_x$).}
    For each variable $x$ and each pair of distinct literal occurrences of $\lit{x}$ or $\bar{\lit{x}}$ (not necessarily from distinct clauses):
    \[
        b\,x_{\cl{c_1}}\,x_{\cl{c_2}}, \quad b\,\hat{x}_{\cl{c_1}}\,\hat{x}_{\cl{c_2}}, \quad b\,\bar{x}_{\cl{c_1}}\,\bar{x}_{\cl{c_2}}.
    \]
    
    \paragraph{Type 4: OR Gadgets for a clause $\cl{c}$ ($G_{u \lor v}$).}
    For each clause $\cl{c}$ with literals $\lit{x}, \lit{y}, \lit{z}$, the two OR gadgets contribute the following cubic terms:
    $$
    b\,\overleftarrow{p_{\cl{c}}}\,\lit{x}_{\cl{c}}
            , b\,p_{\cl{c}}\,\lit{x}_{\cl{c}}
      , b\,\overrightarrow{p_{\cl{c}}}\,\lit{y}_{\cl{c}}
      , b\,p_{\cl{c}}\,\lit{y}_{\cl{c}}, b\,\overleftarrow{o_{\cl{c}}}\,p_{\cl{c}}
            , b\,o_{\cl{c}}\,p_{\cl{c}}
      , b\,\overrightarrow{o_{\cl{c}}}\,\lit{z}_{\cl{c}}
      , b\,o_{\cl{c}}\,\lit{z}_{\cl{c}}.
    $$
    \paragraph{Type 5: Clause Synchronization ($G_{\cl{c}}$).}
    For each clause $\cl{c}$, the pinning cubic term is:
    \[
        b\,t_{\cl{c}}\,o_{\cl{c}}.
    \]
    
    We now count the number of dead copies by analyzing the clique structure induced by each type of cubic term.
    
    \textbf{Type 1:} The terms $\{b\,t_{\cl{c_1}}\,t_{\cl{c_2}} : \cl{c_1} \ne \cl{c_2}\}$ (together with individual constraints on $b,t_{\cl{c_1}},t_{\cl{c_2}}$ respectively) will form a clique constraint in $\mathcal{S}'$ among all $\{ b^{(i)} : 1\leq i\leq Nm\}\cup \{t_{\cl{c}}^{(j)} :  \cl{c} \in C, 1 \le j \le M\}$ copies. Similarly, for $\{b\,\hat{t}_{\cl{c_1}}\,\hat{t}_{\cl{c_2}}\}$ and $\{b\,\tilde{t}_{\cl{c_1}}\,\tilde{t}_{\cl{c_2}}\}$ respectively. By \Cref{lem:twin-survival}, each clique kills at most 2 copies of any individual symbol. Thus Type 1 kills at most 2 copies each of $ t_{\cl{c}}, \hat{t}_{\cl{c}}, \tilde{t}_{\cl{c}}$ for any clause $\cl{c}$ and at most $6$ copies of $b$.
    
    \textbf{Type 2:} There are $9m$ cubic terms (three for each of the three literal occurrences in a clause). Each literal-occurrence symbol, such as $x_{\cl{c}}, \hat{x}_{\cl{c}}$, or $\bar{x}_{\cl{c}}$, appears in one Type 2 term, so Type 2 kills at most $2$ copies of it by \Cref{lem:single-cubic-survival}. The marker $\hat{t}_{\cl{c}}$ appears in the nine Type 2 terms of clause $\cl{c}$, so Type 2 kills at most $18$ copies of $\hat{t}_{\cl{c}}$. The markers $t_{\cl{c}}$ and $\tilde{t}_{\cl{c}}$ do not appear in Type 2. For $b$, Type 2 kills at most $18m$ copies.
    
    \textbf{Type 3:} For each variable $x$ and each of its three symbol types ($x_{\cl{c}}, \hat{x}_{\cl{c}}, \bar{x}_{\cl{c}}$), the terms $\{b\,x_{\cl{c_1}}\,x_{\cl{c_2}}\}$ (allowing $\cl{c_1}=\cl{c_2}$ when a clause contains multiple $x$-literals) form the same clique among all copies of $b$ and the corresponding occurrence-symbol copies across clauses; intra-clause duplicate occurrences are treated as distinct occurrence symbols. By \Cref{lem:twin-survival}, Type 3 kills at most $2$ copies of each participating non-$b$ symbol and at most $6n$ copies of $b$.
    
    \textbf{Type 4:} Each clause $c$ has two OR gadgets ($p_{\cl{c}}$ and $o_{\cl{c}}$), each with 4 cubic terms. Each OR symbol (e.g., $\overleftarrow{p_{\cl{c}}}, p_{\cl{c}}, \overrightarrow{p_{\cl{c}}}, \overleftarrow{o_{\cl{c}}}, o_{\cl{c}}, \overrightarrow{o_{\cl{c}}}$) appears in at most 4 Type 4 terms; the worst case is the middle symbol $p_{\cl{c}}$, which is both the output of the first OR gadget and an input of the second. By \Cref{lem:single-cubic-survival}, Type 4 kills at most $2 \times 4 = 8$ copies per OR symbol. For literal symbols $\lit{x},\lit{y},\lit{z}$ in the OR gadgets, each appears in at most 2 Type 4 terms, killing at most $2 \times 2 = 4$ copies per literal symbol. The symbol $b$ appears in all $8$ Type 4 terms per clause, killing at most $2 \times 8 = 16$ copies of $b$ per clause.
    
    \textbf{Type 5:} Each clause has one term $b\,t_{\cl{c}}\,o_{\cl{c}}$. This kills at most 2 copies each of $b$, $t_{\cl{c}}$ and $o_{\cl{c}}$ by \Cref{lem:single-cubic-survival}.
    
    \medskip
    \textbf{Summary for non-$b$ symbols:} Summing over all types:
    \begin{itemize}
        \item Truth markers: $t_{\cl{c}}$ has at most $2+2=4$ dead copies, $\tilde{t}_{\cl{c}}$ has at most $2$ dead copies, and $\hat{t}_{\cl{c}}$ has at most $2+18=20$ dead copies.
        \item Variable symbols: Type 2 ($2$) + Type 3 ($2$) + Type 4 ($4$) $\le 8$ dead.
        \item OR symbols: Type 4 ($8$) + Type 5 ($2$) $\le 10$ dead.
    \end{itemize}
    Taking the worst case, at least $M - 20$ copies survive for any non-$b$ symbol. 
    
    \textbf{For $b$ copies:} Variables not appearing in any clause can be discarded, so $n\leq 3m$. Hence the total number of dead $b$ copies is at most
    $$
    6 + 18m + 6n + 16m + 2m \leq 6+54m \leq 60m,
    $$
    since $m\geq 1$.
    Thus at least $Nm - 60m$ copies of $b$ survive globally.
\end{proof}

\subsubsection{Completeness and Soundness}

We now establish the wire count bounds for $f'_\Phi$.

\begin{lemma}[Twin Polynomial Completeness]\label{lem:twin-complete}
    If $\Phi$ has a satisfiable assignment, then there exists a circuit realizing $f'_\Phi$ using at most $(6M+N)m$ qubit wires.
\end{lemma}

\begin{proof}
    Construct $M$ parallel copies of the circuit from \cref{thm:completeness}, but first add only the Hadamard gates in the temporal order of the original circuit (the same Hadamard gates in different copies are placed close to one another with small perturbations), one for each twin index $j \in \{1, 2, \dots, M\}$. Hence, the Hadamards already realize all quadratic terms arising from the original quadratic terms in $f_{\Phi}$ (note that \cref{thm:completeness} deals with perfect completeness). We place each remaining base symbol copy $b^{(i)}$, for $i=M+1,\dots,Nm$, on its own additional wire.

    Next, we add CZ gates across copies: for any cubic term $b\,u\,v$ in $f_\Phi$, we add CZ gates between all copies of $b^{(i)}, u^{(j)}, v^{(k)}$ across all indices to realize the tripartite complete graphs in $f'_\Phi$. As the circuit from \cref{thm:completeness} can implement the CCZ gate for $b\,u\,v$, and we parallelize $M$ copies, all current supports of $b^{(i)}, u^{(j)}, v^{(k)}$ intersect for all $i,j,k$. Thus the CZ gates can be added as desired. 

    Finally, for the ``clique constraints'', we can add CZ gates between all pairs of copies of each symbol on their respective wires using the same reasoning.

    Total wires: $M \times (6m+1) + Nm - M = (6M+N)m$.
\end{proof}

\begin{lemma}[Twin Polynomial Soundness]\label{lem:twin-sound}
  Let $N > 60$ and $M$ be sufficiently large. If every assignment to $\Phi$ satisfies at most $t$ clauses, then every circuit realizing $f'_\Phi$ requires at least
    \[
        K' \;=\; (7M+N)m - tM - 750m.
    \]
\end{lemma}

\begin{proof}
    Consider any circuit realizing $f'_\Phi$. By \Cref{lem:global-cubic-survival}, for any symbol $s \in \mathcal{S} \setminus \{b\}$, there are at least $M - 20$ globally surviving copies, and for $b$, there are at least $Nm - 60m$ globally surviving copies.
    
    \paragraph{Step 1: Surviving copies enforce the original structure.}
    Fix a clause $\cl{c} \in C$ and consider the globally surviving copies of symbols in $\mathcal{S}_{\cl{c}}$ (the 18 symbols belonging to clause $\cl{c}$). We use the type-wise bounds from the proof of \Cref{lem:global-cubic-survival}. The three truth-marker symbols lose at most $4+20+2=26$ indices in total. The three literal-occurrence wires lose at most $3\cdot(8+4+4)=48$ indices in total. The six OR symbols lose at most $2+8+2+2+6+2=22$ indices in total. Hence, by a union bound, there are at least $M-96$ indices $1\leq j\leq M$ such that all 18 symbols in $\mathcal{S}_{\cl{c}}$ have their $j$-th copy surviving \emph{all} cubic terms. We denote the surviving copies by $\mathcal{S}_{\cl{c}}^{(j)} = \{s^{(j)} : s \in \mathcal{S}_{\cl{c}}\}$ where $j \in J_{\cl{c}} \subseteq \{1, 2, \dots, M\}$ is the index set of size at least $M - 96$. Moreover, since $N>60$, there exists a global base copy $b^{(r)}$ that survives all cubic terms involving $b^{(r)}$. 
    
    \paragraph{Step 2: Applying the original structural lemmas.}  
    The key observation is that, given a clause $\cl{c}$ and the surviving copies $\mathcal{S}_{\cl{c}}^{(j)}, j\in J_{\cl{c}}$ with the fixed base copy $b^{(r)}$, the ``tripartite complete graph'' that arose from the original cubic term in $f_{\Phi}$ actually enforces the same geometric constraints as in the original polynomial. Specifically, suppose a triple $(b^{(i)}, u^{(j)}, v^{(k)})$ survives the original cubic term $b\,u\,v$ in $f_\Phi$. The quadratic terms arising from the tripartite complete graph in $f'_\Phi$ and the surviving condition ($b^{(i)}, u^{(j)}, v^{(k)}$ are on distinct wires) enforce overlap constraints $\supp(b^{(i)}) \cap \supp(u^{(j)}) \ne \emptyset$, $\supp(b^{(i)}) \cap \supp(v^{(k)}) \ne \emptyset$, and $\supp(u^{(j)}) \cap \supp(v^{(k)}) \ne \emptyset$, exactly the same as the CCZ gate in $f_\Phi$ by Helly's theorem (\cref{prelim:helly}).

    Therefore, for each surviving copy index $j \in J_{\cl{c}}$, the surviving copies $\mathcal{S}_{\cl{c}}^{(j)}$ together with the fixed base copy $b^{(r)}$ satisfy the same geometric constraints as in \Cref{lem:qubit-iso,lem:good-cl}. Hence, we can similarly define a clause ``copy'' $\cl{c}^{(j)}$ to be a \emph{good clause} if $\mathcal{S}_{\cl{c}}^{(j)}$ (with $j\in J_{\cl{c}}$) uses only $6$ wires. Further, given a good clause copy $\cl{c}^{(j)}$, we can define the \emph{encoded assignment} as in \Cref{def:encode}, which is well-defined by \Cref{lem:encode-wd}, also crucially using the definition of surviving copies. The arguments in \cref{lem:var-rlz} also go through by showing that for any variable $x\in X$, all occurrences $x_{\cl{c_1}^{(j_1)}}$ and $x_{\cl{c_2}^{(j_2)}}$ encode the same truth value for different \emph{good} clauses \(\cl{c_1}^{(j_1)},\cl{c_2}^{(j_2)}\), with $j_1 \in J_{\cl{c_1}}$ and $j_2 \in J_{\cl{c_2}}$. The correctness of OR gadgets in \cref{lem:or-rlz,lem:cl-rlz} also follows.
    
    \paragraph{Step 3: Wire count lower bound.}
    We now apply the counting argument from \Cref{thm:maxsat-reduction} to the surviving clause copies. 
    
    By the definition of global surviving, surviving copies of the same symbol cannot share wires. Furthermore, for different $j_1,j_2 \in J_{\cl{c}}$, symbols in $\mathcal{S}_{\cl{c}}^{(j_1)}$ and $\mathcal{S}_{\cl{c}}^{(j_2)}$ cannot share wires either, because by continuity we can focus on two consecutive symbols from the two sets along a wire. Their product must appear in $f'_\Phi$, then either they are copies of the same symbol (impossible as discussed before), or the product comes from a ``matching'' of quadratic terms in $f_{\Phi}$ (impossible since $j_1\neq j_2$), or a term in a ``tripartite complete graph'' from some cubic term in $f_{\Phi}$, contradicting the surviving condition. The same consecutive-symbol argument also excludes sharing between surviving clause copies with different clause labels, since the matching terms coming from original quadratic terms are clause-local and all cross-clause quadratic terms in $f'_\Phi$ arise from tripartite replacements of cubic synchronization terms. Finally, every globally surviving base copy is isolated, because $b$ has no original quadratic terms and every quadratic term involving a copy of $b$ comes either from the clique constraint on $b$-copies or from a tripartite replacement of a cubic term.

    Therefore, the surviving base symbol copies (at least $Nm-60m$) and all surviving clause symbol copies $\mathcal{S}_{\cl{c}}^{(j)}$ for different $(\cl{c}, j)$ pairs cannot share wires. Moreover, each clause copy $\cl{c}^{(j)}$ for $j \in J_{\cl{c}}$ requires only 6 wires when it is a good clause copy (defined above), and at least 7 wires otherwise.
    
    Let $g$ be the total number of good clause copies across all surviving clauses. Since each clause $\cl{c}$ has at least $|J_{\cl{c}}| \geq M - 96$ clause copies, the total number of clause copies is at least $\sum_{\cl{c}} |J_{\cl{c}}| \geq m(M - 96)$. The total wire count is at least:
    \[
        w(f'_\Phi) \;\ge\; Nm - 60m + 6g + 7\left(\sum_{\cl{c}} |J_{\cl{c}}| - g\right) \;=\; Nm - 60m + 7\sum_{\cl{c}} |J_{\cl{c}}| - g.
    \]
    
    If the circuit uses at most $K' = (7M+N)m - tM - 750m$ wires, then:
    \[
        Nm - 60m + 7\sum_{\cl{c}} |J_{\cl{c}}| - g \;\le\; K', \quad\text{which implies}\quad g \geq Nm - 60m + 7m(M-96) - K' = tM+18m >tM .
    \]

    Now, we extract a Boolean assignment from the good clause copies. By Step 2, each good clause copy encodes a well-defined truth assignment via \Cref{def:encode} and \Cref{lem:encode-wd}. Since $g > tM$ and each of the $m$ clauses has only $M$ copies, there must be more than $t$ distinct clauses each with at least one good copy. By \Cref{lem:cl-rlz}(a), each good clause copy corresponds to a satisfied clause under the encoded assignment.
    
    Therefore, the encoded assignment satisfies more than $t$ clauses, contradicting the assumption that every assignment satisfies at most $t$ clauses. This completes the proof that $w(f'_\Phi) \geq K'$.
\end{proof}

\begin{theorem}[Degree-2 Approximation Hardness]\label{thm:degree2-hard}
    For any constant $0<\epsilon<1/48$, the promise problem \approxproblem{($\tfrac{49}{48} - \epsilon$)} restricted to degree-2 polynomials is $\mathsf{NP}$-hard.
\end{theorem}

\begin{proof}
    We use the same reduction $\Phi \mapsto f_{\Phi}$ and twin polynomial construction $f'_\Phi$ as above.
    We set $N = 61$ (satisfying the requirement $N > 60$) and choose a sufficiently large $M$.
    
    If $\Phi$ is satisfiable (all $m$ clauses can be satisfied), then by \Cref{lem:twin-complete}, there exists a circuit realizing $f'_\Phi$ using at most
    \[
        w(f'_\Phi) \;\le\; (6M+N)m.
    \]
    
    Apply \Cref{thm:hastad} with the gap parameter $\delta=\epsilon$. If every assignment satisfies at most $(\tfrac{7}{8}+\delta)m$ clauses, then by \Cref{lem:twin-sound} with $t = (\tfrac{7}{8}+\delta)m$, every circuit realizing $f'_\Phi$ requires at least
    \[
        w(f'_\Phi) \;\ge\; (7M+N)m - \left(\frac{7}{8}+\delta\right)mM - 750m
        \;=\; \left(6 + \frac{1}{8} - \delta\right)mM + Nm - 750m.
    \]
    
    Let us analyze the approximation gap. Let $k = (6M+N)m$. Then in the \textsf{Yes} case, $w(f'_\Phi) \le k$.
    In the \textsf{No} case, since $\delta=\epsilon$, we have:
    \[
        w(f'_\Phi) \;\ge\; \left(6.125 - \delta\right)mM + Nm - 750m \geq (\frac{49}{48} - \epsilon) k +  (5\epsilon M - 750 - N)m.
    \]
    Now, for fixed $\epsilon>0$, choose $M > (750+N)/(5\epsilon)$. Restricting to sufficiently large instances (or padding by duplicating clauses) ensures $M \leq Nm$, and then the last term is positive. This proves that \approxproblem{($\tfrac{49}{48} - \epsilon$)} is \NP-hard even when restricted to degree-2 polynomials.
\end{proof}

\section{Algorithms for \problem{}}\label{sec:alg}
In this section, we present a nondeterministic polynomial-time algorithm to solve the search version of \problem{} defined in \cref{def:search-problem}: given a degree-3 polynomial $f \in \mathbb{F}_2[\mathcal{S}]$ without a constant term and an integer $k$, the algorithm either constructs a quantum circuit on at most $k$ qubits realizing $f$ as its associated polynomial or concludes that no such circuit exists.

Without loss of generality, we can always assume all symbols in $\mathcal{S}$ appear in $f$. Then, we denote $|\mathcal{S}| = n$ and the number of monomials in $f$ as $m \geq n / 3$. Furthermore, we assume the input polynomial $f$ is given as a list of its monomials in dictionary order. Then, in $O(m)$ time, one can construct an adjacency list representation for the connection hypergraph where the vertex set is $\mathcal{S}$ and each quadratic or cubic monomial in $f$ corresponds to an edge or hyperedge in the hypergraph. Then, in $O(m)$ time we also construct the connection graph $G_f$, the skeleton of the above hypergraph: the undirected graph with vertex set $\mathcal{S}$ and an edge $(u,v) \in E(G_f)$ if and only if the product $uv$ appears in a monomial of $f$ (i.e., as $uv$ or $uvw$ for some $w$).

As mentioned in the discussion after \cref{def:search-problem}, we only need to consider the case where $k \leq n$ and search for a circuit on exactly $k$ qubits. Moreover, if a cubic term $uvw$ appears in $f$, then any circuit realizing $f$ must implement a CCZ gate on $u,v,w$ on different wires. If the quadratic term $uv$ also appears in $f$, then the circuit must use a CZ gate on $u,v$ (because they are already on different wires), hence we obtain a circuit for $f-uv$ using the same number of qubits. The same holds for $f+uv$. Therefore, with another $O(m)$ time for cleaning up and adding back such redundant quadratic terms, we can assume without loss of generality that no quadratic term $uv$ appears in $f$ if there exists a cubic term $uvw$ in $f$ for some $w$.

Before presenting our algorithm, we first provide a structural characterization of valid quantum circuits corresponding to the associated polynomial $f$. This characterization shows that to reconstruct a valid quantum circuit, it suffices to find a good partition of the symbols into $k$ ordered lists $L_i$ (representing the order of symbols on each wire). A further check algorithm on $L_i$, using a constructed directed graph $\GH$ (representing required orders between Hadamard gates), can verify whether a valid circuit layout for $f$ (consistent with $L_i$) exists. An illustration is in \cref{fig:structure}.

\begin{lemma} \label{lem:structure}
    A valid quantum circuit with associated polynomial $f$ exists if and only if all symbols can be partitioned into $k$ ordered lists $L_i$ (each starting with a local source $\localendpoint{s}$ and ending with a local sink $\localendpoint{e}$), such that for every adjacent pair of real symbols in any $L_i$, the corresponding quadratic monomial appears in $f$, and the directed graph $\GH$ (ignoring self-loops and multiple edges) defined as follows has no directed cycles (i.e., admits a topological ordering):
    \begin{itemize}
        \item The vertices of $\GH$ (representing Hadamard gates) correspond to the ``links'' between consecutive symbols in the lists. Specifically, for each list $L_i = (\localendpoint{s}, u_{i,1}, \dots, u_{i, n_i}, \localendpoint{e})$, we introduce internal link vertices between adjacent elements. Note that the list endpoints $\localendpoint{s}$ and $\localendpoint{e}$ are identified with the global vertices $S$ and $E$ in $\GH$.
        \item For each list $L_i$, add directed edges $S \to H_{i,1} \to H_{i,2} \to \dots \to H_{i, n_i-1} \to E$ to $\GH$ (representing the linear order on the wire). Let $L(u)$ and $R(u)$ denote the link vertices immediately preceding and succeeding symbol $u$ in its list.
        \item For each quadratic monomial $uv$ in $f$, add directed edges $L(u) \to R(v)$ and $L(v) \to R(u)$ to $\GH$.
        \item For each cubic monomial $uvw$ in $f$, add directed edges enforcing pairwise overlaps: $L(u) \to R(v)$, $L(v) \to R(u)$, $L(u) \to R(w)$, $L(w) \to R(u)$, $L(v) \to R(w)$, and $L(w) \to R(v)$ to $\GH$.
    \end{itemize}
    Moreover, given a partition into $k$ ordered lists $L_i$ (using two-way linked lists), checking whether $\GH$ is acyclic and constructing a valid circuit if it is acyclic can be done in time $O(m)$.
\end{lemma}

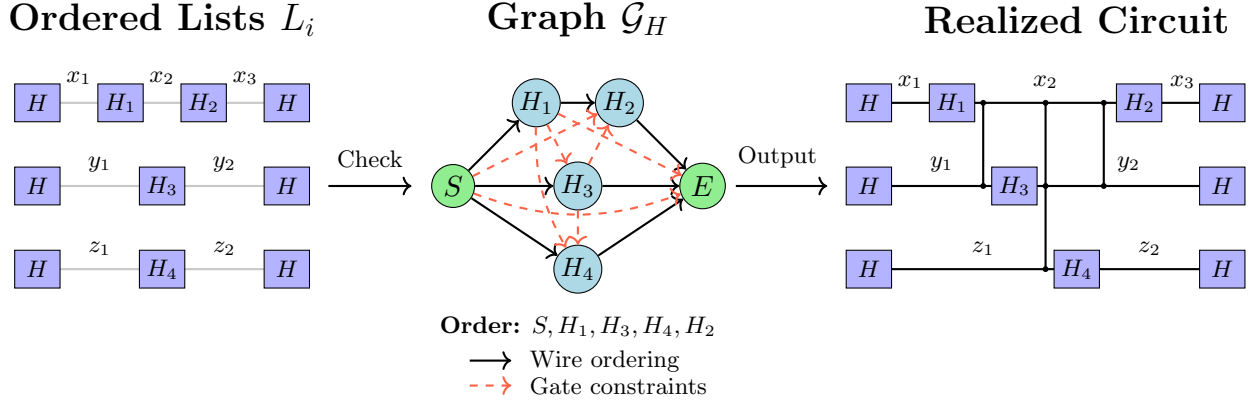
\begin{figure}[t]
    \centering
    
  	\tikzset{hgate/.style={draw,fill=blue!30,inner sep=0pt,minimum width=6mm,minimum height=5mm,text width=6mm,align=center,font=\footnotesize}}
    \begin{tikzpicture}[scale=1.1, baseline=(current bounding box.center)]
        \definecolor{wirecolor}{RGB}{200,200,200}
        \definecolor{hgatecolor}{RGB}{173,216,230}
        \definecolor{edgecolor}{RGB}{255,99,71}
        \definecolor{timecolor}{RGB}{144,238,144}
        
        \node[font=\Large\bfseries] at (-2,4) {Ordered Lists $L_i$};
        
        \draw[thick,wirecolor] (-3.5,3) -- (-0.5,3);
  \node[hgate] at (-3.5,3) {$H$};
  \node[hgate] at (-2.5,3) {$H_1$};
  \node[hgate] at (-1.5,3) {$H_2$};
  \node[hgate] at (-0.5,3) {$H$};
        \node[above,font=\footnotesize] at (-3,3.1) {$x_1$};
        \node[above,font=\footnotesize] at (-2,3.1) {$x_2$};
        \node[above,font=\footnotesize] at (-1,3.1) {$x_3$};
  
        \draw[thick,wirecolor] (-3.5,2) -- (-0.5,2);

        \node[hgate] at (-3.5,2) {$H$};
        \node[hgate] at (-2,2) {$H_3$};
        \node[hgate] at (-0.5,2) {$H$};
        \node[above,font=\footnotesize] at (-2.75,2.05) {$y_1$};
        \node[above,font=\footnotesize] at (-1.25,2.05) {$y_2$};
        \draw[thick,wirecolor] (-3.5,1) -- (-0.5,1);
  \node[hgate] at (-3.5,1) {$H$};
  \node[hgate] at (-2,1) {$H_4$};
  \node[hgate] at (-0.5,1) {$H$};
        \node[above,font=\footnotesize] at (-2.75,1.05) {$z_1$};
        \node[above,font=\footnotesize] at (-1.25,1.05) {$z_2$};
        
        \draw[->,thick,black] (0,2) -- (1,2);
        \node[above,font=\footnotesize] at (0.5,2.1) {Check};
        
        \node[font=\Large\bfseries] at (3,4) {Graph $\GH$};
        
        \node[draw,circle,fill=timecolor,inner sep=2pt] (start) at (1.5,2) {$S$};
        
        \node[draw,circle,fill=hgatecolor,inner sep=1pt,font=\small] (h1) at (2.5,3) {$H_1$};
        \node[draw,circle,fill=hgatecolor,inner sep=1pt,font=\small] (h2) at (3.5,3) {$H_2$};
        \node[draw,circle,fill=hgatecolor,inner sep=1pt,font=\small] (h3) at (3,2) {$H_3$};
        \node[draw,circle,fill=hgatecolor,inner sep=1pt,font=\small] (h4) at (3,1) {$H_4$};
        
        \node[draw,circle,fill=timecolor,inner sep=2pt] (end) at (4.5,2) {$E$};
        
        \draw[->,thick] (start) -- (h1);
        \draw[->,thick] (h1) -- (h2);
        \draw[->,thick] (start) -- (h3);
        \draw[->,thick] (start) -- (h4);
         \draw[->,thick] (h2) -- (end);
         \draw[->,thick] (h3) -- (end);
        \draw[->,thick] (h4) -- (end);
        
        \draw[->,thick,edgecolor,dashed] (start) to (h2);
        \draw[->,thick,edgecolor,dashed] (h1) to (h3);
         \draw[->,thick,edgecolor,dashed] (h1) to (end);
           \draw[->,thick,edgecolor,dashed] (h3) to (h2);
        \draw[->,thick,edgecolor,dashed] (h1) to[bend right=15] (h4);
        \draw[->,thick,edgecolor,dashed] (h3) to (h4);
        \draw[->,thick,edgecolor,dashed] (start) to[bend right=20] (end);
     
        \draw[->,thick] (1.7,-0.1) -- (2.2,-0.1);
        \node[font=\footnotesize,right] at (2.3,-0.1) {Wire ordering};
        \draw[->,thick,edgecolor,dashed] (1.7,-0.4) -- (2.2,-0.4);
        \node[font=\footnotesize,right] at (2.3, -0.4) {Gate constraints};
        
        \node[font=\footnotesize] at (3,0.3) {\textbf{Order:} $S,H_1,H_3,H_4,H_2$};
        \draw[->,thick,black] (4.9,2) -- (6.0,2);
        \node[above,font=\footnotesize] at (5.4,2.1) {Output};

        \node[font=\Large\bfseries] at (9,4) {Realized Circuit};

        \foreach \y in {3,2,1} {
            \draw[thick] (6.5,\y) -- (10.75,\y);
        }

  \node[hgate] at (6.5,3) {$H$};
  \node[hgate] at (6.5,2) {$H$};
  \node[hgate] at (6.5,1) {$H$};

  \node[hgate] at (7.5,3) {$H_1$};
  \node[hgate] at (8.25,2) {$H_3$};
  \node[hgate] at (9,1) {$H_4$};
  \node[hgate] at (9.75,3) {$H_2$};

  \node[hgate] at (10.75,3) {$H$};
  \node[hgate] at (10.75,2) {$H$};
  \node[hgate] at (10.75,1) {$H$};

        \draw[thick] (7.875,3) -- (7.875,2);
        \fill[black] (7.875,3) circle (1pt);
        \fill[black] (7.875,2) circle (1pt);

        \draw[thick] (9.325,3) -- (9.325,2);
        \fill[black] (9.325,3) circle (1pt);
        \fill[black] (9.325,2) circle (1pt);

        \draw[thick] (8.625,3) -- (8.625,1);
        \foreach \y in {3,2,1} {
            \fill[black] (8.625,\y) circle (1pt);
        }

        \node[above,font=\footnotesize] at (7,3.05) {$x_1$};
        \node[above,font=\footnotesize] at (8.625,3.05) {$x_2$};
        \node[above,font=\footnotesize] at (10.275,3.05) {$x_3$};

        \node[above,font=\footnotesize] at (7.375,2.05) {$y_1$};
        \node[above,font=\footnotesize] at (9.625,2.05) {$y_2$};

        \node[above,font=\footnotesize] at (7.875,1.05) {$z_1$};
        \node[above,font=\footnotesize] at (9.875,1.05) {$z_2$};
    \end{tikzpicture}
    \caption{\small Illustration of the check algorithm in \cref{lem:structure} for associated polynomial $f = x_1x_2 + x_2x_3 + y_1y_2 + z_1z_2 + x_2y_1 + x_2y_2 + x_2y_2z_1$. Left: given ordered lists $L_i$ with symbols on each wire. Middle: Directed graph $\GH$ is constructed, where vertices represent Hadamard gates, black solid edges enforce wire ordering, and red dashed edges represent constraints given by quadratic and cubic monomials (ignoring self-loops or multiple edges). Right: Realized quantum circuit with Hadamard, CZ, and CCZ gates arranged according to the topological order of $\GH$. In this example, $\GH$ is acyclic and the check algorithm constructs the valid quantum circuit accordingly.}
    \label{fig:structure}
    \end{figure}

\begin{proof}
    ($\Rightarrow$) Suppose a valid $k$-qubit circuit $\mathcal{C}$ realizing $f$ exists. We can partition the $n$ symbols into $k$ ordered lists $L_1, \ldots, L_k$ based on the symbols assigned to each wire in temporal order. Then clearly adjacent symbols on each list correspond to quadratic monomials in $f$, since they are connected by Hadamard gates on the same wire.

    Further, the physical realization of the circuit implies a global time-ordering of all Hadamard gates. Let $t(H)$ be the time coordinate of a Hadamard gate $H$ in $\mathcal{C}$. We show that $t$ induces a topological ordering of $\GH$ by verifying that all directed edges respect the time ordering:
    \begin{itemize}
        \item For any list $L_i$, the symbols appear sequentially, so the Hadamard gates satisfy $t(H_{i,j}) < t(H_{i,j+1})$. Thus, all list edges in $\GH$ respect the time ordering.
        \item For any quadratic term $uv$ in $f$, if it is implemented by a Hadamard gate, then $u,v$ must be adjacent symbols on the same list. Assume $u$ comes first in the list, then $R(u)=L(v)$, so $L(u)$, $R(u)$, and $R(v)$ are consecutive Hadamard gates on the same wire, hence $t(L(u)) < t(R(u)) < t(R(v))$ and $L(v)\to R(u)$ is a self-loop and can be ignored, thus respecting the time order. Otherwise, it is implemented by a CZ gate in $\mathcal{C}$, hence their active intervals, $\supp(u) = (t(L(u)), t(R(u)))$ and $\supp(v) = (t(L(v)), t(R(v)))$, must overlap. Overlap implies that the start of one cannot be later than the end of the other, i.e., $t(L(u)) < t(R(v))$ and $t(L(v)) < t(R(u))$. Thus, the edges $L(u) \to R(v)$ and $L(v) \to R(u)$ in $\GH$ respect the time ordering.
        \item For any cubic term $(u,v,w)$, the circuit must implement a CCZ gate, requiring the mutual intersection of three intervals, which implies pairwise intersection. Thus, the corresponding edges in $\GH$ respect the time ordering.
    \end{itemize}
    \noindent ($\Leftarrow$) Suppose there exists a partition into $k$ lists $L_i$ such that $\GH$ is acyclic. A directed acyclic graph (DAG) always admits a topological ordering~\cite{Kahn1962}. Let $\pi: V(\GH) \to \{1, \ldots, |V(\GH)|\}$ be such a topological ordering.
    We construct a quantum circuit with associated polynomial $f$ as follows:
    \begin{enumerate}
        \item Place the $k$ wires. For each list $L_i$, place Hadamard gates $H_{i,j}$ at time $\pi(H_{i,j})$. The wire edges $S\to H_{i,1}, H_{i,j} \to H_{i,j+1}, H_{i, n_i-1} \to E$ ensure increasing time coordinates. Then, the symbol $u$ occupies the time interval $I_u = (\pi(L(u)), \pi(R(u)))$.
        \item For every linear term $u$ in $f$, place a Z gate acting on the wire of $u$ at any time $t \in I_u$.
        \item For every quadratic term $uv$ in $f$ that does not correspond to adjacent symbols on the same list (adjacent product already constructed in $f$ by item 1), the edges $L(u) \to R(v)$ and $L(v) \to R(u)$ in $\GH$ (since $u,v$ are not adjacent, $L(u), R(u), L(v),R(v)$ must be different Hadamards) ensure $\pi(L(u)) < \pi(R(v))$ and $\pi(L(v)) < \pi(R(u))$. This implies $I_u \cap I_v \neq \emptyset$. We can thus place a CZ gate acting on wires of $u$ and $v$ at any time $t \in I_u \cap I_v$.
        \item For every cubic term $uvw$, the cleanup for redundant quadratic subterms guarantees that none of $uv, uw, vw$ appear in $f$ and hence are not adjacent in the lists. Thus, the six edges in $\GH$ ensure $\pi(L(u)) < \pi(R(v))$, $\pi(L(v)) < \pi(R(u))$, $\pi(L(u)) < \pi(R(w))$, $\pi(L(w)) < \pi(R(u))$, $\pi(L(v)) < \pi(R(w))$, and $\pi(L(w)) < \pi(R(v))$. This implies that the intervals $I_u, I_v, I_w$ pairwise intersect. By Helly's theorem in \cref{prelim:helly} for one-dimensional intervals, pairwise intersection of three intervals implies a common intersection. Thus, $I_u \cap I_v \cap I_w \neq \emptyset$, and we can place a CCZ gate at any time $t \in I_u \cap I_v \cap I_w$.
    \end{enumerate}
    Note that we have assumed all symbols appear in $f$. One can check that the constructed circuit indeed has associated polynomial $f$, hence is a valid realization.

    Finally, given a partition into $k$ ordered lists $L_i$ (using two-way linked lists), constructing $\GH$ involves constructing $O(n)$ new vertices and adding $O(m)$ directed edges; each can be done in $O(1)$ time assuming $f$ is given as a sequential list, which can be efficiently turned into an adjacency list for graph representation. Then, using a direct topological sort, which runs in time $O(m+n) = O(m)$ using Kahn's algorithm~\cite{Kahn1962}, we can validate whether $\GH$ is acyclic and even obtain a topological ordering if it is acyclic. Since the topological ordering directly gives a timeline arrangement for the circuit, we can add the $O(m)$ gates onto it and construct the valid circuit in time $O(m)$ as well.
\end{proof}

\subsection{Nondeterministic Algorithm}
In this section, we present our nondeterministic polynomial-time algorithm with few nondeterministic bits for the search version of \problem{}.

Based on the previous \cref{lem:structure}, finding a valid quantum circuit reduces to finding a good partition of symbols into $k$ ordered lists $L_i$ satisfying the conditions in \cref{lem:structure}. To this end, the algorithm guesses two sets of symbols, $S_{\textnormal{head}}$ and $S_{\textnormal{tail}}$, naturally representing the symbols at the start and end of the $k$ wires (hence witness size $2 \log_2 \binom{n}{k}$).
Given these sets, the algorithm attempts to partition all symbols into $k$ ordered lists (representing the wires) using a deterministic greedy procedure. The main idea is to iteratively extend the wires by adding symbols one by one.

The procedure works by maintaining a set of ``current head'' symbols, initialized to $S_{\textnormal{head}}$. In each step, it identifies a head symbol $v \in \mathsf{CurrentHeads}$ that has \emph{exactly} one neighbor among the unassigned symbols (symbols currently not inside any list) in the graph $G_f$ (whose vertices are symbols and edges correspond to pairs appearing in some monomial in $f$). This unique neighbor property is crucial: in the underlying valid circuit, the only interaction of $v$ with unassigned symbols is that with the immediate successor $u$ on the same wire. If such a $v$ is found, the neighbor $u$ is appended to the same $L_i$ as $v$ and the set of heads is updated by replacing $v$ with $u$. If a head symbol is in $S_{\textnormal{tail}}$, that wire is marked complete. An illustration is in \cref{fig:heuristic}.
\begin{figure}[!htbp]
\centering
\begin{tikzpicture}[xscale=1.5, yscale=1.5, baseline=(current bounding box.center)]
  \definecolor{headcolor}{RGB}{255,182,193}
  \definecolor{assignedcolor}{RGB}{144,238,144}
  \definecolor{currentcolor}{RGB}{255,255,0}
  \definecolor{tailcolor}{RGB}{173,216,230}
  
  \foreach \y in {2,1,0} {
    \draw[thick] (0,\y) -- (8,\y);
  }

  \foreach \y in {2,1,0} {
    \node[draw,fill=blue!20,inner sep=5pt,font=\small] at (0,\y) {$H$};
    \node[draw,fill=blue!20,inner sep=5pt,font=\small] at (8,\y) {$H$};
  }

  \foreach \pos/\y in {1/2, 6/2, 3/1, 5/0} {
    \node[draw,fill=blue!20,inner sep=5pt,font=\small] at (\pos,\y) {$H$};
  }

  \draw[thick] (4,2) -- (4,0);
  \foreach \y in {2,1,0} {
    \fill[black] (4,\y) circle (2pt);
  }

  \draw[thick] (2,2) -- (2,1);
  \foreach \y in {2,1} {
    \fill[black] (2,\y) circle (2pt);
  }

    \draw[thick] (5,2) -- (5,1);
  \foreach \y in {2,1} {
    \fill[black] (5,\y) circle (2pt);
  }

  \node[above,fill=headcolor,rounded corners=2pt,inner sep=2pt] at (0.5,2.15) {$x_1$};
  \node[above,fill=headcolor,rounded corners=2pt,inner sep=2pt] at (1.5,1.15) {$y_1$};
  \node[above,fill=headcolor,rounded corners=2pt,inner sep=2pt] at (2.5,0.15) {$z_1$};
  
  \node[above,fill=assignedcolor,rounded corners=2pt,inner sep=2pt] at (3.5,2.15) {$x_2$};
  
  \node[above,fill=tailcolor,rounded corners=2pt,inner sep=2pt] at (5.5,1.15) {$y_2$};
  
  \node[above,fill=tailcolor,rounded corners=2pt,inner sep=2pt] at (7.0,2.15) {$x_3$};
  \node[above,fill=tailcolor,rounded corners=2pt,inner sep=2pt] at (6.5,0.15) {$z_2$};

  \draw[->,thick,red] (0.8,2.3) -- (3.2,2.3);
\draw[->,thick,red] (1.8,1.3) -- (5.2,1.3);
  \draw[->,thick,red] (2.8,0.3) -- (6.2,0.3);
\draw[->,thick,red] (3.8,2.3) -- (6.7,2.3);
\node[below,font=\footnotesize] at (5.25,2.65) {\textbf{Step 4:} $x_2 \rightarrow x_3$};

  \node[below,font=\footnotesize] at (2,2.65) {\textbf{Step 1:} $x_1 \rightarrow x_2$};
  \node[below,font=\footnotesize] at (3.5,1.65) {\textbf{Step 2:} $y_1\rightarrow \ y_2$};
  \node[below,font=\footnotesize] at (4.5,0.65) {\textbf{Step 3:} $z_1 \rightarrow z_2$};

\end{tikzpicture}
\caption{\small Illustration of the nondeterministic algorithm's progress for $f = x_1x_2 + x_2x_3 + y_1y_2 + z_1z_2 + x_2y_1 + x_2y_2 + x_2y_2z_1 $. Given the witness $S_{\text{head}} = \{x_1, y_1, z_1\}$ (pink) and $S_{\text{tail}} = \{x_3, y_2, z_2\}$ (blue), the algorithm deterministically finds the circuit. At each step, it finds a current head symbol (green) with degree $1$ in the graph $G_f$ and extends the wire. The first step detects (for example) $x_1$ and replaces $x_1$ by $x_2$ in the current head, then $y_1$ by $y_2$. After replacement, the algorithm detects that the wire should be complete. Then, it detects $z_1$ and replaces $z_1$ by $z_2$, marks the wire complete, and finally replaces $x_2$ by $x_3$. All symbols are assigned to wires successfully, and the output lists are $L_1 = (x_1, x_2, x_3)$, $L_2 = (y_1, y_2)$, and $L_3 = (z_1, z_2)$. The algorithm then feeds the ordered lists $L_i$ to the check algorithm in \cref{lem:structure}, which verifies whether positions of $H$ gates following $L_i$ on each wire can be arranged to realize $f$.
}
\label{fig:heuristic}
\end{figure}
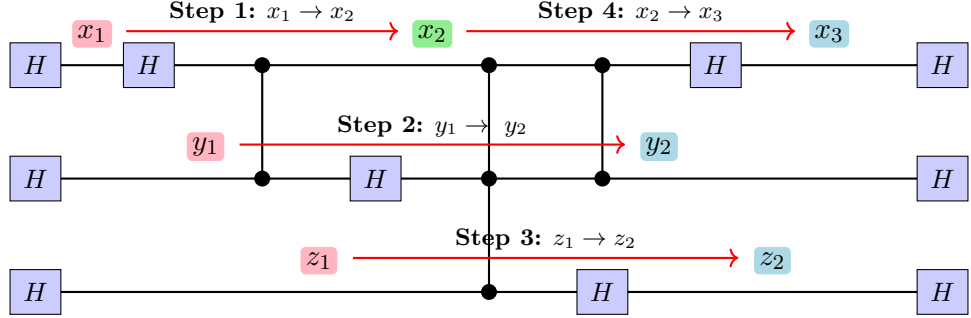

If all symbols are successfully partitioned into $k$ lists, the algorithm verifies the solution using the Check algorithm in \cref{lem:structure}. Recall that the Check algorithm runs in time $O(m)$ and either constructs a valid circuit or concludes no valid circuit exists for the given lists. 

The algorithm is summarized in \cref{alg:nondeterministic}.

    \begin{algorithm}[!ht]
    \caption{Nondeterministic algorithm for \problem{}}
    \label{alg:nondeterministic}
    \SetAlgoLined
    \KwIn{A degree-3 polynomial $f$ on $\mathcal{S}$, an integer $k$, and two sets $S_{\textnormal{head}}, S_{\textnormal{tail}} \subseteq \mathcal{S}$ of size $k$ each as witness (total $2 \log_2 \binom{n}{k}$ bits)}
    \KwOut{A quantum circuit on at most $k$ qubits or ``No solution''}
    \KwMaintain{The ordered lists $L_1, \ldots, L_k$ representing the wires};
    
    Construct graph $G_f$ where vertices are $\mathcal{S}$ and edges correspond to pairs appearing in some monomial in $f$\;
    Initialize $k$ lists $L_i$, each with a symbol from $S_{\textnormal{head}}$\;
    $\mathsf{CurrentHeads} \gets S_{\textnormal{head}}$;
    
    \While{$ \bigcup_{i} L_i \neq \mathcal{S} $ \label{whileloop}}{
        \While{$ \exists v\in \mathsf{CurrentHeads} \cap S_{\textnormal{tail}} $}{
            Mark $v$'s list as complete\;
            $\mathsf{CurrentHeads}\gets\mathsf{CurrentHeads} \setminus \{ v \};   $
        }
        \KwFind $v \in  \mathsf{CurrentHeads}$ that has exactly one neighbor $u$ in $G_f$ such that $u \notin \bigcup_{i} L_i$ \;\label{greedyextension}
        \eIf{no such $v$ exists}{
            \Return ``No solution''\;
        }{
            Let $u$ be such neighbor of $v$\;
            Append $u$ to $v$'s list $L_i$\;
            $\mathsf{CurrentHeads} \gets (\mathsf{CurrentHeads} \setminus \{ v \}) \cup \{ u \}$\; 
        }
    }
    \Return $\textnormal{Check}(f,  \{L_1, \ldots, L_k\})$; \hfill \Comment{Check($\cdot$) is the check algorithm defined in \cref{lem:structure}}
    \end{algorithm}

Now we prove the correctness of the greedy extension with the unique neighbor step in \cref{alg:nondeterministic}. The proof relies crucially on the time ordering of gates in a valid circuit realization.
\begin{lemma}[Correctness of Greedy Extension] \label{lem:degree1}
If the witness sets ($S_{\textnormal{head}}$ and $S_{\textnormal{tail}}$) are exactly the symbols at the start and end of each wire in some valid $k$-qubit circuit realizing $f$, then the condition in \textnormal{\cref{greedyextension}} of \textnormal{\cref{alg:nondeterministic}} is always satisfied; that is, there exists a current head symbol $v \in \mathsf{CurrentHeads}$ with degree exactly $1$ to symbols not in $\bigcup_i L_i$ in $G_f$.
\end{lemma}

\begin{proof}
Let $\mathcal{C}$ be a valid $k$-qubit circuit realizing $f$ with the symbols at the start and end of each wire exactly $S_{\textnormal{head}}$ and $S_{\textnormal{tail}}$. We prove by induction that at each step of the algorithm, \textnormal{\cref{greedyextension}} of \textnormal{\cref{alg:nondeterministic}} is satisfied. In fact, we prove a stronger statement: before each iteration of the loop, the ordered lists $L_i$ correspond exactly to the symbols assigned to each wire in $\mathcal{C}$ up to some timestamp (which may be different for different wires). In other words, they are prefixes of the symbols on each wire in $\mathcal{C}$. Moreover, $\mathsf{CurrentHeads}$, maintained by the algorithm, corresponds exactly to the rightmost symbols of the unfinished prefixes; completed wires have already reached their tail symbols and have been removed from $\mathsf{CurrentHeads}$. Then, the set of symbols not in $\bigcup_i L_i$ coincides with all symbols that appear strictly after this cross-section in $\mathcal{C}$.

\noindent \textbf{Base Case:} Initially, $\mathsf{CurrentHeads} = S_{\textnormal{head}}$. Since $S_{\textnormal{head}}$ consists of the first symbol on each wire, each $L_i$ contains exactly the first symbol on each wire in $\mathcal{C}$. Thus, the hypothesis holds. 

\noindent \textbf{Inductive Step:} Assume the hypothesis holds. We show there exists $v \in \mathsf{CurrentHeads}$ with degree 1 to symbols not in $\bigcup_i L_i$. 
Consider the symbols in $\mathsf{CurrentHeads}$. In the valid circuit $\mathcal{C}$, each of these symbols occupies a time interval. Among these currently active symbols, let $v$ be the one that \emph{ends earliest} (i.e., its rightmost Hadamard gate appears first in the time ordering of the circuit). 

We claim $v$ has no edges in $G_f$ to any symbol $z \notin \bigcup_i L_i$ on a \emph{different} wire. Suppose for contradiction $v$ has an edge to $z \notin \bigcup_i L_i$ where $z$ is on a different wire. This edge corresponds to a quadratic or cubic term, which requires the time intervals of $v$ and $z$ to overlap. However, $z \notin \bigcup_i L_i$ means that $z$ starts after some current head $v' \in \mathsf{CurrentHeads}$ (on $z$'s wire) ends. Since $v$ corresponds to the interval ending \emph{earliest} among all current heads, $v$ ends before $v'$ ends, and thus before $z$ starts. Hence, $v$ and $z$ cannot overlap, a contradiction.

Therefore, the only possible neighbor of $v$ not in $\bigcup_i L_i$ is the symbol $u$ immediately following $v$ on the \emph{same} wire. This edge $(v, u)$ always exists unless $v$ is the last symbol on its wire, in which case $v \in S_{\textnormal{tail}}$ and would be handled by the inner while loop. Thus, $v$ has exactly degree 1 to symbols not in $\bigcup_i L_i$. The algorithm then updates $\mathsf{CurrentHeads}$ by replacing $v$ with $u$, which corresponds to advancing the timestamp on $v$'s wire past the Hadamard gate separating $v$ and $u$. The inductive hypothesis is maintained.

Thus, the algorithm always finds a valid candidate for extension until all symbols are assigned.
\end{proof}
Finally, we conclude the correctness and witness size of \cref{alg:nondeterministic} in the following theorem.
\begin{theorem}\label{thm:nondeterministic}
The nondeterministic algorithm correctly solves the search version of \textnormal{\problem{(f,k)}} in polynomial time with a witness of size $2 \log_2 \binom{n}{k}=O(k \log(en/k))$ bits. 
\end{theorem}

\begin{proof}
\textbf{Correctness:} If a valid $k$-qubit circuit $\mathcal{C}$ whose associated polynomial is $f$ exists, then there exist valid head and tail sets $S_{\textnormal{head}}$ and $S_{\textnormal{tail}}$, corresponding to the symbols at the start and end of each wire. Given a witness encoding these correct sets, by \cref{lem:degree1}, the condition of greedy extension in \cref{greedyextension} is always satisfied. Therefore, the while loop in \cref{whileloop} will reach its end, and all symbols will be assigned in the same order as $\mathcal{C}$ using the induction statement in \cref{lem:degree1}. By \cref{lem:structure}, the Check algorithm will verify the validity and produce one valid circuit realizing $f$ on $k$ qubits (not necessarily $\mathcal{C}$ itself).

\noindent \textbf{Soundness:} On the other hand, if the algorithm returns a circuit, then by correctness of \cref{lem:structure}, this circuit has associated polynomial $f$ and uses at most $k$ qubits. Hence soundness holds as well.

\noindent \textbf{Witness Complexity:} The witness consists of two sets $S_{\textnormal{head}}$ and $S_{\textnormal{tail}}$, each of size $k$ chosen from $n$ symbols. Each set can be encoded using $\log_2 \binom{n}{k}$ bits, giving a total witness size of $2 \log_2 \binom{n}{k} = O(\log_2 \binom{n}{k}) = O(k \log(en/k))$ bits.

\noindent \textbf{Time Complexity:} Given the witness, the greedy extension takes $\textnormal{poly}(n)$ time to process $n$ symbols, and the Check algorithm runs in $O(m) = \textnormal{poly}(n)$ time as stated in \cref{lem:structure}. Therefore, the algorithm runs in polynomial time given the witness.

\end{proof}

\begin{corollary}\label{cor:deterministic}
    There exists a deterministic algorithm that solves the search version of \textnormal{\problem{(f,k)}} in time $\binom{n}{k}^2 \cdot \textnormal{poly}(n)$, which is tight up to constant factors in the exponent under the Exponential Time Hypothesis (\textnormal{ETH}) when $k =\Theta(n)$.
\end{corollary}
\begin{proof}
    The deterministic algorithm invokes the nondeterministic algorithm by enumerating all witnesses. This leads to a total runtime of $\binom{n}{k}^2 \cdot \textnormal{poly}(n)$. When $k = \Theta(n)$, this becomes $2^{O(n)}$. Recall the conditional lower bound from \cref{thm:eth}, which states that under the Exponential Time Hypothesis, there exists a constant $\epsilon > 0$ such that there is no algorithm that solves \textnormal{\problem{(f,k)}} in time $2^{\epsilon n}$. This establishes the tightness of our deterministic algorithm up to constant factors in the exponent.
\end{proof}

\subsection{Efficient Fixed-Parameter Tractable Algorithm}\label{sec:fpt_algorithm}
In this section, we present an \textsf{FPT} algorithm for \problem{} with runtime $k^{O(k)} n$, whose high-level overview is presented in the technical overview in \cref{sec:fpt}. In the remainder of this section, we formalize the state representation, describe the DP transitions for each node type, prove their correctness, and analyze the time complexity in detail.

Given the definition of path decomposition (see \cref{def:tree_decomposition}), one can note that the greedy extension procedure \cref{greedyextension} in the proof of \cref{lem:degree1} essentially constructs a path decomposition of the graph $G_f$ with bags of size at most $k+1$, and hence width at most $k$, where $k$ is the number of wires. We formalize this observation in the following lemma.

\begin{lemma}\label{lem:pathdecomposition}
If there exists a valid $k$-qubit circuit realizing the associated polynomial $f$, then the connection graph $G_f$ (vertices are $\mathcal{S}$ and edges correspond to pairs appearing in some monomial in $f$) has pathwidth at most $k$.
\end{lemma}

\begin{proof}
    We can explicitly construct a path decomposition of $G_f$ by simulating the greedy extension procedure described in \cref{lem:degree1}. Since there exists a valid $k$-qubit circuit realizing the associated polynomial $f$, the condition in \cref{lem:degree1} is satisfied: the witness sets $S_{\textnormal{head}}$ and $S_{\textnormal{tail}}$ correspond to the symbols at the start and end of each wire in this circuit.
    
    Let $S_0, S_1, \dots, S_{n-k}$ be the sequence of the set $\mathsf{CurrentHeads}$ throughout the execution of the algorithm, where $S_0 = S_{\textnormal{head}}$. Note that when a wire is completed (i.e., its head symbol is in $S_{\textnormal{tail}}$), the corresponding symbol is removed from $\mathsf{CurrentHeads}$ and not replaced. However, we still include it in the sets $S_t$ for the purpose of constructing the path decomposition. Thus, each $S_t$ has size exactly $k$ for all $0 \leq t \leq n-k$.
    
    In each step $t \in \{1,\dots, n-k\}$, the algorithm selects a vertex $v \in S_{t-1}$ and replaces it with its neighbor $u$ on the same wire, resulting in $S_t = (S_{t-1} \setminus \{v\}) \cup \{u\}$.
    
    We define the bags of the path decomposition by setting $B_0=S_0$ and, for $1\leq t\leq n-k$, $B_t = S_{t-1} \cup S_t = S_{t-1} \cup \{u\}$, where $u$ is the symbol introduced at step $t$. Since $|S_0|=k$ and $|S_{t-1}| = k$, every bag has size at most $k+1$. Thus, the width of this decomposition is at most $k$.
    
    We now verify the properties of path decomposition:
    \begin{itemize}
        \item \textbf{Vertex Cover:} Every symbol in $\mathcal{S}$ is eventually assigned to a wire. It starts as a member of $S_{\textnormal{head}}$ or enters $\mathsf{CurrentHeads}$ at some step. Since every symbol is in some $S_t$, it belongs to at least one bag (using $B_0=S_0$ for the initial symbols).
        \item \textbf{Edge Cover:} 
        Consider any edge in $G_f$. If it connects $v$ and $u$ where $u$ follows $v$ on the same wire, then this edge is covered in the step $t$ where $v$ is replaced by $u$, as $\{v, u\} \subseteq B_t$.
        If the edge connects $x$ and $y$ on different wires, consider the first step $t$ at which at least one of $x,y$ belongs to $\mathsf{CurrentHeads}$. If both are already in $S_t$, the edge is covered by $B_t$ (with $B_0=S_0$ when $t=0$). Otherwise, assume without loss of generality that $x\in S_t$ and $y$ has not yet entered $\mathsf{CurrentHeads}$. If $x$ is already in $S_{\textnormal{tail}}$, then it remains in all subsequent $S_j$ for $j\ge t$. If not, then $x$ also remains in subsequent $S_j$ until $y$ enters $\mathsf{CurrentHeads}$; otherwise $x$ would have at least two neighbors among untouched symbols, namely $y$ and the next symbol on $x$'s wire, contradicting the greedy extension condition. Eventually all symbols enter $\mathsf{CurrentHeads}$, so there is a step $t'\ge t$ where both $x$ and $y$ are in $S_{t'}$, and the edge is covered by $B_{t'}$.
        \item \textbf{Running Intersection:} For any symbol $w$, the set of steps where $w \in S_t$ forms a contiguous interval (it enters once and then either leaves once or remains as a retained completed tail). Since $B_t = S_{t-1} \cup S_t$, the set of bags containing $w$ also forms a contiguous subsegment of the path.
    \end{itemize}
    Therefore, $(B_0, B_1, \dots, B_{n-k})$ is a valid path decomposition of $G_f$ with width at most $k$.
\end{proof}

With the pathwidth (and hence treewidth) upper bound at hand, establishing that \problem{} is \textsf{FPT} becomes a direct consequence of Courcelle's Theorem~\cite{Courcelle1990}. The existence of a valid partition into $k$ wires can be expressed in Monadic Second-Order logic ($\textnormal{MSO}_2$) over the natural edge-labelled incidence structure of the polynomial. Specifically, we check for the existence of $k$ disjoint ordered paths whose direct wire adjacencies correspond to quadratic monomials of $f$, and such that when the temporal constraints induced by all quadratic and cubic terms of $f$ are added, the resulting graph $\GH$ (as defined in \cref{lem:structure}) contains no directed cycles. Then, Courcelle's Theorem implies that this property can be checked in linear time, where the parameter is the treewidth of the graph and also the size of the $\textnormal{MSO}_2$ formula. Note that the size of the formula depends only on $k$ (the number of wires). Moreover, one can first use a fixed-parameter tractable algorithm to determine whether the treewidth of $G_f$ is at most $k$ (for example, use the exact treewidth algorithm in~\cite{Bodlaender1996}), return ``No'' if the treewidth exceeds $k$. Otherwise, apply Courcelle's Theorem to decide the problem in linear time, where the hidden constant depends on $k$. This shows that \problem{} is \textsf{FPT}.

This already answers the question affirmatively. However, since the hidden constant in Courcelle's Theorem can be multi-exponential, we would like to provide a more concrete dynamic programming algorithm that solves the problem with a more practical dependence on $k$, specifically $k^{O(k)}$. The first step is to compute a tree decomposition of $G_f$ with width promised to be bounded by $k$.

As mentioned, there exist exact algorithms for treewidth that are fixed-parameter tractable. Bodlaender~\cite{Bodlaender1996} provided an exact algorithm for computing treewidth in time $2^{O(k^3)} n$. Korhonen et al.~\cite{KorhonenLokshtanov2023} improve this runtime to $2^{O(k^2)} n^4$.

For our purposes, faster approximation algorithms suffice and are better for the dependence on $k$. Bodlaender et al.~\cite{Bodlaender2016} provided a $5$-approximation algorithm for treewidth that runs in time $2^{O(k)} n$. More recently, Korhonen~\cite{Korhonen2021} presented a single-exponential time $2$-approximation algorithm for treewidth. We state the result below.
\begin{theorem}[Theorem 1.1 in~\cite{Korhonen2021}]\label{thm:treewidth}
There is an algorithm that, given an $n$-vertex graph $G$ and an integer $k$, in time $2^{O(k)}n$ either outputs a tree decomposition of $G$ of width at most $2k + 1$ or determines that the
treewidth of $G$ is larger than $k$. 
\end{theorem}

By using the $2$-approximation algorithm from \cref{thm:treewidth} with parameter $k$, we can obtain a tree decomposition of $G_f$ with width $w \le 2k+1 = O(k)$ in time $2^{O(k)} n$. Furthermore, by transforming the tree decomposition into a nice one efficiently via \cref{lem:nice_td}, we can assume a nice tree decomposition of $G_f$ with width $w$ is provided.

Now, we can solve \problem{} using dynamic programming on the nice tree structure, following the standard approach for problems parameterized by treewidth~\cite{Cygan2015}. The algorithm processes the tree decomposition in a bottom-up manner, computing a set of valid partial circuit configurations for the symbols inside the subtree rooted at each node $x$. A partial configuration is defined over the set of symbols $V_x$ in the subtree, but to ensure it can be consistently extended to the rest of the graph, we only need to track information exposed at the boundary $B_x$ (the bag at node $x$). To this end, we define a set of dynamic programming \emph{states} $\mathcal{D}_x$ at each node $x$, where each state in $\mathcal{D}_x$ encodes the necessary information about the partial configuration on $V_x$ that is relevant to the boundary $B_x$.

The state consists of two components: $\mathcal{W}_x$ (wire structure) and $\mathcal{O}_x$ (ordering of Hadamard gates), where:
\begin{enumerate}
    \item $\mathcal{W}_x$ maintains $k$ ordered lists $W_i$ (each containing symbols from $B_x$ along with local endpoints $\localendpoint{s}$ and $\localendpoint{e}$), with edge type annotations between consecutive entries indicating: (0) no direct connection, (1) a direct Hadamard link, or (2) connected via forgotten symbols (i.e., vertices in $V_x \backslash B_x$).
    \item $\mathcal{O}_x$ maintains a topological ordering (as defined in \cref{def:topological_ordering}) of the boundary virtual Hadamard vertices induced by $\mathcal{W}_x$ according to \cref{lem:structure}, ensuring compatibility with the partial constraint graph $\GH[V_x]$.
\end{enumerate}
We formalize each state in $\mathcal{D}_x$ as a tuple $(\mathcal{W}_x, \mathcal{O}_x)$ defined as follows. In principle, each $\mathcal{W}_x$ is a partition of symbols and ordered-pair types, while each $\mathcal{O}_x$ is a topological ordering on the corresponding virtual gate vertices. We also include their intuitive meanings, which will be used in defining valid states shortly afterwards.

\begin{enumerate}
  \item \textbf{Wire Structure} $\mathcal{W}_x = (W_i,\tau)$: Consists of a
        partition of all symbols in $B_x$, together with $\localendpoint{s}$ and
        $\localendpoint{e}$, into $k$ ordered lists
        $W_i = (\localendpoint{s}, u_{i,1}, \dots, u_{i,n_i}, \localendpoint{e})$
        ($1\leq i\leq k$) where each $u_{i,j} \in B_x$\footnote{Note that not
        necessarily the product of consecutive symbols in $W_i$ appears in $f$, nor in
        $L_i$ below.}.
        For each ordered pair of consecutive entries in $W_i$, we maintain an
        ordered pair type $\tau:\text{Consecutive entries
        in } W_i\to \{0, 1, 2\}$ whose intuitive meanings are:
        \begin{itemize}
        \item $\tau(u, v) = 1$: Entries $u$ and $v$ are directly connected by
                        a Hadamard link on the same wire (i.e., no symbols in
                        $V_x \backslash B_x$ appear between them).
                        Furthermore, we require that
                        $u\neq \localendpoint{s}$, $v\neq \localendpoint{e}$, and the quadratic monomial
                        $uv$ appears in $f$.
                \item $\tau(u, v) = 2$: Entries $u$ and $v$ are connected via a path
                        using only forgotten symbols in $V_x \backslash B_x$ on the same
                        wire, where every consecutive pair of real symbols on the path is directly connected
                        by Hadamard gates.
                \item $\tau(u, v) = 0$: Entries $u$ and $v$ are not directly connected
                        by a Hadamard gate, and there are no forgotten symbols from
                        $V_x \backslash B_x$ between them on the wire.
    \end{itemize}
    
    \item \textbf{Topological Ordering} $\mathcal{O}_x$: We first introduce virtual vertices $\{S, E\} \cup \{L(u), R(u) \mid u \in B_x\}$ ($L(u)$ and $R(u)$ can be viewed as the left and right endpoints of the wire segment corresponding to symbol $u$ as in \cref{lem:structure}). Moreover, for ordered consecutive symbols $u, v$ with $\tau(u,v) = 1$, we identify $R(u) = L(v)$. Further, we denote the virtual vertex set after identification as $B_x(\tau)$. Then, the ordering $\mathcal{O}_x$ is a topological ordering on $B_x(\tau)$ with $S$ the minimum and $E$ the maximum. Intuitively, $\mathcal{O}_x$ represents a time ordering of Hadamard gates on the two sides of intervals represented by symbols in $B_x$.
\end{enumerate}

An example state is illustrated in \cref{fig:wx-ox}.
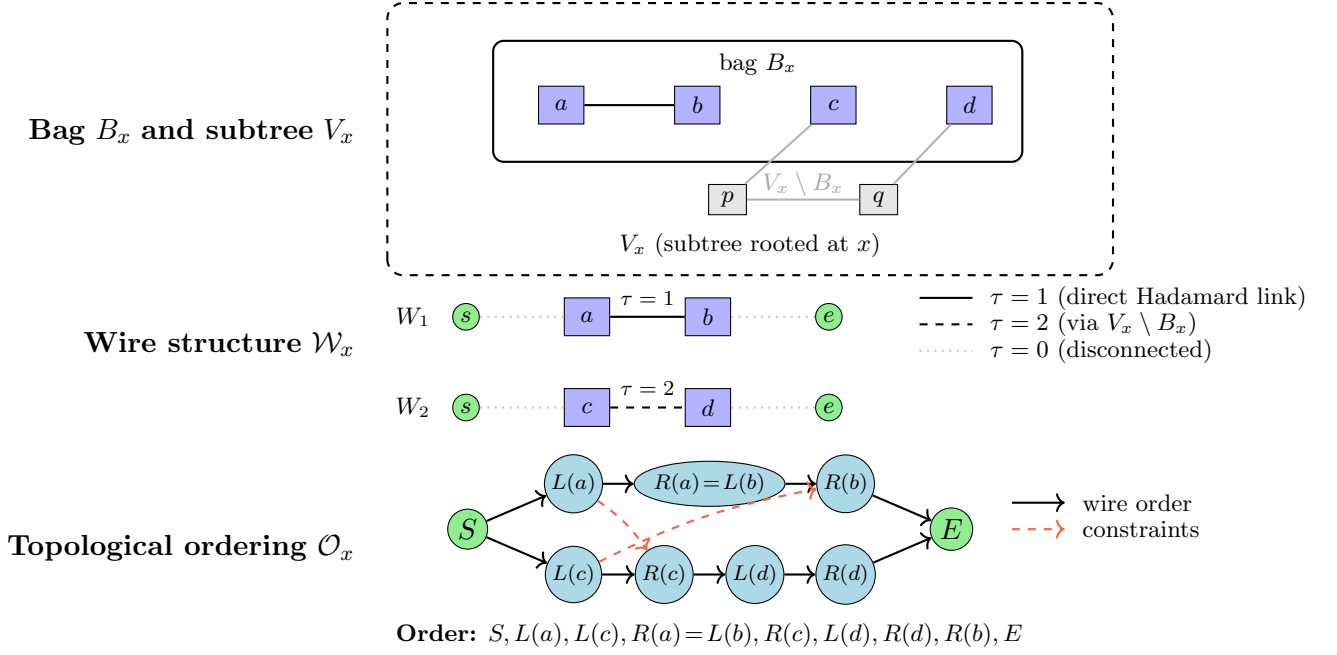
\begin{figure}[t]
    \centering
    \definecolor{wirecolor}{RGB}{200,200,200}
    \definecolor{hgatecolor}{RGB}{173,216,230}
    \definecolor{edgecolor}{RGB}{255,99,71}
    \definecolor{timecolor}{RGB}{144,238,144}
    \tikzset{
        sym/.style={draw,fill=blue!30,inner sep=0pt,minimum width=6mm,minimum height=5mm,text width=6mm,align=center,font=\footnotesize},
        endpoint/.style={draw,circle,fill=timecolor,inner sep=1.5pt,font=\footnotesize},
        ghost/.style={draw,fill=gray!20,inner sep=1.2pt,minimum width=5mm,minimum height=4mm,font=\footnotesize},
        vtx/.style={draw,circle,fill=hgatecolor,inner sep=1.2pt,font=\scriptsize}
    }
    \renewcommand{\arraystretch}{1.1}
    \begin{tabular}{@{}p{0.3\linewidth}p{0.7\linewidth}@{}}
        \raggedleft\textbf{Bag} $B_x$ \textbf{and subtree} $V_x$ &
        \begin{tikzpicture}[scale=1.0, baseline=(current bounding box.center)]
            \draw[dashed,rounded corners=6pt,thick] (-0.4,-0.8) rectangle (9.2,2.8);
            \node[font=\footnotesize] at (4.4,-0.4) {$V_x$ (subtree rooted at $x$)};
            
            \draw[rounded corners=4pt,thick] (1.0,0.7) rectangle (8.0,2.3);
            \node[font=\footnotesize] at (4.5,2.0) {bag $B_x$};
            
            \node[sym] (a) at (1.9,1.45) {$a$};
            \node[sym] (b) at (3.7,1.45) {$b$};
            \node[sym] (c) at (5.5,1.45) {$c$};
            \node[sym] (d) at (7.3,1.45) {$d$};
            
            \node[ghost] (p) at (4.1,0.2) {$p$};
            \node[ghost] (q) at (6.1,0.2) {$q$};
            \draw[thick,gray!60] (c) -- (p) -- (q) -- (d);
            \draw[thick] (a) -- (b);
            \node[font=\footnotesize,gray!70] at (5.1,0.4) {$V_x \setminus B_x$};
        \end{tikzpicture}
        \\[1.2cm]
        \raggedleft\textbf{Wire structure} $\mathcal{W}_x$ &
        \begin{tikzpicture}[scale=1.0, baseline=(current bounding box.center)]
            \node[endpoint] (s1) at (0.0,1.2) {$\localendpoint{s}$};
            \node[sym] (a2) at (1.6,1.2) {$a$};
            \node[sym] (b2) at (3.2,1.2) {$b$};
            \node[endpoint] (e1) at (4.8,1.2) {$\localendpoint{e}$};
            \draw[thick,wirecolor,dotted] (s1) -- (a2);
            \draw[thick] (a2) -- (b2);
            \draw[thick,wirecolor,dotted] (b2) -- (e1);
            \node[font=\footnotesize] at (-0.7,1.2) {$W_1$};
            
            \node[endpoint] (s2) at (0.0,0.0) {$\localendpoint{s}$};
            \node[sym] (c2) at (1.6,0.0) {$c$};
            \node[sym] (d2) at (3.2,0.0) {$d$};
            \node[endpoint] (e2) at (4.8,0.0) {$\localendpoint{e}$};
            \draw[thick,wirecolor,dotted] (s2) -- (c2);
            \draw[thick,dashed] (c2) -- (d2);
            \draw[thick,wirecolor,dotted] (d2) -- (e2);
            \node[font=\footnotesize] at (-0.7,0.0) {$W_2$};
            
            \node[font=\scriptsize] at (2.4,1.45) {$\tau=1$};
            \node[font=\scriptsize] at (2.4,0.25) {$\tau=2$};
            
            \draw[thick] (6.0,1.45) -- (6.7,1.45);
            \node[font=\footnotesize,right] at (6.8,1.45) {$\tau=1$ (direct Hadamard link)};
            \draw[thick,dashed] (6.0,1.1) -- (6.7,1.1);
            \node[font=\footnotesize,right] at (6.8,1.1) {$\tau=2$ (via $V_x \setminus B_x$)};
            \draw[thick,wirecolor,dotted] (6.0,0.75) -- (6.7,0.75);
            \node[font=\footnotesize,right] at (6.8,0.75) {$\tau=0$ (disconnected)};
        \end{tikzpicture}
        \\[1.2cm]
        \raggedleft\textbf{Topological ordering} $\mathcal{O}_x$ &
        \begin{tikzpicture}[scale=1.0, baseline=(current bounding box.center)]
            \node[draw,circle,fill=timecolor,inner sep=1.6pt] (S) at (0.0,1.0) {$S$};
            \node[vtx] (La) at (1.4,1.6) {$L(a)$};
            \node[draw,ellipse,fill=hgatecolor,font=\scriptsize,inner xsep=-1pt,inner ysep=2pt] (Rab) at (3.2,1.6) {$R(a)\!=\!L(b)$};
            \node[vtx] (Rb) at (5.0,1.6) {$R(b)$};
            \node[vtx] (Lc) at (1.4,0.4) {$L(c)$};
            \node[vtx] (Rc) at (2.6,0.4) {$R(c)$};
            \node[vtx] (Ld) at (3.8,0.4) {$L(d)$};
            \node[vtx] (Rd) at (5.0,0.4) {$R(d)$};
            \node[draw,circle,fill=timecolor,inner sep=1.6pt] (E) at (6.4,1.0) {$E$};
            
            \draw[->,thick] (S) -- (La);
            \draw[->,thick] (La) -- (Rab);
            \draw[->,thick] (Rab) -- (Rb);
            \draw[->,thick] (Rb) -- (E);
            
            \draw[->,thick] (S) -- (Lc);
            \draw[->,thick] (Lc) -- (Rc);
            \draw[->,thick] (Rc) -- (Ld);
            \draw[->,thick] (Ld) -- (Rd);
            \draw[->,thick] (Rd) -- (E);
            
            \draw[->,thick,edgecolor,dashed] (La) to[bend left=10] (Rc);
            \draw[->,thick,edgecolor,dashed] (Lc) to[bend left=8] (Rb);
            
            \draw[->,thick] (7.2,1.35) -- (7.9,1.35);
            \node[font=\footnotesize,right] at (8.0,1.35) {wire order};
            \draw[->,thick,edgecolor,dashed] (7.2,1.0) -- (7.9,1.0);
            \node[font=\footnotesize,right] at (8.0,1.0) {constraints};
            
            \node[font=\footnotesize,align=center] at (3.2,-0.4) {\textbf{Order:} $S, L(a), L(c), R(a)\!=\!L(b), R(c), L(d), R(d), R(b), E$};
        \end{tikzpicture}
    \end{tabular}
    \caption{\small Example state $(\mathcal{W}_x, \mathcal{O}_x)$ at a node $x$ in a tree decomposition. Top: the bag $B_x$ and the vertex set $V_x$ of the subtree rooted at $x$; gray symbols $p,q \in V_x \setminus B_x$ are forgotten, and the path $c\!-\!p\!-\!q\!-\!d$ illustrates connectivity used by $\tau=2$. Middle: the wire structure $\mathcal{W}_x$ partitions $B_x$ into ordered lists with edge types $\tau \in \{0,1,2\}$, where $\tau=0$ indicates complete disconnection. Bottom: a topological ordering $\mathcal{O}_x$ on the virtual vertices $B_x(\tau)$ with the identification $R(a)=L(b)$ induced by $\tau(a,b)=1$. Squares denote symbols (vertices of $G_f$); circles denote virtual vertices (Hadamard gates) in the constraint graph.}
    \label{fig:wx-ox}
\end{figure}

Before describing the algorithm, we define the set of \textbf{realizable} states formally. A state $(\mathcal{W}_x, \mathcal{O}_x)$ is \textbf{realizable} for a node $x$ if there exists a partition of symbols (with list endpoints $\localendpoint{s}, \localendpoint{e}$) in $V_x$ to $k$ ordered lists $L_i$ and an ordered pair type $\rho:\text{Consecutive entries in } L_i\to \{0, 1\}$:
\begin{enumerate}
    \item \textbf{Wire Structure Validity}: $W_i \subseteq  L_i$, and the order of symbols in $W_i$ is preserved in $L_i$. Moreover, the ordered pair type $\rho$ extends $\tau$ as follows:
    \begin{itemize}
        \item For consecutive entries $u, v \in W_i$ with edge type $\tau(u,v) \in \{0,1\}$, $u,v$ must be consecutive in $L_i$ and $\rho(u,v) = \tau(u,v)$.
        \item For consecutive entries $u, v \in W_i$ with edge type $\tau(u,v) = 2$, there are symbols $u_1,\dots,u_r$ with $r\geq 1$ between $u=u_0$ and $v=u_{r+1}$ in $L_i$ such that $\rho(u_j,u_{j+1})=1$ for every $j$ with both $u_j$ and $u_{j+1}$ real symbols.
        \item For consecutive entries $u, v\in L_i$, if $\rho(u,v) = 1$, then $u\neq \localendpoint{s}$ and $v\neq \localendpoint{e}$ and the quadratic term $uv$ appears in $f$.
    \end{itemize}
    
    \item \textbf{Compatible Topological Ordering}: Construct the partial constraint graph $\GH[V_x(\rho)]$ as in \cref{lem:structure} from the lists $L_i$ and the ordered pair type $\rho$, again ignoring self-loops and parallel edges. Specifically:
    \begin{itemize}
        \item \emph{Vertices}: Similar to the definition of $\mathcal{O}_x$, create virtual vertices $\{S, E\} \cup \{L(u), R(u) \mid u \in V_x\}$ with identifications $R(u) = L(v)$ for consecutive symbols $u, v$ in each list $L_i$ such that the edge type $\rho(u,v) = 1$. As a slight abuse of notation, we denote the identified vertex set as $V_x(\rho)$.
        \item \emph{Interval edges}: $L(u) \to R(u)$ for all $u \in V_x$.
        \item \emph{Wire edges}: $R(u) \to L(v)$ whenever real symbols $u,v\in V_x$ are consecutive in some list $L_i$; for the first and last real symbols of a list, add $S\to L(v)$ and $R(u)\to E$, respectively.
        \item \emph{Gate constraint edges}: For every pair of distinct symbols $\{u, v\} \subseteq V_x$ appearing together in a monomial of $f$ (i.e., in a term $uv$ or $uvw$ for some $w$), add edges $L(u) \to R(v)$ and $L(v) \to R(u)$.
    \end{itemize}
    The graph $\GH[V_x(\rho)]$ must be acyclic, and $\mathcal{O}_x$ must be the restriction to $B_x(\tau)$ of some topological ordering of $\GH[V_x(\rho)]$ (equivalently, it is consistent with $\GH[V_x(\rho)]$ in the sense of \cref{def:consistent_ordering}). The identifications are compatible because $\tau(u,v) = 1$ implies $\rho(u,v) = 1$ for consecutive symbols $u,v \in W_i \subseteq L_i$.
\end{enumerate}

Before proceeding, we recall \cref{lem:bag_subtree} (the bag subtree property), together with \cref{lem:acyclic_union} (consistent topological orderings). Briefly speaking, they state that different subtrees of the decomposition interact only through the current bag, and that acyclicity of the union of two partial constraint graphs can be certified by a common boundary ordering consistent with both partial graphs. These are the key lemmas that justify the correctness of the introduce and join transitions.

We state the algorithm to compute the set of realizable states $\mathcal{D}_x$ recursively using the nice tree decomposition in \cref{def:nice_td}. The base case and the transitions for each type of node are given below. For clarity, we prove the correctness right after presenting the algorithm.

\paragraph{Leaf Node ($B_x = \emptyset$):} 
$$\mathcal{D}_x = \{ (\mathcal{W}_x, \mathcal{O}_x) \}$$
where $\mathcal{W}_x$ contains $k$ lists $W_i = (\localendpoint{s}, \localendpoint{e})$, and $\mathcal{O}_x$ is the trivial ordering $S < E$ on virtual vertices.

\paragraph{Correctness for Leaf Node:}
The unique state $(\mathcal{W}_x, \mathcal{O}_x)$ with $k$ empty lists $(\localendpoint{s}, \localendpoint{e})$, $\tau(\localendpoint{s},\localendpoint{e}) = 0$, ordering $S < E$ is realizable. Indeed, we can set $L_i = (\localendpoint{s}, \localendpoint{e})$ for all $i \in [k]$ with no symbols from $V_x = \emptyset$ to place, and $\rho(\localendpoint{s},\localendpoint{e}) = 0$. The constraint graph $\GH[V_x(\rho)]$ has only vertices $S, E$ and the trivial ordering $S < E$ is a valid topological order. This is the only possible state since $B_x = \emptyset$.

\paragraph{Introduce Node ($B_x = B_y \cup \{u\}$):}
For each state $(\mathcal{W}_y, \mathcal{O}_y) \in \mathcal{D}_y$, we generate new states by inserting $u$:
\begin{itemize}
    \item \textbf{Choose Wire and Position}: For each wire $i \in [k]$, let $W_i^y = (\localendpoint{s} = v_0, v_1, \dots, v_{m_i}, v_{m_i+1} = \localendpoint{e})$ be the ordered list in $\mathcal{W}_y$. For each position $j \in \{0, \dots, m_i\}$, we try to insert $u$ between $v_j$ and $v_{j+1}$.
    
    \item \textbf{Check Feasibility}: If $\tau_y(v_j, v_{j+1}) \in \{1, 2\}$ (already connected or via forgotten symbols), skip.
    
    \item \textbf{Set Edge Types}: For the insertion, enumerate all valid edge type assignments $\tau_x(v_j, u) \in \{0, 1\}$ and $\tau_x(u, v_{j+1}) \in \{0, 1\}$ such that: 
    \begin{itemize}
        \item If $\tau_x(v_j, u) = 1$, then $v_j \neq  \localendpoint{s}$ and the quadratic monomial $v_j u$ appears in $f$.
        \item If $\tau_x(u, v_{j+1}) = 1$, then $v_{j+1} \neq \localendpoint{e}$ and the quadratic monomial $u v_{j+1}$ appears in $f$.
    \end{itemize}

    Let $\mathcal{W}_x$ be the updated wire structure with $u$ inserted.
    
    \item \textbf{Update Topological Ordering}: Construct virtual vertices for $u$: $L(u)$ and $R(u)$. Update identifications based on the new edge types:
    \begin{itemize}
        \item If $\tau_x(v_j, u) = 1$, identify $R(v_j) = L(u)$ (note $v_j\neq \localendpoint{s}$)
        \item If $\tau_x(u, v_{j+1}) = 1$, identify $R(u) = L(v_{j+1})$ (note $v_{j+1} \neq \localendpoint{e}$)
    \end{itemize}
    This gives the identified virtual vertex set $B_x(\tau_x)$. Insert $L(u)$ and $R(u)$ into updated $\mathcal{O}_y$ to form the topological ordering $\mathcal{O}_x$ on $B_x(\tau_x)$, ensuring the following directed edges are consistent with $\mathcal{O}_x$ after insertion:
    \begin{itemize}
        \item Interval and wire order: the interval edge $L(u)\to R(u)$ must be consistent with $\mathcal{O}_x$; moreover, the immediate predecessor of $L(u)$ is $S$ if $v_j=\localendpoint{s}$ and $R(v_j)$ otherwise, while the immediate successor of $R(u)$ is $E$ if $v_{j+1}=\localendpoint{e}$ and $L(v_{j+1})$ otherwise; these order constraints must also be consistent with $\mathcal{O}_x$.
        \item Constraint edges: For each monomial in $f$ containing $u$ and another symbol $w \in B_y$, add edges $L(u) \to R(w)$ and $L(w) \to R(u)$ (as defined in realizability condition). 
    \end{itemize}
    
    \item If $(\mathcal{W}_x, \mathcal{O}_x)$ pass all checks, add it to $\mathcal{D}_x$.
\end{itemize}

\paragraph{Correctness for Introduce Node:} We show that a state $(\mathcal{W}_x, \mathcal{O}_x)$ is realizable for $x$ if and only if it is generated by the introduce procedure from some realizable $(\mathcal{W}_y, \mathcal{O}_y) \in \mathcal{D}_y$.

    \noindent ($\Rightarrow$) Suppose $(\mathcal{W}_x, \mathcal{O}_x)$ is realizable. Then there exist lists $(L_i^x)$ partitioning $V_x$ and an ordered pair type $\rho_x$ satisfying the realizability conditions. Let $u$ appear in list $W_i^x$ between symbols $u_j$ and $u_{j+1}$ in $B_y \cup \{\localendpoint{s}, \localendpoint{e}\}$. By \cref{lem:bag_subtree} applied to the child subtree, there is no edge of $G_f$ between $u$ and $V_y\setminus B_y=V_x\setminus B_x$. Therefore, $\tau_x(u_j, u)$ and $\tau_x(u, u_{j+1})$ (which are extended by $\rho_x$ from realizability) cannot be type $2$; they can only be $0$ or $1$, and hence $u_j, u, u_{j+1}$ must be consecutive in the list $L_i^x$.
    
    Then, we can define the state $(\mathcal{W}_y, \mathcal{O}_y)$ for $y$ by removing $u$ from $W_i^x$ and updating $\tau_y(u_j, u_{j+1}) = 0$. For the topological order, we note that in $\mathcal{O}_x$ on $B_x(\tau_x)$, the two boundary sides of the gap between $u_j$ and $u_{j+1}$ are not identified since $u$ is placed between them in $L_i^x$. Hence we can simply remove the labels $L(u)$ and $R(u)$ from $B_x(\tau_x)$, keeping any boundary labels identified with them, to get $B_y(\tau_y)$, giving $\mathcal{O}_y$. It is straightforward to see that $L_i^x$ and $\rho_x$ with $u$ removed are valid extensions for $(\mathcal{W}_y, \mathcal{O}_y)$; hence, $(\mathcal{W}_y, \mathcal{O}_y)$ is realizable. Thus, the introduce procedure would recover $(\mathcal{W}_x, \mathcal{O}_x)$ from $(\mathcal{W}_y, \mathcal{O}_y)$, since $(\mathcal{W}_x, \mathcal{O}_x)$ would pass all checks because it is realizable.
    
    \noindent ($\Leftarrow$) Suppose $(\mathcal{W}_x, \mathcal{O}_x)$ is generated from realizable $(\mathcal{W}_y, \mathcal{O}_y)$ by the introduce procedure. Since $(\mathcal{W}_y, \mathcal{O}_y)$ is realizable, there exist lists $(L_i^y)$ for $V_y$ and $\rho_y$ satisfying the conditions. The algorithm inserts $u$ at position $j$ in wire $i$, between symbols $v_j$ and $v_{j+1}$ in $W_i^y$. From the feasibility check, $\tau_y(v_j,v_{j+1}) = 0$; hence, $v_j$ and $v_{j+1}$ are consecutive in $L_i^y$, so $u$ can be inserted between them.
    
    We construct lists $(L_i^x)$ for $V_x$ by inserting $u$ at the same position in $L_i^y$ (between the same symbols $v_j$ and $v_{j+1}$). Define $\rho_x$ by adding $\rho_x(v_j, u) = \tau_x(v_j, u)$ and $\rho_x(u, v_{j+1}) = \tau_x(u, v_{j+1})$. All other pairs remain the same as in $\rho_y$. Clearly $(L_i^x)$ and $\rho_x$ satisfy the wire structure validity. 

    For compatible topological ordering, again using \cref{lem:bag_subtree} applied to the child subtree, all constraint edges involving the introduced symbol $u$ and the already processed symbols of $V_y$ have their other endpoint in the boundary $B_y$. Therefore, the directed graph $\GH[V_x(\rho_x)]$ is the union of $\GH[V_y(\rho_y)]$ and the part of $\GH[V_x(\rho_x)]$ supported on the boundary vertices $B_x(\tau_x)$. $\GH[V_y(\rho_y)]$ is acyclic since $(\mathcal{W}_y, \mathcal{O}_y)$ is realizable. The boundary part is also acyclic since the algorithm checks that the new wire-order and constraint edges involving $u$ are consistent with $\mathcal{O}_x$, while the old boundary reachability relations are already consistent with $\mathcal{O}_y$ and $\mathcal{O}_x$ restricts to $\mathcal{O}_y$. Hence the common boundary order is consistent with both pieces. By \cref{lem:acyclic_union}, $\GH[V_x(\rho_x)]$ is acyclic, and $\mathcal{O}_x$ is consistent with it. Therefore, $(\mathcal{W}_x, \mathcal{O}_x)$ is realizable.
    \paragraph{Forget Node ($B_x = B_y \setminus \{u\}$):}
For each state $(\mathcal{W}_y, \mathcal{O}_y) \in \mathcal{D}_y$, we remove $u$:
\begin{itemize}
    \item \textbf{Locate Symbol}: Find $u$ in some list $W_i^y = (\dots, v_j, u, v_{j+1}, \dots)$ in $\mathcal{W}_y$.
    
    \item \textbf{Check Connection}: Verify that $\tau_y(v_j, u) \neq 0$ if $v_j\neq \localendpoint{s}$, and that $\tau_y(u, v_{j+1}) \neq 0$ if $v_{j+1}\neq \localendpoint{e}$; otherwise skip (cannot forget disconnected symbols).
    
    \item \textbf{Update Wire Structure}: Remove $u$ from the list to get $W_i^x = (\dots, v_j, v_{j+1}, \dots)$. Set $\tau_x(v_j, v_{j+1}) = 2$. This marks that $v_j$ and $v_{j+1}$ are now connected through forgotten symbols.
    
    \item \textbf{Update Ordering}: Delete the virtual vertices $L(u)$ and $R(u)$ from $B_y(\tau_y)$ to obtain $B_x(\tau_x)$, keeping all other labels in the same order. If $L(u)$ or $R(u)$ is identified with a boundary endpoint such as $R(v_j)$ or $L(v_{j+1})$, the boundary endpoint label is kept; only the label belonging to the forgotten symbol $u$ is removed.
    
    \item Add $(\mathcal{W}_x, \mathcal{O}_x)$ to $\mathcal{D}_x$.
\end{itemize}
\paragraph{Correctness of Forget Node:}
    We show $(\mathcal{W}_x, \mathcal{O}_x)$ is realizable for $x$ if and only if there exists realizable $(\mathcal{W}_y, \mathcal{O}_y) \in \mathcal{D}_y$ from which it is obtained by the forget procedure.
    
    \noindent ($\Rightarrow$) Suppose $(\mathcal{W}_x, \mathcal{O}_x)$ is realizable. Then there exist lists $(L_i^x)$ for $V_x$ and $\rho$ satisfying the conditions. Since $u \in B_y \subseteq V_x$, symbol $u$ must be placed in some list $L_i^x$ between symbols $v_j, v_{j+1} \in W_i^x$ ($v_j$ and $v_{j+1}$ are consecutive in $W_i^x$ but not in $L_i^x$). Then, from wire structure validity for $(\mathcal{W}_x, \mathcal{O}_x) $, the only possibility for $\tau_x(v_j, v_{j+1})$ is $2$. Then, there exists a path from $v_j$ to $v_{j+1}$ in $L_i^x$ using only vertices in $V_x \setminus B_x$ (i.e., forgotten symbols), and consecutive real-symbol pairs on the path have $\rho_x$ value $1$, meaning their product appears in $f$.

    Hence, now we can set $\tau_y(v_j, u)$ and $\tau_y(u, v_{j+1})$ to be $0$, $1$, or $2$ appropriately: an endpoint ($\localendpoint{s}$ or $\localendpoint{e}$) side may be type $0$ or $2$, depending on whether there are other forgotten symbols on that side, while a real-symbol side is type $1$ or $2$ depending on whether it is directly adjacent to $u$. The ordering $\mathcal{O}_y$ on $B_y(\tau_y)$ is obtained by inserting $L(u), R(u)$ into $\mathcal{O}_x$ on $B_x(\tau_x)$ consistently according to a topological ordering of the same acyclic partial graph $\GH[V_x(\rho_x)] = \GH[V_y(\rho_y)]$.
    
    This gives a realizable $(\mathcal{W}_y, \mathcal{O}_y)$ that produces $(\mathcal{W}_x, \mathcal{O}_x)$ via forgetting.
    
    \noindent ($\Leftarrow$) Suppose $(\mathcal{W}_x, \mathcal{O}_x)$ is obtained from realizable $(\mathcal{W}_y, \mathcal{O}_y)$ by forgetting $u$. Since $(\mathcal{W}_y, \mathcal{O}_y)$ is realizable with lists $(L_i^y)$ and $\rho_y$ for $V_y = V_x$, hence they are already valid for $V_x$.
    
    The wire structure validity holds because the segment between $v_j$ and $v_{j+1}$ in $L_i^y$ contains $u$ and possibly other forgotten symbols, with all consecutive real-symbol pairs connected with $\rho_y = 1$. After removing $u$, setting $\tau_x(v_j, v_{j+1}) = 2$ correctly reflects that they are connected via forgotten symbols.
    
    The compatible topological ordering holds because $\mathcal{O}_x$ is the restriction of $\mathcal{O}_y$ after deleting only the labels belonging to $u$. Since $\mathcal{O}_y$ is the restriction of a topological ordering of the same partial graph $\GH[V_y(\rho_y)]=\GH[V_x(\rho_x)]$, the restricted order $\mathcal{O}_x$ is consistent with all boundary reachability relations in $B_x(\tau_x)$. Thus $(\mathcal{W}_x, \mathcal{O}_x)$ is realizable.
\paragraph{Join Node ($B_x = B_y = B_z$):}
For each pair $(\mathcal{W}_y, \mathcal{O}_y) \in \mathcal{D}_y$ and $(\mathcal{W}_z, \mathcal{O}_z) \in \mathcal{D}_z$, merge if compatible:
\begin{itemize}
    \item \textbf{Check Wire Compatibility}: For each wire $i$, the ordered lists $W_i^y$ and $W_i^z$ must have the same sequence of symbols from $B_x \cup \{ \localendpoint{s}, \localendpoint{e} \}$.
    
    \item \textbf{Merge Edge Types}: For each pair of consecutive entries $(v, w)$ in $B_x\cup\{\localendpoint{s},\localendpoint{e}\}$ on the same wire:
    \begin{itemize}
        \item If $\tau_y(v, w) = \tau_z(v, w) = 1$, set $\tau_x(v, w) = 1$. 
        \item If exactly one of $\tau_y(v, w)$ and $\tau_z(v, w)$ is $1$, skip.
        \item If $\tau_y(v, w) = \tau_z(v, w) = 2$, skip (cannot tolerate different forgotten paths).
        \item If $\tau_y(v,w)=\tau_z(v,w)=0$, set $\tau_x(v,w)=0$.
        \item Otherwise ($\{ \tau_y(v, w) , \tau_z(v, w) \} = \{ 0,2 \}$), set $\tau_x(v, w) = 2$.
    \end{itemize}
    
    \item \textbf{Check Ordering Compatibility}: The orderings $\mathcal{O}_y$ and $\mathcal{O}_z$ are orderings on the identified virtual vertex sets $B_y(\tau_y)$ and $B_z(\tau_z)$ respectively. Since the wire and edge-type compatibility checks ensure $B_y(\tau_y) = B_z(\tau_z) = B_x(\tau_x)$, check that the orderings agree: $\mathcal{O}_y = \mathcal{O}_z$ on $B_x(\tau_x)$. 
    
    \item If all checks pass, let $\mathcal{W}_x$ be the merged wire structure with edge types $\tau_x$, set $\mathcal{O}_x=\mathcal{O}_y=\mathcal{O}_z$ on the common boundary, and add $(\mathcal{W}_x, \mathcal{O}_x)$ to $\mathcal{D}_x$.
\end{itemize}
 \paragraph{Correctness of Join Node:}
    We show $(\mathcal{W}_x, \mathcal{O}_x)$ is realizable for $x$ if and only if there exist compatible realizable states $(\mathcal{W}_y, \mathcal{O}_y) \in \mathcal{D}_y$ and $(\mathcal{W}_z, \mathcal{O}_z) \in \mathcal{D}_z$ that merge to form it.
    
    \noindent ($\Rightarrow$) Suppose $(\mathcal{W}_x, \mathcal{O}_x)$ is realizable with lists $(L_i^x)$ for $V_x = V_y \cup V_z$ and $\rho$. By the separation property of tree decompositions \cref{lem:bag_subtree}, there is no edge of $G_f$ between $V_y\setminus B_x$ and $V_z\setminus B_x$; hence the two child subinstances interact only through the boundary $B_x$. 
    
    Now, we directly set the wire structure and ordering for $y$ and $z$ to be the same as $x$: $\mathcal{W}_y = \mathcal{W}_x = \mathcal{W}_z $. For edge types, if $\tau_x(v, w) = 2$ for consecutive entries in $B_x\cup\{\localendpoint{s},\localendpoint{e}\}$, the forgotten path connecting them lies entirely in either $V_y \setminus B_x$ or $V_z \setminus B_x$ by separation property \cref{lem:bag_subtree}, so we set $\tau_y(v,w)$, $\tau_z(v,w)$ to be $2,0$ accordingly. If $\tau_x(v, w) = 1$ or $0$, we set $\tau_y(v,w) = \tau_z(v,w) = \tau_x(v,w)$. This gives compatible wire structures that merge to $\mathcal{W}_x$. 

    For topological ordering, we set $\mathcal{O}_y = \mathcal{O}_x = \mathcal{O}_z$ on $B_x(\tau_x) = B_y(\tau_y) = B_z(\tau_z)$. It is easy to see that $(\mathcal{W}_y, \mathcal{O}_y)$ and $(\mathcal{W}_z, \mathcal{O}_z)$ are realizable by restricting the lists $(L_i^x)$ and $\rho$ to $V_y$ and $V_z$ respectively.
    
    Both restrictions preserve the realizability.  The orderings $\mathcal{O}_y$ and $\mathcal{O}_z$ agree on boundary vertices since they both come from $\mathcal{O}_x$. These states are compatible and merge to $(\mathcal{W}_x, \mathcal{O}_x)$.
    
    \noindent ($\Leftarrow$) Suppose compatible realizable $(\mathcal{W}_y, \mathcal{O}_y)$ and $(\mathcal{W}_z, \mathcal{O}_z)$ merge to form $(\mathcal{W}_x, \mathcal{O}_x)$. Since both are realizable, there exist lists $(L_i^y)$ for $V_y$ with $\rho_y$ and $(L_i^z)$ for $V_z$ with $\rho_z$. 
    
    Since the wire structures agree on $B_x$, we can merge the lists: for each wire $i$, take symbols from both $L_i^y$ and $L_i^z$, keeping symbols in $B_x\cup \{ \localendpoint{s}, \localendpoint{e} \}$ at their designated positions. The merging process is done by considering each consecutive pair $(v,w)$ in $B_x\cup\{\localendpoint{s},\localendpoint{e}\}$ on the same wire:
    \begin{itemize}
        \item If $\tau_x(v,w) \in \{0,1\}$, then from merging mechanism, $\tau_y(v,w)= \tau_z(v,w) = \tau_x(v,w)$. Hence, they are consecutive in both $L_i^y$ and $L_i^z$. We set them to be consecutive in the merged list and $\rho_x(v,w) = \tau_x(v,w)$.
        \item If $\tau_x(v,w) = 2$, then from merging mechanism, one of $\tau_y(v,w), \tau_z(v,w) $ is $2$, connected by a forgotten path, and the other is $0$. We take the path from the one with $\tau = 2$ to connect $v$ and $w$ in the merged list. Consecutive real-symbol pairs on the path have $\rho = 1$ in the corresponding list.
    \end{itemize}
    This gives valid lists $(L_i^x)$ for $V_x$ and $\rho_x$ satisfying wire structure validity.
    
    For compatible topological ordering, by \cref{lem:bag_subtree}, the constraint graphs $\GH[V_y(\rho_y)]$ and $\GH[V_z(\rho_z)]$ only interact through their common boundary virtual vertices $B_x(\tau_x) = B_y(\tau_y) = B_z(\tau_z)$. Since $(\mathcal{W}_y, \mathcal{O}_y)$ and $(\mathcal{W}_z, \mathcal{O}_z)$ are realizable, both $\GH[V_y(\rho_y)]$ and $\GH[V_z(\rho_z)]$ are acyclic with topological orderings extending $\mathcal{O}_y$ and $\mathcal{O}_z$ respectively on $B_y(\tau_y)$ and $B_z(\tau_z)$. Since $\mathcal{O}_y$ and $\mathcal{O}_z$ agree on $B_x(\tau_x)$, by \cref{lem:acyclic_union}, the union graph $\GH[V_x(\rho_x)] = \GH[V_y(\rho_y)] \cup \GH[V_z(\rho_z)]$ is acyclic. Moreover, $\mathcal{O}_x = \mathcal{O}_y = \mathcal{O}_z$ on $B_x(\tau_x)$ extends to a topological ordering of $\GH[V_x(\rho_x)]$, proving that $(\mathcal{W}_x, \mathcal{O}_x)$ is realizable.

Now we are ready to state and prove the main algorithm theorem.
\begin{theorem}[Main Algorithm]\label{thm:main_algorithm}
    Given a degree-3 polynomial $f \in \mathbb{F}_2[\mathcal{S}]$ with no constant term on $n$ symbols with $m$ monomials and an integer $k$, we can solve the search version of \textnormal{\problem{(f,k)}} in time $k^{6k + o(k)} \cdot n + O(m)$.
\end{theorem}

\begin{proof}
    The full algorithm proceeds in the following steps:
    \begin{enumerate}
        \item \textbf{Clean up polynomial}: Check $k\leq n$, delete all symbols not appearing in $f$, and remove all quadratic subterms $uv$ that appear in a cubic term $uvw$. This takes time $O(m)$.
        \item \textbf{Construct graph $G_f$}: From the input polynomial $f$, construct the undirected connection graph $G_f$ in adjacency list form, where vertices are symbols in $\mathcal{S}$ and an edge $(u,v)$ included if and only if the product $uv$ appears in a monomial of $f$ (i.e., as $uv$ or $uvw$ for some $w$). This takes time $O(m)$.
        
        \item \textbf{Compute and Convert to  nice tree decomposition}: Use the algorithm from \cref{thm:treewidth} to compute a tree decomposition of $G_f$ with width $w \le 2k + 1 = O(k)$ in time $2^{O(k)} n$. If the algorithm determines that the treewidth is larger than $k$, directly return ``No solution'' using \cref{lem:pathdecomposition}. Then, using \cref{lem:nice_td}, convert the tree decomposition to a nice tree decomposition with width $w$ and $O(wn) = O(kn)$ nodes in time $O(w^2 n) = O(k^2 n)$.
            
        \item \textbf{Dynamic programming}: Run the DP algorithm on the nice tree decomposition. At each node $x$:
        \begin{itemize}
            \item Let $b_x=|B_x|\le w+1$. The number of possible ordered partitions of the $b_x$ boundary symbols into $k$ ordered wire lists is at most $b_x!\binom{b_x+k}{k}=k^{2k+o(k)}$, using $b_x\le 2k+2$.
            \item The number of possible edge types $\tau$ between consecutive symbols is at most $3^{b_x+k}=2^{O(k)}$.
            \item The number of possible topological orderings $\mathcal{O}_x$ on at most $2b_x+2\le 4k+6$ virtual vertices (with global minimum $S$ and maximum $E$) is at most $(2b_x+2)! = k^{4k + o(k)}$.
            \item Therefore, $|\mathcal{D}_x| \le k^{6k + o(k)}$.
        \end{itemize}
        For each node type:
        \begin{itemize}
            \item \textbf{Leaf node}: Initialize in constant time.
            \item \textbf{Introduce node}: For each state in $\mathcal{D}_y$, try $O(k)$ positions to insert the new symbol, enumerate $O(1)$ edge type combinations, and check topological ordering consistency in time $O(k^2)$. Total: $|\mathcal{D}_y| \cdot O(k^3)$ time per node.
            \item \textbf{Forget node}: For each state in $\mathcal{D}_y$, update in time $O(k)$. Total: $|\mathcal{D}_y| \cdot O(k)$ time per node.
            \item \textbf{Join node}: Index states of each child by their ordered boundary lists and topological ordering. For each state from $\mathcal{D}_y$, the compatible choices in $\mathcal{D}_z$ differ only in the edge-type vector, so they can be checked in $2^{O(k)}\cdot k^2$ time. Total: $|\mathcal{D}_y|\cdot 2^{O(k)}\cdot k^2$ time per node.
        \end{itemize}
        Since $|\mathcal{D}_x| = k^{6k + o(k)}$ for all nodes $x$, and there are $O(kn)$ nodes, the total DP time is $k^{6k+o(k)} \cdot 2^{O(k)}\cdot O(k^2\cdot kn) =k^{6k+o(k)} \cdot n$.
        
        \item \textbf{Verify and construct circuit}: At the root node $r$, the bag $B_r$ is empty (by the nice tree decomposition property). Check if $\mathcal{D}_r$ is non-empty. If $\mathcal{D}_r = \emptyset$, return ``No solution''. 
        
        If $\mathcal{D}_r \neq \emptyset$, then there exists a realizable state $(\mathcal{W}_r, \mathcal{O}_r)$ for the entire $f$. Since $B_r = \emptyset$, the boundary wire structure $\mathcal{W}_r$ has all $k$ boundary lists in the form $(\localendpoint{s}, \localendpoint{e})$; the type on such an endpoint pair may be $2$ if that wire contains forgotten symbols. The topological ordering $\mathcal{O}_r$ only orders the virtual vertices $\{S, E\}$. The realizability of $(\mathcal{W}_r, \mathcal{O}_r)$ guarantees that there exists a topological ordering on the constraint graph $\GH[V_r(\rho_r)] = \GH[\mathcal{S}]$, which by \cref{lem:structure} implies the existence of a quantum circuit on $k$ qubits with associated polynomial $f$.
        
        To extract the circuit, we trace back through the DP table from the root to the leaves:
        \begin{itemize}
            \item At the root $r$, pick any state $(\mathcal{W}_r, \mathcal{O}_r) \in \mathcal{D}_r$.
            \item For each node $x$ with children $y$ (and possibly $z$ for join nodes), we have recorded which parent states in $\mathcal{D}_x$ were obtained from which child states in $\mathcal{D}_y$ (and $\mathcal{D}_z$). Follow these pointers backwards to obtain states at all nodes.
            \item At introduce nodes where symbol $u$ is added, the chosen state specifies the position of $u$ on a specific wire and the Hadamard gate configuration (via edge types $\tau$).
            \item The traceback reconstructs the full ordered wire lists. We then run the check/construction procedure from \cref{lem:structure}: place symbols on wires according to the reconstructed lists, use a topological ordering of the final constraint graph to place Hadamard gates, and apply Z, CZ, and CCZ gates according to the monomials in $f$.
        \end{itemize}
        This extraction takes $O(kn)$ time to trace back, plus $O(m)$ time to construct the circuit gates and adding back removed quadratic terms and symbols. The total is $O(m + kn)$.
    \end{enumerate} 
    
    The total time complexity is dominated by the DP step, giving the desired $k^{6k + o(k)} \cdot n + O(m)$ time bound.
\end{proof}

\section{Conclusion and Future directions}\label{sec:conclusion}

In this paper, we have conducted a comprehensive complexity-theoretic study of
the quantum circuit width problem.
Our results establish both strong computational hardness barriers and efficient
algorithmic solutions under an appropriate and natural parameterization.

On the hardness side, we proved that determining whether the width $w(f)$ of a
degree-3 polynomial $f$ is at most $k$ is \NP-complete, settling Montanaro's
open question in the negative~\cite{mont}.
Furthermore, we showed that even approximating the width within a constant
factor of $\tfrac{49}{48}$ is \NP-hard, ruling out efficient approximation
algorithms.
Notably, this approximation hardness persists even when restricted to degree-2
polynomials, establishing hardness for a broad class of quantum gate sets
including $\{H,\mathrm{CS}\}$ and Clifford$+T$.
Under the Exponential Time Hypothesis, we obtained a conditional lower bound of
$2^{\Omega(n)}$ for any algorithm solving \problem{}.

On the algorithmic side, we developed a nondeterministic polynomial-time algorithm with witness complexity $O(\log_2\binom{n}{k}) = O(k\log(en/k))$ bits, yielding an XP algorithm with runtime $\binom{n}{k}^2\operatorname{poly}(n)$ via enumeration. When $k = \Theta(n)$, we obtain a deterministic algorithm by enumerating witnesses that achieves runtime $2^{O(n)}$, matching the ETH lower bound up to constant factors in the exponent. More significantly, we proved that \problem{} is fixed-parameter tractable with respect to parameter $k$, presenting an $\mathsf{FPT}$ algorithm with runtime $k^{6k + o(k)} \cdot n$. This algorithm exploits the bounded pathwidth structure of the connection graph and employs dynamic programming over tree decompositions.

Together, these results situate the width problem firmly in classical complexity theory, demonstrating that width optimization is computationally intractable in the worst case, yet tractable when the target width is sufficiently small. This suggests that width behaves more like a structural parameter (akin to treewidth) rather than a measure of intrinsic quantum computational power.

\paragraph{Future directions.}
Several interesting questions remain open:

\begin{enumerate}
    \item \textbf{Approximation hardness for larger constants and restricted degree-2 instances.} Our inapproximability result establishes a hardness factor of $\tfrac{49}{48} \approx 1.02$. Can this be improved to a larger constant? A fundamental limitation of our approach is that the reduction encodes Boolean satisfiability using \emph{geometric constraints} on segment overlaps. This geometric encoding strategy is inherently limited to Boolean variables and cannot naturally extend to constraint satisfaction problems (CSPs) over larger alphabets. Consequently, standard PCP-based hardness-of-approximation techniques for CSPs~\cite{Hastard,Dinur2007} do not directly apply to our setting. Obtaining inapproximability with arbitrarily large constant factors may require fundamentally different encoding mechanisms or new connections between quantum circuit geometry and combinatorial optimization. Closely related to this is the question of better characterizing the family of degree-2 polynomials used in our hardness reduction. Our degree-2 result is obtained through a specific twin-construction that simulates cubic overlap constraints using quadratic interactions. It would be very interesting to understand whether the same hardness persists for cleaner or more natural subclasses of quadratic polynomials, or equivalently, whether one can impose meaningful restrictions on the associated family of quantum circuits while preserving hardness.
    
    \item \textbf{Optimality of the FPT runtime.} Our $\mathsf{FPT}$ algorithm achieves runtime $k^{6k + o(k)} \cdot n$. Can this be improved to single-exponential time $2^{O(k)} \cdot n$, or is there a matching lower bound? This question is particularly intriguing in light of a similar work on the \textsc{Directed Feedback Vertex Set} problem parameterized by treewidth~\cite{LokshtanovRamanujanSaurabh2017}. That work presents an algorithm with similar structure to ours---maintaining compatible topological orderings on tree decomposition bags---yet proves a lower bound showing that no algorithm with runtime $t^{o(t)}$ exists (where $t$ is the treewidth) under the Exponential Time Hypothesis, via a reduction from $k$-CNF satisfiability with $k$-permutation constraints. It remains unclear whether their reduction technique can be adapted to \problem{}, or whether our problem admits a faster single-exponential algorithm. Resolving this question would provide valuable insights into the fine-grained complexity landscape of parameterized quantum circuit problems.

    \item \textbf{Tractable special cases and the boundary of hardness.} While our results suggest that minimum circuit width is intractable in the worst case (unless $\mathsf{NP} \subseteq \mathsf{BQP}$), they do not rule out broad structured subclasses of polynomials for which the problem may remain tractable. Montanaro already identified several special cases in which the minimum circuit width can be found efficiently, even classically~\cite{mont}. It would therefore be valuable to better understand the boundary between these easy cases and the hard instances constructed here. In particular, one may ask whether the techniques developed in this paper point to additional tractable families---for example, classes of polynomials whose overlap structure is sufficiently rigid that a minimum-width realization can be recovered directly, or classes of associated circuits constrained to particularly simple architectures. Any such characterization would complement our worst-case hardness results and could be especially relevant for practical compilation problems.
\end{enumerate}

\paragraph{Acknowledgements}
We thank the anonymous reviewers for their helpful comments and suggestions, especially on the importance of better characterizing the degree-2 instances used in our hardness reduction and of clarifying the role of potentially tractable special cases.
ZJ acknowledges the support by National Natural Science Foundation of China (Grant
No.\ 12347104), National Key Research and Development Program of China
(Grant No.\ 2023YFA1009403), Beijing Science and Technology Planning Project (Grant No.\
Z25110100810000), and Beijing Natural Science Foundation (Grant No.\ Z220002).

\newpage
\newpage
\crefalias{section}{appendix}
\begin{appendices}

\section{Missing Proofs in \cref{subsec:alg_prelim}}\label{app:lemmas}

\subsection{Boundary Property of Tree Decomposition}\label{app:bag_subtree}

\begin{lemma}\label{lem:bag_subtree_proof}
Let $(T, \{ B_t \}_{t\in V(T)})$ be a rooted tree decomposition of a graph $G$ and let $x$ be a node of $T$ with children $y_1, y_2, \dots, y_m$. For each $i$, let
$
    V_{y_i}:=\bigcup_{z\in V(T_{y_i})} B_z,
$
where $T_{y_i}$ is the subtree of $T$ rooted at $y_i$. Then there is no edge of $G$ between $V_{y_i}\setminus B_x$ and $V(G)\setminus V_{y_i}$. Consequently, if $i\neq j$, then there is no edge of $G$ between $V_{y_i}\setminus B_x$ and $V_{y_j}\setminus B_x$.
\end{lemma}

\begin{proof}
Fix a child $y_i$ of $x$ and write $Y=V_{y_i}$. Suppose, for contradiction, that there is an edge $uv$ with
$
    u\in Y\setminus B_x
    \quad\text{and}\quad
    v\notin Y.
$
By the edge-cover property of tree decompositions, some bag $B_z$ contains both $u$ and $v$. Since $v\notin Y$, the node $z$ is not in the subtree $T_{y_i}$. On the other hand, since $u\in Y$, there is a node $z'\in V(T_{y_i})$ such that $u\in B_{z'}$. The path in $T$ from $z$ to $z'$ must pass through $x$, because $T_{y_i}$ is connected to the rest of the rooted tree through the edge $xy_i$.

The bags containing $u$ form a connected subtree of $T$ by the running-intersection property. This connected subtree contains both $z$ and $z'$, since $u\in B_z\cap B_{z'}$. Therefore it contains every node on the path from $z$ to $z'$, in particular $x$. Hence $u\in B_x$, contradicting $u\in Y\setminus B_x$. Thus no such edge exists.

It remains to prove the consequence for two distinct children. Let $i\neq j$. First note that
\[
    (V_{y_i}\setminus B_x)\cap (V_{y_j}\setminus B_x)=\emptyset.
\]
Indeed, if a vertex $w$ belonged to both $V_{y_i}$ and $V_{y_j}$, then there would be bags containing $w$ in both child subtrees. The path between these two bags passes through $x$, so the running-intersection property would imply $w\in B_x$. Hence no vertex outside $B_x$ can belong to both child subtrees.

Therefore
\[
    V_{y_j}\setminus B_x \subseteq V(G)\setminus V_{y_i}.
\]
Applying the first part with $Y=V_{y_i}$ shows that there is no edge between $V_{y_i}\setminus B_x$ and $V_{y_j}\setminus B_x$.
\end{proof}

\subsection{Acyclicity of Graph Union with Consistent Orderings}\label{app:acyclic_union}

\begin{lemma}\label{lem:acyclic_union_proof}
Let $G_1$ and $G_2$ be directed acyclic graphs with vertex sets $V_1$ and $V_2$, and let $V_0=V_1\cap V_2$. If there is a total order $\mathcal O_0$ on $V_0$ that is consistent with both $G_1$ and $G_2$ in the sense of \cref{def:consistent_ordering}, then $G_1\cup G_2$ is acyclic and admits a topological ordering extending $\mathcal O_0$. Conversely, if $G_1\cup G_2$ is acyclic, then the restriction of any topological ordering of $G_1\cup G_2$ to $V_0$ is consistent with both $G_1$ and $G_2$.
\end{lemma}
\begin{proof}
First assume that the same total order $\mathcal O_0$ on $V_0$ is consistent with both $G_1$ and $G_2$. Suppose for contradiction that $G=G_1\cup G_2$ contains a directed cycle. Since $G_1$ and $G_2$ are individually acyclic, this cycle must use edges from both graphs. Decompose the cycle into maximal nonempty directed subpaths that lie entirely in $G_1$ or entirely in $G_2$. The endpoints of every such maximal subpath lie in the common boundary $V_0$, because a switch from one graph to the other can occur only at a vertex that belongs to both vertex sets.

If such a subpath goes from $a\in V_0$ to $b\in V_0$, then $a\neq b$ since the corresponding graph is acyclic. By consistency of $\mathcal O_0$ with that graph, we have
\[
    a<_{\mathcal O_0} b.
\]
Traversing all maximal subpaths around the directed cycle therefore gives a strict chain
\[
    a_1<_{\mathcal O_0}a_2<_{\mathcal O_0}\cdots<_{\mathcal O_0}a_r<_{\mathcal O_0}a_1,
\]
which is impossible. Hence $G$ is acyclic.

It remains to show that $G$ has a topological ordering extending $\mathcal O_0$. Write the boundary order as
\[
    v_1<_{\mathcal O_0}v_2<_{\mathcal O_0}\cdots<_{\mathcal O_0}v_s.
\]
Add the directed edges $v_i\to v_{i+1}$ for $i=1,\dots,s-1$, and call the resulting graph $G^+$. We claim that $G^+$ is acyclic. Indeed, if $G^+$ contained a directed cycle, then since $G$ is already acyclic, the cycle would use at least one added boundary edge. Between consecutive added boundary edges, the cycle follows directed paths in $G$ with endpoints in $V_0$. Decomposing these paths into maximal subpaths inside $G_1$ or $G_2$ and using the same consistency argument as above, every such path is increasing with respect to $\mathcal O_0$. The added boundary edges are also increasing with respect to $\mathcal O_0$. Thus the cycle would again give a strict increasing cycle in the total order $\mathcal O_0$, impossible. Therefore $G^+$ is acyclic.

Any topological ordering of $G^+$ is a topological ordering of $G$ and extends $\mathcal O_0$, since the added edges force
\[
    v_1<v_2<\cdots<v_s.
\]

Conversely, suppose $G_1\cup G_2$ is acyclic. Then any topological ordering of the union restricts to an order on $V_0$. Every directed path in either $G_1$ or $G_2$ is also a directed path in the union, so the restricted order is consistent with both graphs.
\end{proof}

\end{appendices}

\bibliography{main}

@inproceedings{WCYJ25,
    address   = {Zagreb, Croatia},
    title     = {{FeynmanDD: Quantum Circuit Analysis with Classical Decision Diagrams}},
    booktitle = {Proceedings of the 37th International Conference on Computer Aided Verification},
    author    = {Wang, Ziyuan and Cheng, Bin and Yuan, Longxiang and Ji, Zhengfeng},
    month     = jul,
    year      = {2025}
}

@inproceedings{Yao1993,
  author    = {Andrew Chi-Chih Yao},
  booktitle = {Proceedings of 1993 IEEE 34th Annual Foundations of Computer Science},
  title     = {Quantum circuit complexity},
  year      = {1993},
  volume    = {},
  number    = {},
  pages     = {352-361},
  doi       = {10.1109/SFCS.1993.366852}
}

@article{Feynman1982,
  author  = {Feynman, R. P.},
  title   = {Simulating Physics with Computers},
  journal = {International Journal of Theoretical Physics},
  volume  = {21},
  number  = {6--7},
  pages   = {467--488},
  year    = {1982},
  doi     = {10.1007/BF02650179}
}

@article{mont,
  doi       = {10.1088/1751-8121/aa565f},
  url       = {https://dx.doi.org/10.1088/1751-8121/aa565f},
  year      = {2017},
  month     = {jan},
  publisher = {IOP Publishing},
  volume    = {50},
  number    = {8},
  pages     = {084002},
  author    = {Montanaro, Ashley},
  title     = {Quantum circuits and low-degree polynomials over {F2}},
  journal   = {Journal of Physics A: Mathematical and Theoretical}
}

@article{Helly1923,
  author  = {Helly, Ed.},
  journal = {Jahresbericht der Deutschen Mathematiker-Vereinigung},
  pages   = {175-176},
  title   = {{\"U}ber Mengen konvexer K{\"o}rper mit gemeinschaftlichen Punkte.},
  url     = {http://eudml.org/doc/145659},
  volume  = {32},
  year    = {1923}
}

@article{ETH,
  title    = {Which Problems Have Strongly Exponential Complexity?},
  journal  = {Journal of Computer and System Sciences},
  volume   = {63},
  number   = {4},
  pages    = {512-530},
  year     = {2001},
  issn     = {0022-0000},
  doi      = {https://doi.org/10.1006/jcss.2001.1774},
  url      = {https://www.sciencedirect.com/science/article/pii/S002200000191774X},
  author   = {Russell Impagliazzo and Ramamohan Paturi and Francis Zane}
}

@article{ETH0,
  title    = {On the Complexity of k-SAT},
  journal  = {Journal of Computer and System Sciences},
  volume   = {62},
  number   = {2},
  pages    = {367-375},
  year     = {2001},
  issn     = {0022-0000},
  doi      = {https://doi.org/10.1006/jcss.2000.1727},
  url      = {https://www.sciencedirect.com/science/article/pii/S0022000000917276},
  author   = {Russell Impagliazzo and Ramamohan Paturi}
}

@article{Hastard,
  author     = {H{\aa}stad, Johan},
  title      = {Some optimal inapproximability results},
  year       = {2001},
  issue_date = {July 2001},
  publisher  = {Association for Computing Machinery},
  address    = {New York, NY, USA},
  volume     = {48},
  number     = {4},
  issn       = {0004-5411},
  url        = {https://doi.org/10.1145/502090.502098},
  doi        = {10.1145/502090.502098},
  journal    = {J. ACM},
  month      = jul,
  pages      = {798–859},
  numpages   = {62}
}

@article{dawson,
  author     = {Dawson, Christopher M. and Hines, Andrew P. and Mortimer, Duncan and Haselgrove, Henry L. and Nielsen, Michael A. and Osborne, Tobias J.},
  title      = {Quantum computing and polynomial equations over the finite field Z2},
  year       = {2005},
  issue_date = {March 2005},
  publisher  = {Rinton Press, Incorporated},
  address    = {Paramus, NJ},
  volume     = {5},
  number     = {2},
  issn       = {1533-7146},
  journal    = {Quantum Info. Comput.},
  month      = mar,
  pages      = {102–112},
  numpages   = {11}
}

@article{Shor1997,
  author   = {Shor, Peter W.},
  title    = {Polynomial-Time Algorithms for Prime Factorization and Discrete Logarithms on a Quantum Computer},
  journal  = {SIAM Journal on Computing},
  volume   = {26},
  number   = {5},
  pages    = {1484-1509},
  year     = {1997},
  doi      = {10.1137/S0097539795293172},
  url      = { 
              
              https://doi.org/10.1137/S0097539795293172
              
              
              
              },
  eprint   = { 
              
              https://doi.org/10.1137/S0097539795293172
              
              
              
              }
}

@inproceedings{Aaronson2011,
  author    = {Aaronson, Scott and Arkhipov, Alex},
  title     = {The computational complexity of linear optics},
  year      = {2011},
  isbn      = {9781450306911},
  publisher = {Association for Computing Machinery},
  address   = {New York, NY, USA},
  url       = {https://doi.org/10.1145/1993636.1993682},
  doi       = {10.1145/1993636.1993682},
  booktitle = {Proceedings of the Forty-Third Annual ACM Symposium on Theory of Computing},
  pages     = {333--342},
  numpages  = {10},
  location  = {San Jose, California, USA},
  series    = {STOC '11}
}

@article{Aaronson2011Permanent,
  author  = {Scott Aaronson},
  title   = {A Linear-Optical Proof that the Permanent is \#P-Hard},
  journal = {arXiv preprint},
  eprint  = {1109.1674},
  year    = {2011}
}

@article{Kuperberg2009,
  author    = {Kuperberg, Greg},
  title     = {How Hard Is It to Approximate the Jones Polynomial?},
  year      = {2015},
  pages     = {183--219},
  doi       = {10.4086/toc.2015.v011a006},
  journal   = {Theory of Computing},
  volume    = {11},
  number    = {6},
  url       = {https://theoryofcomputing.org/articles/v011a006}
}

@article{FujiiMorimae2017,
  doi       = {10.1088/1367-2630/aa5fdb},
  url       = {https://dx.doi.org/10.1088/1367-2630/aa5fdb},
  year      = {2017},
  month     = {mar},
  publisher = {IOP Publishing},
  volume    = {19},
  number    = {3},
  pages     = {033003},
  author    = {Fujii, Keisuke and Morimae, Tomoyuki},
  title     = {Commuting quantum circuits and complexity of Ising partition functions},
  journal   = {New Journal of Physics}
}

@article{Bremner2011,
  author  = {Michael J. Bremner and Richard Jozsa and Dan J. Shepherd},
  title   = {Classical simulation of commuting quantum computations implies collapse of the polynomial hierarchy},
  journal = {Proc. R. Soc. A},
  volume  = {467},
  number  = {2126},
  pages   = {459--472},
  year    = {2011}
}

@article{Shepherd2009,
  author  = {D. Shepherd and M. J. Bremner},
  title   = {Temporally unstructured quantum computation},
  journal = {Proc. Roy. Soc. Ser. A},
  volume  = {465},
  number  = {2105},
  pages   = {1413--1439},
  year    = {2009}
}

@article{MarkovShi,
author = {Markov, Igor L. and Shi, Yaoyun},
title = {Simulating Quantum Computation by Contracting Tensor Networks},
journal = {SIAM Journal on Computing},
volume = {38},
number = {3},
pages = {963-981},
year = {2008},
doi = {10.1137/050644756},
}

@misc{CWD+25,
author = {Bin Cheng and Ziyuan Wang and Ruixuan Deng and Jianxin Chen and Zhengfeng Ji},
title = {Breaking the Treewidth Barrier in Quantum Circuit Simulation with Decision Diagrams},
howpublished = {arXiv:2510.06775},
year = {2025}
}

@article{Adleman1997,
  author   = {Adleman, Leonard M. and DeMarrais, Jonathan and Huang, Ming-Deh A.},
  title    = {Quantum Computability},
  journal  = {SIAM Journal on Computing},
  volume   = {26},
  number   = {5},
  pages    = {1524-1540},
  year     = {1997},
  doi      = {10.1137/S0097539795293639},
  url      = {https://doi.org/10.1137/S0097539795293639}
}

@article{FortnowRogers1999,
  title    = {Complexity Limitations on Quantum Computation},
  journal  = {Journal of Computer and System Sciences},
  volume   = {59},
  number   = {2},
  pages    = {240-252},
  year     = {1999},
  issn     = {0022-0000},
  doi      = {https://doi.org/10.1006/jcss.1999.1651},
  url      = {https://www.sciencedirect.com/science/article/pii/S0022000099916513},
  author   = {Lance Fortnow and John Rogers}
}

@article{Dalzell2020,
  doi       = {10.22331/q-2020-05-11-264},
  url       = {https://doi.org/10.22331/q-2020-05-11-264},
  title     = {How many qubits are needed for quantum computational supremacy?},
  author    = {Dalzell, Alexander M. and Harrow, Aram W. and Koh, Dax Enshan and La Placa, Rolando L.},
  journal   = {{Quantum}},
  issn      = {2521-327X},
  publisher = {{Verein zur F{\"{o}}rderung des Open Access Publizierens in den Quantenwissenschaften}},
  volume    = {4},
  pages     = {264},
  month     = may,
  year      = {2020}
}

@article{Kahn1962,
  author    = {Kahn, A. B.},
  title     = {Topological sorting of large networks},
  journal   = {Communications of the ACM},
  volume    = {5},
  number    = {11},
  pages     = {558--562},
  year      = {1962},
  publisher = {ACM}
}

@article{Robertson1983,
  title   = {Graph minors. I. Excluding a forest},
  author  = {Robertson, Neil and Seymour, Paul D},
  journal = {Journal of Combinatorial Theory, Series B},
  volume  = {35},
  number  = {1},
  pages   = {39--61},
  year    = {1983}
}

@article{Bodlaender1994,
  title   = {A tourist guide through treewidth},
  author  = {Bodlaender, Hans L},
  journal = {Acta cybernetica},
  volume  = {11},
  number  = {1-2},
  pages   = {1--21},
  year    = {1993}
}

@book{Downey2013,
  title     = {Fundamentals of Parameterized Complexity},
  author    = {Downey, Rodney G and Fellows, Michael R},
  volume    = {4},
  year      = {2013},
  publisher = {Springer}
}

@book{Cygan2015,
  author    = {Marek Cygan and
               Fedor V. Fomin and
               Lukasz Kowalik and
               Daniel Lokshtanov and
               D{\'{a}}niel Marx and
               Marcin Pilipczuk and
               Michal Pilipczuk and
               Saket Saurabh},
  title     = {Parameterized Algorithms},
  publisher = {Springer},
  year      = {2015},
  url       = {https://doi.org/10.1007/978-3-319-21275-3},
  doi       = {10.1007/978-3-319-21275-3},
  isbn      = {978-3-319-21274-6},
  bibsource = {dblp computer science bibliography, https://dblp.org}
}

@article{Bodlaender2016,
  author  = {Bodlaender, Hans L. and Drange, P\r{a}l Gr{\o}n\r{a}s and Dregi, Markus S. and Fomin, Fedor V. and Lokshtanov, Daniel and Pilipczuk, Michal},
  title   = {A {$c^k n$} 5-Approximation Algorithm for Treewidth},
  journal = {SIAM Journal on Computing},
  volume  = {45},
  number  = {2},
  pages   = {317--378},
  year    = {2016},
  doi     = {10.1137/130947374},
  url     = {https://doi.org/10.1137/130947374}
}

@inproceedings{Korhonen2021,
  title        = {A Single-Exponential Time 2-Approximation Algorithm for Treewidth},
  author       = {Korhonen, Tuukka},
  booktitle    = {2021 IEEE 62nd Annual Symposium on Foundations of Computer Science (FOCS)},
  pages        = {184--192},
  year         = {2021},
  organization = {IEEE}
}

@book{Kloks1994,
  title     = {Treewidth: computations and approximations},
  author    = {Kloks, Ton},
  volume    = {842},
  year      = {1994},
  publisher = {Springer Science \& Business Media}
}

@article{Bodlaender1996,
  title     = {A linear-time algorithm for finding tree-decompositions of small treewidth},
  author    = {Bodlaender, Hans L.},
  journal   = {SIAM Journal on Computing},
  volume    = {25},
  number    = {6},
  pages     = {1305--1317},
  year      = {1996},
  publisher = {SIAM}
}

@article{Courcelle1990,
  title     = {The monadic second-order logic of graphs. I. Recognizable sets of finite graphs},
  author    = {Courcelle, Bruno},
  journal   = {Information and computation},
  volume    = {85},
  number    = {1},
  pages     = {12--75},
  year      = {1990},
  publisher = {Elsevier}
}

@article{KorhonenLokshtanov2023,
  author  = {Korhonen, Tuukka and Lokshtanov, Daniel},
  title   = {An Improved Parameterized Algorithm for Treewidth},
  journal = {SIAM Journal on Computing},
  volume  = {0},
  number  = {0},
  pages   = {STOC23-155-STOC23-214},
  year    = {2023},
  doi     = {10.1137/23M1595059},
  url     = { 
             
             https://doi.org/10.1137/23M1595059
             
             
             
             },
  eprint  = { 
             
             https://doi.org/10.1137/23M1595059
             
             
             
             }
}

@inproceedings{LokshtanovRamanujanSaurabh2017,
  author    = {Lokshtanov, Daniel and Ramanujan, M. S. and Saurabh, Saket},
  title     = {On Directed Feedback Vertex Set Parameterized by Treewidth},
  booktitle = {Graph-Theoretic Concepts in Computer Science - 44th International Workshop, WG 2018},
  series    = {Lecture Notes in Computer Science},
  volume    = {11159},
  pages     = {385--398},
  year      = {2018},
  publisher = {Springer},
  doi       = {10.1007/978-3-030-00256-5_31},
  url       = {https://doi.org/10.1007/978-3-030-00256-5_31}
}

@article{Dinur2007,
  author    = {Dinur, Irit},
  title     = {The PCP Theorem by Gap Amplification},
  journal   = {Journal of the ACM},
  volume    = {54},
  number    = {3},
  pages     = {12},
  year      = {2007},
  publisher = {ACM},
  doi       = {10.1145/1236457.1236459},
  url       = {https://doi.org/10.1145/1236457.1236459}
}

\end{document}